\documentclass[acmsmall,screen,nonacm]{acmart}
\acmJournal{PACMPL}
\acmVolume{1}
\acmNumber{POPL}
\acmArticle{1}
\acmYear{2026}
\acmMonth{1}

\AtBeginDocument{%
  }

\setcopyright{rightsretained}
\usepackage{amsmath, amsthm, stmaryrd, mathtools, mathrsfs, wasysym}
\usepackage{mathpartir}
\usepackage{hyperref}
\providecommand{\artifactnote}{}
\usepackage{enumitem}
\usepackage[margin=false,inline=true]{fixme}
\usepackage[colorinlistoftodos]{todonotes}

\usepackage{listings}
\usepackage{thmtools, thm-restate}
\usepackage{footnote}
\usepackage{float}
\usepackage{proof}
\usepackage{color}
\usepackage{xstring}
\usepackage{graphicx}
\usepackage{tikz}
\usetikzlibrary{arrows.meta,positioning}
\usepackage{environ}
\usepackage{changepage}
\usepackage{xspace}
\usepackage{relsize}
\usepackage{xparse}
\usepackage{xargs}     
\usepackage{pifont}    
\usepackage[english]{babel}
\usepackage[autostyle, english = american]{csquotes}
\MakeOuterQuote{"}

\ExplSyntaxOn

\keys_define:nn { semrulelet }
  {
    i .tl_set:N = \l_semrulelet_i_tl,
    M .tl_set:N = \l_semrulelet_M_tl,
    V .tl_set:N = \l_semrulelet_V_tl,
    x .tl_set:N = \l_semrulelet_x_tl,
    e .tl_set:N = \l_semrulelet_e_tl,
    vzero .tl_set:N = \l_semrulelet_vzero_tl,
    U .tl_set:N = \l_semrulelet_U_tl,
    p .tl_set:N = \l_semrulelet_p_tl,
    v .tl_set:N = \l_semrulelet_v_tl,
    Uprime .tl_set:N = \l_semrulelet_Uprime_tl,
    L .tl_set:N = \l_semrulelet_L_tl,
  }
\NewDocumentCommand{\semRuleLet}{m}{
  \keys_set:nn { semrulelet } { #1 }
  \ensuremath{
  \Rule*
    { \semDerivExpr{\l_semrulelet_V_tl}{\l_semrulelet_e_tl}{\l_semrulelet_vzero_tl}{\l_semrulelet_U_tl} \\
      \semDerivDeclOne{\l_semrulelet_i_tl}{\l_semrulelet_M_tl}{\l_semrulelet_V_tl, \l_semrulelet_x_tl=\l_semrulelet_vzero_tl}{\l_semrulelet_p_tl}{\l_semrulelet_v_tl}{\l_semrulelet_Uprime_tl}{\l_semrulelet_L_tl}\\
      len~(\l_semrulelet_M_tl) = len(\l_semrulelet_L_tl) }
    { \semDerivDeclOne{\l_semrulelet_i_tl}{\l_semrulelet_M_tl}{\l_semrulelet_V_tl}{\synCompLet{\l_semrulelet_x_tl}{\l_semrulelet_e_tl}~\l_semrulelet_p_tl}{\l_semrulelet_v_tl}{\l_semrulelet_U_tl, \l_semrulelet_Uprime_tl}{\l_semrulelet_L_tl} }
  }
}

\keys_define:nn { semrulestateinit }
  {
    i .tl_set:N = \l_semrulestateinit_i_tl,
    V .tl_set:N = \l_semrulestateinit_V_tl,
    x .tl_set:N = \l_semrulestateinit_x_tl,
    e .tl_set:N = \l_semrulestateinit_e_tl,
    vzero .tl_set:N = \l_semrulestateinit_vzero_tl,
    U .tl_set:N = \l_semrulestateinit_U_tl,
    p .tl_set:N = \l_semrulestateinit_p_tl,
    v .tl_set:N = \l_semrulestateinit_v_tl,
    Uprime .tl_set:N = \l_semrulestateinit_Uprime_tl,
    L .tl_set:N = \l_semrulestateinit_L_tl,
  }
\NewDocumentCommand{\semRuleStateInit}{m}{
  \keys_set:nn { semrulestateinit } { #1 }
  \ensuremath{
  \Rule*
    { \semDerivExpr{\l_semrulestateinit_V_tl}{\l_semrulestateinit_e_tl}{\l_semrulestateinit_vzero_tl}{\l_semrulestateinit_U_tl} \\
      \semDerivDeclOne{\l_semrulestateinit_i_tl+1}{\cdot}{\l_semrulestateinit_V_tl, \l_semrulestateinit_x_tl= \l_semrulestateinit_vzero_tl, set\l_semrulestateinit_x_tl= setter\sb{\l_semrulestateinit_i_tl}}{\l_semrulestateinit_p_tl}{\l_semrulestateinit_v_tl}{\l_semrulestateinit_Uprime_tl}{\l_semrulestateinit_L_tl} \\
      C = \{state\mapsto \l_semrulestateinit_vzero_tl\} }
    { \semDerivDeclOne{\l_semrulestateinit_i_tl}{\cdot}{\l_semrulestateinit_V_tl}{\synState{\l_semrulestateinit_x_tl}{\l_semrulestateinit_e_tl}\l_semrulestateinit_p_tl}{\l_semrulestateinit_v_tl}{\l_semrulestateinit_U_tl, \l_semrulestateinit_Uprime_tl}{C,\l_semrulestateinit_L_tl} }
  }
}

\keys_define:nn { semrulestatererender }
  {
    i .tl_set:N = \l_semrulestatererender_i_tl,
    M .tl_set:N = \l_semrulestatererender_M_tl,
    V .tl_set:N = \l_semrulestatererender_V_tl,
    x .tl_set:N = \l_semrulestatererender_x_tl,
    vzero .tl_set:N = \l_semrulestatererender_vzero_tl,
    Sigma .tl_set:N = \l_semrulestatererender_Sigma_tl,
    p .tl_set:N = \l_semrulestatererender_p_tl,
    v .tl_set:N = \l_semrulestatererender_v_tl,
    U .tl_set:N = \l_semrulestatererender_U_tl,
    L .tl_set:N = \l_semrulestatererender_L_tl,
  }
\NewDocumentCommand{\semRuleStateRerender}{m}{
  \keys_set:nn { semrulestatererender } { #1 }
  \ensuremath{
  \Rule*
    { \semDerivDeclOne{\l_semrulestatererender_i_tl+1}{\l_semrulestatererender_M_tl}{\l_semrulestatererender_V_tl, \l_semrulestatererender_x_tl=\l_semrulestatererender_vzero_tl, setX=setter\sb{\l_semrulestatererender_i_tl}; \l_semrulestatererender_Sigma_tl}{\l_semrulestatererender_p_tl}{\l_semrulestatererender_v_tl}{\l_semrulestatererender_U_tl}{\l_semrulestatererender_L_tl}\\\\
      C = \{state\mapsto \l_semrulestatererender_vzero_tl\}\\
      len(\l_semrulestatererender_M_tl) = len(\l_semrulestatererender_L_tl) }
    { \semDerivDeclOne{\l_semrulestatererender_i_tl}{C,\l_semrulestatererender_M_tl}{\l_semrulestatererender_V_tl; \l_semrulestatererender_Sigma_tl}{\synState{\l_semrulestatererender_x_tl}{e}~\l_semrulestatererender_p_tl}{\l_semrulestatererender_v_tl}{\l_semrulestatererender_U_tl}{C,\l_semrulestatererender_L_tl} }
  }
}

\keys_define:nn { semruleeffectrerenderyes }
  {
    i .tl_set:N = \l_semruleeffectrerenderyes_i_tl,
    M .tl_set:N = \l_semruleeffectrerenderyes_M_tl,
    V .tl_set:N = \l_semruleeffectrerenderyes_V_tl,
    Sigma .tl_set:N = \l_semruleeffectrerenderyes_Sigma_tl,
    p .tl_set:N = \l_semruleeffectrerenderyes_p_tl,
    v .tl_set:N = \l_semruleeffectrerenderyes_v_tl,
    Uprime .tl_set:N = \l_semruleeffectrerenderyes_Uprime_tl,
    L .tl_set:N = \l_semruleeffectrerenderyes_L_tl,
    e .tl_set:N = \l_semruleeffectrerenderyes_e_tl,
    U .tl_set:N = \l_semruleeffectrerenderyes_U_tl,
  }
\NewDocumentCommand{\semRuleEffectRerenderYesChanges}{m}{
  \keys_set:nn { semruleeffectrerenderyes } { #1 }
  \ensuremath{
  \Rule*
    { \semDerivDeclOne{\l_semruleeffectrerenderyes_i_tl+1}{\l_semruleeffectrerenderyes_M_tl}{\l_semruleeffectrerenderyes_V_tl;\; \l_semruleeffectrerenderyes_Sigma_tl}{\l_semruleeffectrerenderyes_p_tl}{\l_semruleeffectrerenderyes_v_tl}{\l_semruleeffectrerenderyes_Uprime_tl}{\l_semruleeffectrerenderyes_L_tl} \\
      \semDerivExpr{\l_semruleeffectrerenderyes_V_tl}{\l_semruleeffectrerenderyes_e_tl}{\l_semruleeffectrerenderyes_v_tl}{\l_semruleeffectrerenderyes_U_tl} \\\\
      \exists i \in [n], v\sb{i} \ne v\sb{i'} \\
      x\sb{1}=v\sb{1'}, x\sb{2}=v\sb{2'},\ldots,x\sb{n}=v\sb{n'} \in \l_semruleeffectrerenderyes_V_tl \\
      len(\l_semruleeffectrerenderyes_M_tl) = len(\l_semruleeffectrerenderyes_L_tl)\\\\
      C = \{values\mapsto [v\sb{1}, v\sb{2}, \ldots, v\sb{n}]\} \\
      C' = \{values\mapsto [v\sb{1'}, v\sb{2'}, \ldots, v\sb{n'}]\} }
    { \semDerivDeclOne{\l_semruleeffectrerenderyes_i_tl}{C,\l_semruleeffectrerenderyes_M_tl}{\l_semruleeffectrerenderyes_V_tl;\; \l_semruleeffectrerenderyes_Sigma_tl}{\synOn{x\sb{1}, x\sb{2},\ldots,x\sb{n}}{\l_semruleeffectrerenderyes_e_tl}~\l_semruleeffectrerenderyes_p_tl}{\l_semruleeffectrerenderyes_v_tl}{\l_semruleeffectrerenderyes_U_tl, \l_semruleeffectrerenderyes_Uprime_tl}{C',\l_semruleeffectrerenderyes_L_tl} }
  }
}

\keys_define:nn { semrulesubcompinit }
  {
    i .tl_set:N = \l_semrulesubcompinit_i_tl,
    Sigma .tl_set:N = \l_semrulesubcompinit_Sigma_tl,
    pF .tl_set:N = \l_semrulesubcompinit_pF_tl,
    vF .tl_set:N = \l_semrulesubcompinit_vF_tl,
    U .tl_set:N = \l_semrulesubcompinit_U_tl,
    LA .tl_set:N = \l_semrulesubcompinit_LA_tl,
    V .tl_set:N = \l_semrulesubcompinit_V_tl,
    y .tl_set:N = \l_semrulesubcompinit_y_tl,
    p .tl_set:N = \l_semrulesubcompinit_p_tl,
    v .tl_set:N = \l_semrulesubcompinit_v_tl,
    Uprime .tl_set:N = \l_semrulesubcompinit_Uprime_tl,
    L .tl_set:N = \l_semrulesubcompinit_L_tl,
    F .tl_set:N = \l_semrulesubcompinit_F_tl,
    tau .tl_set:N = \l_semrulesubcompinit_tau_tl,
  }
\NewDocumentCommand{\semRuleSubcompInit}{m}{
  \keys_set:nn { semrulesubcompinit } { #1 }
  \ensuremath{
  \Rule*
    { \semDerivDeclOne{\l_semrulesubcompinit_i_tl + 1}{\cdot}{x'\sb{1}=v\sb{1}, x'\sb{2}=v\sb{2}\ldots, x'\sb{n}=v\sb{n}; \l_semrulesubcompinit_Sigma_tl}{\l_semrulesubcompinit_pF_tl}{\l_semrulesubcompinit_vF_tl}{\l_semrulesubcompinit_U_tl}{\l_semrulesubcompinit_LA_tl}\\\\
      \semDerivDeclOne{\l_semrulesubcompinit_i_tl+1+len(\l_semrulesubcompinit_LA_tl)}{\cdot}{\l_semrulesubcompinit_V_tl, \l_semrulesubcompinit_y_tl= \l_semrulesubcompinit_vF_tl; \l_semrulesubcompinit_Sigma_tl}{\l_semrulesubcompinit_p_tl}{\l_semrulesubcompinit_v_tl}{\l_semrulesubcompinit_Uprime_tl}{\l_semrulesubcompinit_L_tl} \\\\
      C = \{ memLen = len(\l_semrulesubcompinit_LA_tl) \} \\
      x\sb{1}=v\sb{1}, x\sb{2}=v\sb{2}\ldots, x\sb{n}=v\sb{n} \in \l_semrulesubcompinit_V_tl \\\\
      \synComponent{\l_semrulesubcompinit_F_tl}{x'\sb{1}, x'\sb{2}, \ldots, x'\sb{n}}{\l_semrulesubcompinit_tau_tl}{\l_semrulesubcompinit_pF_tl} \texttt{ exists} }
    { \semDerivDeclOne{\l_semrulesubcompinit_i_tl}{\cdot}{\l_semrulesubcompinit_V_tl; \l_semrulesubcompinit_Sigma_tl}{\synSubComp{\l_semrulesubcompinit_y_tl}{A(x\sb{1}, x\sb{2}, \ldots, x\sb{n})}~\l_semrulesubcompinit_p_tl}{\l_semrulesubcompinit_v_tl}{\l_semrulesubcompinit_U_tl, \l_semrulesubcompinit_Uprime_tl}{C,\l_semrulesubcompinit_LA_tl,\l_semrulesubcompinit_L_tl} }
  }
}

\keys_define:nn { semrulesubcomprerender }
  {
    i .tl_set:N = \l_semrulesubcomprerender_i_tl,
    M .tl_set:N = \l_semrulesubcomprerender_M_tl,
    l .tl_set:N = \l_semrulesubcomprerender_l_tl,
    pF .tl_set:N = \l_semrulesubcomprerender_pF_tl,
    vF .tl_set:N = \l_semrulesubcomprerender_vF_tl,
    U .tl_set:N = \l_semrulesubcomprerender_U_tl,
    LA .tl_set:N = \l_semrulesubcomprerender_LA_tl,
    V .tl_set:N = \l_semrulesubcomprerender_V_tl,
    y .tl_set:N = \l_semrulesubcomprerender_y_tl,
    p .tl_set:N = \l_semrulesubcomprerender_p_tl,
    v .tl_set:N = \l_semrulesubcomprerender_v_tl,
    Uprime .tl_set:N = \l_semrulesubcomprerender_Uprime_tl,
    L .tl_set:N = \l_semrulesubcomprerender_L_tl,
    Sigma .tl_set:N = \l_semrulesubcomprerender_Sigma_tl,
    F .tl_set:N = \l_semrulesubcomprerender_F_tl,
    tau .tl_set:N = \l_semrulesubcomprerender_tau_tl,
  }
\NewDocumentCommand{\semRuleSubcompRerender}{m}{
  \keys_set:nn { semrulesubcomprerender } { #1 }
  \ensuremath{
  \Rule*
    { \semDerivDeclOne{\l_semrulesubcomprerender_i_tl+1}{take(\l_semrulesubcomprerender_M_tl, \l_semrulesubcomprerender_l_tl)}{x'\sb1=v\sb1, x'\sb2=v\sb2\ldots, x'\sb{n}=v\sb{n}}{\l_semrulesubcomprerender_pF_tl}{\l_semrulesubcomprerender_vF_tl}{\l_semrulesubcomprerender_U_tl}{\l_semrulesubcomprerender_LA_tl}\\\\
      \semDerivDeclOne{\l_semrulesubcomprerender_i_tl+1+\l_semrulesubcomprerender_l_tl}{drop(\l_semrulesubcomprerender_M_tl, \l_semrulesubcomprerender_l_tl)}{\l_semrulesubcomprerender_V_tl, \l_semrulesubcomprerender_y_tl= \l_semrulesubcomprerender_vF_tl}{\l_semrulesubcomprerender_p_tl}{\l_semrulesubcomprerender_v_tl}{\l_semrulesubcomprerender_Uprime_tl}{\l_semrulesubcomprerender_L_tl} \\\\
      C = \{ memLen = \l_semrulesubcomprerender_l_tl \} \\
      x\sb1=v\sb1, x\sb2=v\sb2\ldots, x\sb{n}=v\sb{n} \in \l_semrulesubcomprerender_V_tl \\\\
      \synComponent{\l_semrulesubcomprerender_F_tl}{x'\sb1, x'\sb2, \ldots, x'\sb{n}}{\l_semrulesubcomprerender_tau_tl}{\l_semrulesubcomprerender_pF_tl} \in \l_semrulesubcomprerender_Sigma_tl }
    { \semDerivDeclOne{\l_semrulesubcomprerender_i_tl}{C,\l_semrulesubcomprerender_M_tl}{\l_semrulesubcomprerender_V_tl; \l_semrulesubcomprerender_Sigma_tl}{\synSubComp{\l_semrulesubcomprerender_y_tl}{A(x\sb1, x\sb2, \ldots, x\sb{n})}~\l_semrulesubcomprerender_p_tl}{\l_semrulesubcomprerender_v_tl}{\l_semrulesubcomprerender_U_tl, \l_semrulesubcomprerender_Uprime_tl}{C',\l_semrulesubcomprerender_LA_tl,\l_semrulesubcomprerender_L_tl} }
  }
}

\keys_define:nn { semrulefunctionapp }
  {
    V .tl_set:N = \l_semrulefunctionapp_V_tl,
    eone .tl_set:N = \l_semrulefunctionapp_eone_tl,
    x .tl_set:N = \l_semrulefunctionapp_x_tl,
    ethree .tl_set:N = \l_semrulefunctionapp_ethree_tl,
    Vprime .tl_set:N = \l_semrulefunctionapp_Vprime_tl,
    Uone .tl_set:N = \l_semrulefunctionapp_Uone_tl,
    etwo .tl_set:N = \l_semrulefunctionapp_etwo_tl,
    vtwo .tl_set:N = \l_semrulefunctionapp_vtwo_tl,
    Utwo .tl_set:N = \l_semrulefunctionapp_Utwo_tl,
    vthree .tl_set:N = \l_semrulefunctionapp_vthree_tl,
    Uthree .tl_set:N = \l_semrulefunctionapp_Uthree_tl,
  }
\NewDocumentCommand{\semRuleFunctionApplication}{m}{
  \keys_set:nn { semrulefunctionapp } { #1 }
  \ensuremath{
  \Rule*
    { \semDerivExpr{\l_semrulefunctionapp_V_tl}{\l_semrulefunctionapp_eone_tl}{\fnClosure{\l_semrulefunctionapp_x_tl}{\l_semrulefunctionapp_ethree_tl}{\l_semrulefunctionapp_Vprime_tl}}{\l_semrulefunctionapp_Uone_tl} \\
      \semDerivExpr{\l_semrulefunctionapp_V_tl}{\l_semrulefunctionapp_etwo_tl}{\l_semrulefunctionapp_vtwo_tl}{\l_semrulefunctionapp_Utwo_tl} \\
      \semDerivExpr{\l_semrulefunctionapp_Vprime_tl, \l_semrulefunctionapp_x_tl{=}~\l_semrulefunctionapp_vtwo_tl}{\l_semrulefunctionapp_ethree_tl}{\l_semrulefunctionapp_vthree_tl}{\l_semrulefunctionapp_Uthree_tl} }
    { \semDerivExpr{\l_semrulefunctionapp_V_tl}{\l_semrulefunctionapp_eone_tl\;\l_semrulefunctionapp_etwo_tl}{\l_semrulefunctionapp_vthree_tl}{\l_semrulefunctionapp_Uone_tl, \l_semrulefunctionapp_Utwo_tl, \l_semrulefunctionapp_Uthree_tl} }
  }
}

\keys_define:nn { confruletopstep }
  {
    M .tl_set:N = \l_confruletopstep_M_tl,
    Sigmaprime .tl_set:N = \l_confruletopstep_Sigmaprime_tl,
    C .tl_set:N = \l_confruletopstep_C_tl,
    v .tl_set:N = \l_confruletopstep_v_tl,
    taur .tl_set:N = \l_confruletopstep_taur_tl,
    p .tl_set:N = \l_confruletopstep_p_tl,
    vr .tl_set:N = \l_confruletopstep_vr_tl,
    U .tl_set:N = \l_confruletopstep_U_tl,
    L .tl_set:N = \l_confruletopstep_L_tl,
    A .tl_set:N = \l_confruletopstep_A_tl,
    tau .tl_set:N = \l_confruletopstep_tau_tl,
    Delta .tl_set:N = \l_confruletopstep_Delta_tl,
    Sigma .tl_set:N = \l_confruletopstep_Sigma_tl,
  }
\NewDocumentCommand{\confRuleTopStep}{m}{
  \keys_set:nn { confruletopstep } { #1 }
  \ensuremath{
  \Rule*
    { \semDerivTopStepOne{\l_confruletopstep_M_tl}{\l_confruletopstep_Sigmaprime_tl}{\synComponent{\l_confruletopstep_C_tl}{\overline{x = \l_confruletopstep_v_tl}}{\l_confruletopstep_taur_tl}{\l_confruletopstep_p_tl}}{phase}{\l_confruletopstep_vr_tl, \l_confruletopstep_U_tl, \l_confruletopstep_L_tl}
      \\\\
      \l_confruletopstep_Sigmaprime_tl
        \Rightarrow \synComponent{\l_confruletopstep_A_tl}{\overline{x\sb{i}{:}~\tau\sb{i}}}{\l_confruletopstep_tau_tl}{p\sb{1},p\sb{2},\ldots, p\sb{n}}
        \Rightarrow \l_confruletopstep_Delta_tl
      \\\\
      len(\l_confruletopstep_M_tl) = len(\l_confruletopstep_L_tl) \\
      \overline{\tyDerivExpr{\cdot}{\l_confruletopstep_v_tl}{\tau}{\cdot}} \\
      \tyDerivExpr{\cdot}{\l_confruletopstep_vr_tl}{\l_confruletopstep_taur_tl}{\cdot}
      \\\\
      \l_confruletopstep_Sigma_tl = \l_confruletopstep_Sigmaprime_tl, \l_confruletopstep_C_tl{:}~(\synComponent{\l_confruletopstep_C_tl}{\overline{x = \tau}}{\l_confruletopstep_taur_tl}{\l_confruletopstep_p_tl}, \l_confruletopstep_Delta_tl) }
    { \confDerivTopStep{\l_confruletopstep_M_tl}{\l_confruletopstep_Sigma_tl}{\synComponent{\l_confruletopstep_C_tl}{\overline{x = \l_confruletopstep_v_tl}}{\l_confruletopstep_taur_tl}{\l_confruletopstep_p_tl}}{phase}{\l_confruletopstep_vr_tl, \l_confruletopstep_U_tl, \l_confruletopstep_L_tl} }
  }
}

\keys_define:nn { confruledecl }
  {
    i .tl_set:N = \l_confruledecl_i_tl,
    M .tl_set:N = \l_confruledecl_M_tl,
    V .tl_set:N = \l_confruledecl_V_tl,
    p .tl_set:N = \l_confruledecl_p_tl,
    v .tl_set:N = \l_confruledecl_v_tl,
    U .tl_set:N = \l_confruledecl_U_tl,
    L .tl_set:N = \l_confruledecl_L_tl,
    Sigma .tl_set:N = \l_confruledecl_Sigma_tl,
    Gamma .tl_set:N = \l_confruledecl_Gamma_tl,
    Delta .tl_set:N = \l_confruledecl_Delta_tl,
    tau .tl_set:N = \l_confruledecl_tau_tl,
  }
\NewDocumentCommand{\confRuleDecl}{m}{
  \keys_set:nn { confruledecl } { #1 }
  \ensuremath{
  \Rule*
    { \semDerivDeclOne{\l_confruledecl_i_tl}{\l_confruledecl_M_tl}{\l_confruledecl_V_tl}{\l_confruledecl_p_tl}{\l_confruledecl_v_tl}{\l_confruledecl_U_tl}{\l_confruledecl_L_tl} \\\\
      \tyDerivDecls{\l_confruledecl_Sigma_tl}{\l_confruledecl_Gamma_tl}{\l_confruledecl_Delta_tl}{\l_confruledecl_p_tl}{\l_confruledecl_tau_tl}\\\\
      \l_confruledecl_Gamma_tl \Vdash \l_confruledecl_V_tl \\
      \tyDerivExpr{\cdot}{\l_confruledecl_v_tl}{\l_confruledecl_tau_tl}{\cdot} }
    { \confDerivDecl{\l_confruledecl_i_tl}{\l_confruledecl_M_tl}{\l_confruledecl_V_tl}{\l_confruledecl_p_tl}{\l_confruledecl_v_tl}{\l_confruledecl_U_tl}{\l_confruledecl_L_tl} }
  }
}

\ExplSyntaxOff

\theoremstyle{plain}
\newtheorem{theorem}{Theorem}[section]

\theoremstyle{definition}

\theoremstyle{remark}

\NewDocumentCommand{\rulename}{m}{\textsc{#1}}

\makeatletter
\NewDocumentCommand{\Rule}{}{\@ifstar\Ruleshort\Rulelong}
\makeatother
\NewDocumentCommand{\Ruleshort}{m m}{\ensuremath{\inferrule{#1}{#2}}}
\NewDocumentCommand{\Rulelong}{m m m}{\ensuremath{\inferrule[(#1)]{#2}{#3}}}
\NewDocumentCommand{\breathe}{}{\vspace{30pt}}

\newenvironment{placedfigure}
  {\par\addvspace{\bigskipamount}\begin{center}}
  {\end{center}\par\addvspace{\bigskipamount}}
\NewDocumentCommand{\listenersafe}{m}{\ensuremath{\vdash #1 \;\mathsf{lsafe}}}

\NewDocumentCommand{\m}{m}{\ensuremath{\mathsf{#1}}}
\NewDocumentCommand{\mi}{m}{\ensuremath{\mbox{\it #1}}}

\NewDocumentCommand{\df}{}{\textit{df}\xspace}

\NewDocumentCommand{\kw}{m}{\ensuremath{{\texttt{#1}}}} 

\NewDocumentCommand{\Next}{m}{\ensuremath{{\bigcirc}^{#1}}}
\NewDocumentCommand{\always}{m}{\ensuremath{\Box^{#1}}}
\NewDocumentCommand{\eventually}{m}{\ensuremath{\Diamond^{#1}}}
\NewDocumentCommand{\cancelEv}{m}{\ensuremath{\oslash\,#1}}
\NewDocumentCommand{\remove}{m}{\ensuremath{\text{\ding{55}}\,#1}}
\NewDocumentCommand{\eventEff}{m m}{\ensuremath{\kw{#1}\langle \kw{#2} \rangle}}
\NewDocumentCommand{\eventEffBar}{m m}{\ensuremath{#1\langle \overline{#2} \rangle}}

\NewDocumentCommand{\openbrace}{}{\texttt{\{}}
\NewDocumentCommand{\closebrace}{}{\texttt{\}}}

\NewDocumentCommand{\openparen}{}{\texttt{(}}
\NewDocumentCommand{\closeparen}{}{\texttt{)}}

\NewEnviron{block}{
  \noindent\openbrace
  \vspace{0em}
  \begin{adjustwidth}{2em}{0pt}
    \BODY
  \end{adjustwidth}
  \vspace{0em}
  \noindent\closebrace
  \par\noindent
}

\NewEnviron{parenblock}{
  \noindent\openparen
  \vspace{0em}
  \begin{adjustwidth}{2em}{0pt}
    \BODY
  \end{adjustwidth}
  \vspace{0em}
  \noindent\closeparen
  \par\noindent
}

\NewDocumentCommand{\react}{}{React\xspace}
\NewDocumentCommand{\lang}{}{Willow\xspace}
\NewDocumentCommand{\js}{}{JavaScript\xspace}

\NewDocumentCommand{\html}{}{HTML\xspace}

\NewDocumentCommand{\api}{}{API\xspace}

\NewDocumentCommand{\tinyComp}{m}{\ensuremath{\texttt{comp}~\mathit{#1}}}

\NewDocumentCommand{\synIf}{s m m m}{
  \IfBooleanTF{#1}{
    \noindent\texttt{if}\;\ensuremath{#2}\;\texttt{then}%
    \begin{block}
      \ensuremath{#3}
    \end{block}%
    \texttt{else}%
    \begin{block}
      \ensuremath{#4}
    \end{block}%
  }{
    \ensuremath{
      \texttt{if}\;#2\;\texttt{then}\;#3\;\texttt{else}\;#4
    }
  }
}

\NewDocumentCommand{\synState}{s m m}{
  \IfBooleanTF{#1}{
    \ensuremath{
      \texttt{state}\;\textit{#2},\; \textit{set{\MakeUppercase #2}}\;\texttt{default}\;#3;\;
    }\newline
  }{
    \ensuremath{
      \texttt{state}\;\textit{#2},\; \textit{set{\MakeUppercase #2}}\;\texttt{default}\;#3;\;
    }
  }
}

\NewDocumentCommand{\synSubComp}{s m m}{
  \IfBooleanTF{#1}{
    \ensuremath{
      \texttt{comp}\;#2 = #3;\;
    }\newline
  }{
    \ensuremath{
      \texttt{comp}\;#2 = #3;\;
    }
  }
}

\NewDocumentCommand{\synCompLet}{s m m}{
  \IfBooleanTF{#1}{
    \ensuremath{
      \texttt{let}\;\mathit{#2} = #3;\;
    }\newline
  }{
    \ensuremath{
      \texttt{let}\;\mathit{#2} = #3;\;
    }
  }
}

\NewDocumentCommand{\synOn}{s m m}{
  \IfBooleanTF{#1}{
    \noindent\texttt{on}\;\ensuremath{#2}\;\texttt{do}%
    \begin{block}
      #3
    \end{block}%
  }{
    \ensuremath{
      \texttt{on}\;#2\; \texttt{do}\;\{\;#3\;\};\;
    }
  }
}

\NewDocumentCommand{\synReturn}{s m}{
  \IfBooleanTF{#1}{
    \noindent\texttt{return}%
    \begin{parenblock}
      \ensuremath{#2}
    \end{parenblock}%
  }{
    \ensuremath{
      \texttt{return}\;#2
    }
  }
}

\lstdefinestyle{htmlstyle}{
  language=HTML,
  basicstyle=\ttfamily\small,
  showstringspaces=false,
  columns=fullflexible,
}

\lstnewenvironment{returnblock}[1][]{
  \noindent\texttt{return}\openbrace\\
  \lstset{style=htmlstyle,#1}
}{
  \closebrace 
}

\NewDocumentCommand{\synComponent}{s m m m m}{
  \IfBooleanTF{#1}{
    \noindent\begin{minipage}{\textwidth}%
      \noindent\texttt{comp}~\ensuremath{\mathit{#2}}~(\ensuremath{#3}):#4%
      \begin{block}
        #5
      \end{block}%
    \end{minipage}
  }{
    \ensuremath{
      \texttt{comp}~\mathit{#2}~(#3):#4~\{~#5~\}
    }
  }
}

\NewDocumentCommand{\tyBool}{}{\texttt{bool}}
\NewDocumentCommand{\tyUnit}{}{\texttt{unit}}

\NewDocumentCommand{\tyArrow}{s m m o}{
  \IfBooleanTF{#1}{
    \ensuremath{
      #2 \to #3
    }
  }{
    \ensuremath{
      #2 \to #3 \mid #4
    }
  }
}

\NewDocumentCommand{\tySetter}{s O{\tau} m}{
  \IfBooleanTF{#1}{
    \ensuremath{
      \tyArrow
        {(\tyArrow*{#2}{#2})}
        {\kw{unit}}
        [\Next{1r}{\stch{#3}}]
    }
  }{
    \ensuremath{
      \tyArrow
        {(\tyArrow{#2}{#2}[\cdot])}
        {\kw{unit}}
        [\Next{1r}{\stch{#3}}]
    }
  }
}

\NewDocumentCommand{\stch}{m}{\ensuremath{
  \raisebox{0.15ex}{\tiny @}\hspace{-0em}#1
}}

\NewDocumentCommand{\synFn}{m m}{\ensuremath{
  \lambda #1.#2
}}

\NewDocumentCommand{\semDerivDeclOne}{m m m m m m m}{\ensuremath{
  [#1 \mid #2 ]\;;\; #3 \vdash #4 \Downarrow #5 \;;\; #6 \;;\; [#7]
}}

\NewDocumentCommand{\semDerivDecl}{m m m m m m m m m}{\ensuremath{
  #1 \;;\; [#2 \mid #3 ]\;;\; #4 \;;\; #5 \vdash #6 \Downarrow #7 \;;\; #8 \;;\; #9
}}

\NewDocumentCommand{\semPfDerivDecl}{m m m m m m m m m}{\ensuremath{
  #1 \;;\; #2 \;;\; #3 \;;\; #4 \;;\; #5 \vdash_\Sigma #6 \Downarrow #7 \;;\; #8 \;;\; #9
}}

\NewDocumentCommand{\semPfTraceArgs}{m m m}{\ensuremath{
  \mathit{ARG}\; #1 = #2 \to #3
}}
\NewDocumentCommand{\semPfTraceExt}{m}{\ensuremath{
  \mathit{EXT}\; #1
}}
\NewDocumentCommand{\semPfTraceFire}{m}{\ensuremath{
  \mathit{FIRE}\; #1
}}
\NewDocumentCommand{\semPfTraceFireCxl}{m}{\ensuremath{
  \mathit{FIRE\text{-}CXL}\; #1
}}
\NewDocumentCommand{\semPfTraceFireSuc}{m m}{\ensuremath{
  \mathit{FIRE\text{-}SUC}\; #1 \to \{\;#2\;\}
}}
\NewDocumentCommand{\semPfTraceOn}{m m}{\ensuremath{
  \mathit{EFF}\;(#1) \to \{\;#2\;\}
}}
\NewDocumentCommand{\semPfTraceLet}{m m}{\ensuremath{
  \mathit{LET}\;(#1) \leftarrow \{\;#2\;\}
}}
\NewDocumentCommand{\semPfTraceReturn}{m}{\ensuremath{
  \mathit{RET}\;#1
}}

\NewDocumentCommand{\semDerivExpr}{O{\Sigma} m m m m}{\ensuremath{
  #1 \;;\; #2 \vdash #3 \Downarrow #4 \;;\; #5
}}

\NewDocumentCommand{\confTop}{O{\Sigma} m m m m m m m m}{\ensuremath{%
  \langle#1 \mid #2 \mid #3\mid #4\mid #5\mid #6 \mid #7 \mid \texttt{#8} \mid #9\rangle
}}

\NewDocumentCommand{\confPfTop}{m m m m m m}{\ensuremath{%
  \langle\Sigma; \Delta; \Gamma_s \mid\mid #1 \mid \mathit{#2} \mid #3%
  \mid\texttt{#4} \mid #5 \mid #6 \rangle
}}

\NewDocumentCommand{\confDerivTopStep}{m m m m m}{\ensuremath{
  #1 \triangleright #2 \Vdash #3 \texttt{ #4 } \triangleright #5
}}

\NewDocumentCommand{\confDerivDecl}{m m m m m m m}{\ensuremath{
  [#1 \mid #2 ]\;;\; #3 \Vdash #4 \Downarrow #5 \;;\; #6 \;;\; [#7]
}}

\NewDocumentCommand{\causesderiv}{m m m}{\ensuremath{
  #1 \vdash \stch{#2} \Rightarrow #3
}}

\NewDocumentCommand{\valTyDeriv}{O{\Sigma} m m}{\ensuremath{
  #1 \models #2 : #3
}}

\NewDocumentCommand{\valTyEnvDeriv}{O{\Sigma} m m}{\ensuremath{
  #1 \;;\; #2 \models_t #3
}}

\NewDocumentCommand{\effQueueDeriv}{m m}{\ensuremath{
  #1 \models_e #2
}}

\NewDocumentCommand{\justifies}{m m}{\ensuremath{
  (#1) \succ #2
}}
\NewDocumentCommand{\justifiesVarSet}{m m m}{\ensuremath{
  (#1) \succ #2 \succ #3
}}

\NewDocumentCommand{\tyDerivComp}{m m m}{\ensuremath{
  #1 \Rightarrow #2 \Rightarrow #3
}}

\NewDocumentCommand{\tyDerivDecl}{m m m m m}{\ensuremath{
  #1\;;\;#2\;;\;#3\vdash #4 \Rightarrow #5
}}

\NewDocumentCommand{\tyDerivDecls}{m m m m m}{\ensuremath{
  #1\;;\;#2\;;\;#3\vdash #4 : #5
}}

\NewDocumentCommand{\tyDerivExpr}{m m m m}{\ensuremath{
  #1 \vdash #2 : #3 \shortmid #4
}}

\NewDocumentCommand{\setterderiv}{m m m}{\ensuremath{
  #1 \Rightarrow #2 \Rightarrow #3
}}

\NewDocumentCommand{\setterQueueDeriv}{m m m}{\ensuremath{
  #1 \Rightarrow #2 \Rightarrow #3
}}

\NewDocumentCommand{\semDerivTopStepOne}{m m m m m}{\ensuremath{
  #1 \triangleright #2 \vdash #3 \texttt{ #4 } \triangleright #5
}}

\NewDocumentCommand{\semDerivTopStep}{m m m m m m m}{\ensuremath{%
  \Sigma \;;\; #1 \;;\; #2 \;;\; #3 \;;\; #4 \vdash #5 \;\texttt{#6}\; #7
}}

\NewDocumentCommand{\fnClosure}{m m m}{\ensuremath{
  \llparenthesis \lambda #1. #2, #3 \rrparenthesis
}}
\NewDocumentCommand{\Closure}{m m m}{\ensuremath{
  \llparenthesis \lambda #1. #2, #3 \rrparenthesis
}}

\NewDocumentCommand{\deltaEntry}{m m m}{#1\left[#2\right] \shortmid #3}

\NewDocumentCommand{\setter}{m}{\textit{set{\MakeUppercase #1}}}

\NewDocumentCommand{\tyRuleReturnDecl}{mmmmm}{
  \ensuremath{
  \Rule*
    { \causesderiv{#3}{#4}{\stch return} }
    { \tyDerivDecls{#1}{#2, #4: #5}{#3}{\synReturn{#4}}{#5} }
  }
}

\NewDocumentCommand{\tyRuleStateDecl}{mmmmmmmm}{
  \ensuremath{
  \Rule*
    { \tyDerivExpr{#2}{#5}{#6}{\cdot}\\
      \deltaEntry{#4}{}{#7} \in #3\\
      #2' = #2, #4:#6, \setter{#4}: (\tySetter[#6]{#4})}
    { \tyDerivDecl{#1}{#2}{#3}{\synState{#4}{#5} #8}{#2'} }
  }
}

\NewDocumentCommand{\tyRuleOnDecl}{mmmmm}{
  \ensuremath{
  \Rule*
    { \tyDerivExpr{#2}{#4}{\tyUnit}{#5} \\
      \forall i \in [n], \causesderiv{#3}{x_i}{#5} }
    { \tyDerivDecl{#1}{#2}{#3}{\synOn{\overline{x}}{#4}}{#2} }
  }
}

\NewDocumentCommand{\tyRuleBool}{mm}{
  \ensuremath{
  \Rule*
    { #2 \in \{true, false\} }
    { \tyDerivExpr{#1}{#2}{\tyBool}{\cdot} }
  }
}
\NewDocumentCommand{\tyRuleBoolDefault}{}{\tyRuleBool{\Gamma}{c}}

\NewDocumentCommand{\tyRuleBase}{mmm}{
  \ensuremath{
  \Rule*
    { #2 \in #3 }
    { \tyDerivExpr{#1}{#2}{#3}{\cdot} }
  }
}
\NewDocumentCommand{\tyRuleBaseDefault}{}{\tyRuleBase{\Gamma}{c}{\alpha}}

\NewDocumentCommand{\tyRuleVar}{mmm}{
  \ensuremath{
  \Rule*
    { }
    { \tyDerivExpr{#1, #2: #3}{#2}{#3}{\cdot} }
  }
}
\NewDocumentCommand{\tyRuleVarDefault}{}{\tyRuleVar{\Gamma}{x}{\tau}}

\NewDocumentCommand{\tyRuleFnApp}{mmmmmmmm}{
  \ensuremath{
  \Rule*
    { \tyDerivExpr{#1}{#2}{(\tyArrow{#4}{#5}[#8])}{#6}\\
      \tyDerivExpr{#1}{#3}{#4}{#7} }
    { \tyDerivExpr{#1}{#2\;#3}{#5}{#6 * #7 * #8} }
  }
}
\NewDocumentCommand{\tyRuleFnAppDefault}{}{\tyRuleFnApp{\Gamma}{e_1}{e_2}{\tau_1}{\tau_2}{F_1}{F_2}{F}}

\NewDocumentCommand{\tyRuleFn}{mmmmmm}{
  \ensuremath{
  \Rule*
    { \tyDerivExpr{#1, #2 :#3}{#4}{#5}{#6} }
    { \tyDerivExpr{#1}{\synFn{#2}{#4}}{(\tyArrow{#3}{#5}[#6])}{\cdot} }
  }
}
\NewDocumentCommand{\tyRuleFnDefault}{}{\tyRuleFn{\Gamma}{x}{\tau_1}{e}{\tau_2}{F}}

\NewDocumentCommand{\tyRuleSeq}{mmmmmmm}{
  \ensuremath{
  \Rule*
    { \tyDerivExpr{#1}{#2}{#3}{#4} \\
      \tyDerivExpr{#1}{#5}{#6}{#7} }
    { \tyDerivExpr{#1}{#2;#5}{#6}{#4 * #7} }
  }
}
\NewDocumentCommand{\tyRuleSeqDefault}{}{\tyRuleSeq{\Gamma}{e_1}{\tau_1}{F_1}{e_2}{\tau_2}{F_2}}

\NewDocumentCommand{\tyRuleBranch}{mmmmmmmm}{
  \ensuremath{
  \Rule*
    { \tyDerivExpr{#1}{#2}{\tyBool}{#3} \\
      \tyDerivExpr{#1}{#4}{#5}{#6}\\
      \tyDerivExpr{#1}{#7}{#5}{#8} }
    { \tyDerivExpr{#1}{\synIf{#2}{#4}{#7}}{#5}{#3 * (#6 + #8)} }
  }
}
\NewDocumentCommand{\tyRuleBranchDefault}{}{\tyRuleBranch{\Gamma}{e_1}{F_1}{e_2}{\tau}{F_2}{e_3}{F_3}}

\NewDocumentCommand{\tyRuleProd}{mmmmmmm}{
  \ensuremath{
  \Rule*
    { \tyDerivExpr{#1}{#2}{#3}{#4} \\
      \tyDerivExpr{#1}{#5}{#6}{#7} }
    { \tyDerivExpr{#1}{(#2, #5)}{#3 \times #6}{#4 * #7} }
  }
}
\NewDocumentCommand{\tyRuleProdDefault}{}{\tyRuleProd{\Gamma}{e_1}{\tau_1}{F_1}{e_2}{\tau_2}{F_2}}

\NewDocumentCommand{\tyRuleFst}{mmmmm}{
  \ensuremath{
  \Rule*
    { \tyDerivExpr{#1}{#2}{#3 \times #4}{#5} }
    { \tyDerivExpr{#1}{\kw{fst}\;#2}{#3}{#5} }
  }
}
\NewDocumentCommand{\tyRuleFstDefault}{}{\tyRuleFst{\Gamma}{e_1}{\tau_1}{\tau_2}{F}}

\NewDocumentCommand{\semRuleVar}{mmm}{
  \ensuremath{
  \Rule*
    { }
    { \semDerivExpr{#1,#2=#3}{#2}{#3}{\cdot} }
  }
}
\NewDocumentCommand{\semRuleVarDefault}{}{\semRuleVar{V}{x}{v}}

\NewDocumentCommand{\semRuleConstant}{mm}{
  \ensuremath{
  \Rule*
    { #2 \in \alpha \cup \{true, false\} \cup \{()\} }
    { \semDerivExpr{#1}{#2}{#2}{\cdot} }
  }
}
\NewDocumentCommand{\semRuleConstantDefault}{}{\semRuleConstant{V}{c}}

\NewDocumentCommand{\semRuleSetterCall}{mmmmmmmmm}{
  \ensuremath{
  \Rule*
    { \semDerivExpr{#1}{#2}{setter_{#3}}{#4} \\
      \semDerivExpr{#1}{#5}{\Closure{#6}{#7}{#8}}{#9} }
    { \semDerivExpr{#1}{#2\;#5}{()}{#4, #9, setter_{#3}(\Closure{#6}{#7}{#8})} }
  }
}
\NewDocumentCommand{\semRuleSetterCallDefault}{}{\semRuleSetterCall{V}{e_1}{x}{U_1}{e_2}{x}{e_3}{V'}{U_2}}

\NewDocumentCommand{\semRuleClosureCreate}{mmm}{
  \ensuremath{
  \Rule*
    { }
    { \semDerivExpr{#1}{\synFn{#2}{#3}}{\Closure{#2}{#3}{#1}}{\cdot} }
  }
}

\NewDocumentCommand{\semRuleSeq}{mmmmmmm}{
  \ensuremath{
  \Rule*
    { \semDerivExpr{#1}{#2}{#3}{#4} \\
      \semDerivExpr{#1}{#5}{#6}{#7} }
    { \semDerivExpr{#1}{#2\;;\;#5}{#6}{#4, #7} }
  }
}
\NewDocumentCommand{\semRuleSeqDefault}{}{\semRuleSeq{V}{e_1}{v_1}{U_1}{e_2}{v_2}{U_2}}

\renewcommand{\artifactnote}{%
  \footnote{The prototype checker is available at
    \url{https://github.com/junewunder/willow-inference-hs}.}}

\begin{document}

\title{A Type-and-Effect System for Temporal Dependency Analysis of Render-based Reactive Programs}

\author{june wunder}
\orcid{0000-0002-3280-9731}
\affiliation{%
  \institution{Boston University}
  \city{Boston}
  \state{Massachusetts}
  \country{USA}
}
\email{june@junewunder.com}

\author{Ankush Das}
\orcid{0000-0003-2459-1258}
\affiliation{%
  \institution{Boston University}
  \city{Boston}
  \state{Massachusetts}
  \country{USA}
}
\email{ankushd@bu.edu}

\author{Marco Gaboardi}
\orcid{0000-0002-5235-7066}
\affiliation{%
  \institution{Boston University}
  \city{Boston}
  \state{Massachusetts}
  \country{USA}
}
\email{gaboardi@bu.edu}

\renewcommand{\shortauthors}{wunder, Das, and Gaboardi}

\begin{abstract}
Reactive programming frameworks such as React allow developers to build
interactive applications by declaratively specifying how outputs depend on
changing inputs.
Although this model makes it easy to reason about what an application computes,
the temporal behavior of reactive programs—when updates occur and how they
propagate—remains difficult to understand and verify.
Applications implicitly rely on timing assumptions buried in framework runtimes,
leading to subtle bugs such as stale reads, transient inconsistencies,
order-dependent behavior, and unintended feedback cycles.
To address these challenges, this paper presents Willow, a core calculus for reactive programming inspired by
React.
Willow gives a time-aware operational semantics that models computation in terms
of renders, the fundamental evaluation step in which components produce user
interface descriptions, and pairs it with a novel type-and-effect system that
statically tracks timing behavior as effects.
A next modality $\Next{}{}$
expresses delays measured not only in renders but in any unit the host
environment exposes---renders, network round-trips, or wall-clock
milliseconds.
A family of modalities tracks the full lifecycle of
event handlers: when they are registered, when they fire, when pending events
are canceled and when handlers are removed.
A key insight is that the
resulting effects form a temporal dependency graph, letting standard
graph algorithms statically detect render cascades and inter-render loops that
cause non-termination or performance degradation.
We formalize Willow and prove preservation of the effect
system with respect to the time-aware semantics.
We also implement a
prototype checker with automatic effect inference and evaluate it on
representative reactive patterns such as debouncing, form inputs, and
API-driven updates.
Our results demonstrate that time-aware typing provides a
practical foundation for reasoning about the temporal correctness of reactive
programs.
\end{abstract}

\maketitle

\section{Introduction}
\label{sec:intro}

Modern interactive software is increasingly structured as
reactive programs~\cite{Berry92SCP,Halbwachs91IEEE,Wan00PLDI}
that continuously respond to streams of events originating from users, sensors, and network services.
Frameworks such as React~\cite{react-fiber,ReactMeta} have popularized a declarative model for
building such systems, in which programs describe how outputs depend on changing inputs
rather than explicitly orchestrating control flow.
This programming style has proven highly effective for building complex interactive applications,
particularly in web and mobile
environments~\cite{Meyerovich09OOPSLA,elliottFunctionalReactiveAnimation1997}.

Despite its widespread adoption, reasoning about the temporal behavior of
reactive programs remains challenging.
Reactive applications implicitly encode timing assumptions about when updates occur,
how quickly state propagates, and how computations interact with asynchronous events.
Timing dependencies can be intentionally or accidentally \emph{mutually
recursive}, i.e., an update to variable $x$ leads to an update to $y$ which, in
turn, leads to an update to $x$.
State changes or other side effects are often conditional on other pieces of
program state, which makes programs more difficult to reason about.
Due to these complexities, timing behaviors are embedded within framework
semantics and runtime scheduling policies.
This leaves timing reasoning to be done through informal means,
developers lack principled tools for predicting or verifying the
temporal properties of their programs.

To make matters worse, reactive programs are particularly susceptible to bugs
that arise from a lack of understanding of precise timing behavior.
%
For instance, computations may observe stale state~\cite{Cooper06ESOP} if updates have not yet
propagated, or transient inconsistencies (glitches) when dependent values are updated
at different times.
Programs may also exhibit order-dependence~\cite{Kahn74IFIP}, where the outcome of a computation
depends on the scheduler's choice of evaluation order, or delayed-effect violations,
where effects occur later than the program assumes due to deferred execution or batching.
Graphical user interfaces have long been recognized as difficult to
test~\cite{memon_gui_2002,banerjee_graphical_2013}, owing to the combinatorial
space of possible event sequences and orderings.
Timing-dependent behaviors are among the hardest cases because they emerge from
the interaction between the program's dependency structure and the framework's
execution model.

A second source of temporal complexity is the \emph{lifecycle} of event handlers
themselves.
Event handlers must be repeatedly torn down and re-bound as program state
changes.
Deciding which state changes invalidate each handler is left entirely to the
programmer.
Programmers can also forget to remove old event handlers, which gives rise to
logic bugs and memory leaks.
Events also fire at different frequencies, and programmers must guard against
high-frequency firings or unintended double inputs.

This motivates the need for static analyses that make the temporal behavior of
reactive computations explicit.
To this end, this paper introduces \lang, a simplified core calculus inspired by
React, to reason about the temporal behavior of reactive programs.
At its core, \lang models time using \emph{renders}, i.e., when a component is evaluated to
produce a description of the user interface.
Concretely, a render is triggered whenever one of the state variables changes.
These renders are treated as the fundamental computation step.

Our first contribution is an operational semantics that makes the timing of
renders and events a first-class part of the model.
Previous
semantics~\cite{madsenSemanticsEssenceReact2020,leeReacttRaceSemanticsUnderstanding2025}
faithfully reproduce React's runtime behavior, so timing is emergent from
execution.
%
We instead design a (simplified) semantics that centers timing, making the
ordering and delay of renders and events explicit in the semantic objects
themselves.
Our semantics models execution in two phases: the render phase where
code runs, and the housekeeping phase where state updates and
events are managed.
Variables are considered immutable during the render phase and are modified
during housekeeping between renders.
A novel \emph{instrumented} semantics tracks variable dependencies and
distinguishes between event-triggered evaluation steps and the propagation of updates through
reactive dependencies.
This formulation makes the timing behavior of reactive programs explicit and provides a
foundation for static reasoning.

Building on this semantics, we introduce a novel type-and-effect system that
tracks timing information.
Effects capture when states \emph{may} change; when events are bound, removed, or
cancelled; and the timing of scheduled events.
The effect of a computation intuitively describes when its observable consequences may occur
relative to the event that triggered it.
%

The key idea of our type system is to integrate state and event dependencies
with quantitative delay information.
For instance, if an update to $x$ causes an update to $y$ after 1 render, this
effect on $y$ is expressed as $\Next{1r}{} @x$ where $@x$ denotes an update to $x$.
Similarly, if a text input box has a bound onChange handler that updates state
$z$, we convey that as $\always{\eventEff{change}{input}}(\Next{1r}{\stch z})$.
Scheduled events are their own effect, simply a label with some specifying data
in a tuple.
The modality $\always{e}(F)$ marks the registration of a \emph{persistent} listener,
and its counterpart $\eventually{e}(F)$ marks a \emph{one-time} listener.
We also add modalities $\remove{e}$ and $\cancelEv{e}$ to mark removing all
listeners from an event, and cancelling one firing of an event $e$.
Finally, these effects can be composed in sequence with $F_1 * F_2$ and branches
are marked by $F_1 + F_2$.
Effects then propagate through expressions in a manner analogous to traditional
effect systems.
This approach enables modular reasoning about reactive programs, allowing the
timing properties of larger components to be inferred from those of their parts.

A key advantage of \lang{} is that captured effects can be viewed as a \emph{temporal dependency graph}
that captures variable dependency over the entire execution, not just a single render.
This enables standard graph algorithms to analyze performance and dataflow
of programs.
We use these algorithms to \emph{(i)} detect long render cascades, where
one event drives many sequential renders and degrades performance,
\emph{(ii)} flag inter-render loops, where updates repeatedly trigger one another,
\emph{(iii)} warn when expensive handlers are bound to high-frequency events,
\emph{(iv)} confirm that stale event handlers are always cleaned up, and
\emph{(v)} analyze the timing of an app's first
render.
As a result, \lang statically detects critical performance defects missed by
compiler transformations~\cite{react_compiler} and other current tools.

To evaluate the practicality of our approach, we formalize \lang and the proposed effect
system and prove key metatheoretic properties.
In particular, we establish type preservation, showing that well-typed programs respect
the timing guarantees described by their effects.
To minimize programmer burden, we have also implemented a prototype \emph{type-and-effect inference algorithm}
in Haskell that infers effects of all examples presented in the paper.\artifactnote

In summary, this paper makes the following contributions:
\begin{itemize}[leftmargin=*]
    \item A core calculus for reactive programming featuring a time-aware operational semantics making explicit the temporal and causal structure of updates
    (Section~\ref{sec:semantics}),

    \item A type-and-effect system for tracking timing constraints, enabling
    static reasoning about when computations produce observable effects (Section~\ref{sec:types}),

    \item Metatheoretic results establishing preservation of the effect system
    with respect to an instrumented semantics (Section~\ref{sec:metatheory}),

    \item Post-typecheck analyses that can statically catch common
    bugs and warn programmers what to inspect for fixing timing logic errors (Section~\ref{subsec:post-analysis}),

    \item A prototype type-and-effect inference algorithm to demonstrate
    feasibility for implementation in real \react programs (Section~\ref{sec:examples}),

    \item The type-checking of real-world examples, namely a stuck loading state, a request race condition, and an update loop
    in a realistic signup form, demonstrating how the system statically finds and
    draws attention to timing bugs (Section~\ref{sec:examples}).
\end{itemize}

\section{Overview}
\label{sec:overview}


\lang's core calculus is inspired by React's programming model and aims at capturing
the essential structure of render-based reactive programs.
Before describing \lang in detail, we provide a brief background on
React.

%

React~\cite{ReactMeta} provides a declarative model for building user interfaces (UIs) for websites and applications.
In this model, programmers describe how the UI should be derived from some application state,
and the framework automatically updates the interface whenever that state changes.

React programs are structured around \emph{components}: functions that take
inputs and return a `piece' of UI to render.
Program execution is then divided into a sequence of \emph{renders}.
During each render, a component computes a snapshot of the UI based on the current values
of its state variables and inputs.
React then compares the resulting interface with the previous one and applies the minimal
set of updates needed to the Document Object Model (DOM), i.e., visible page.
This render cycle occurs repeatedly as the program responds to events,
giving the user the impression of a continuously interactive system.
Renders typically happen at least 60 times per second, meaning each render
can only take up to 16 ms to execute, compute and apply the minimal updates.

React also provides several \emph{hooks}, which are primitive mechanisms for
managing state and side effects.
Two of the most fundamental are $\m{useState}$ and $\m{useEffect}$ hooks.
The $\m{useState}$ hook allows components to store values that persist across renders,
while $\m{useEffect}$ allows code to execute in response to changes in specific pieces of state.
The difference between a component and a hook in \react is that a component
returns render-able \html{} whereas a hook is primarily used for encapsulation
of the program logic to make programs more readable.

\paragraph{\textbf{Setters in React}}
The  $\m{useState}$ hook is  typically used in  declarations of the form:
\begin{lstlisting}
  const [x, setX] = useState(e);
\end{lstlisting}
Here, \kw{x} denotes a state variable and \kw{setX} is a \emph{setter}
function that schedules updates to the state variable.
State is \emph{immutable during a render}: updates are not applied immediately
but instead scheduled to occur between renders.
The expression \kw{e} gives an initial value to \kw{x}
the first time the component is rendered.
%
%
%
Setter functions are usually given a function of type $\tau \to \tau$
that computes a new value from a value in the previous state.
%
Thus, setters allow the type signature  $(\tau \to \tau) \to \tyUnit$.
When a setter receives a function argument, e.g., \kw{setX(x => x + 1)},
the function is evaluated between renders using the most recent state value.
Multiple setter calls during a single render are queued and
applied sequentially before the next render begins.
For instance, if a variable \kw{x} has value 3 at the start of a render and
\kw{setX(x => x + 1)} was invoked twice during that render, then \kw{x}
would remain 3 throughout the render, but
next render would observe the value 5 for \kw{x}.

\paragraph{\textbf{Effects in React}}
\react is purposefully designed to resemble functional programming.
Stateful variables are immutable during renders because each render is intended
to be a snapshot in time of the overall execution.
Side effects, such as state updates, are therefore typically performed in response
to events or within \emph{effect hooks}.
This is where \kw{useEffect} comes in handy.
\begin{lstlisting}
  useEffect(() => {
    ...side effects...
    return () => {...cleanup...}
  }, [x, y, z...])
\end{lstlisting}
The first argument is an effectful function, while the second argument specifies
an array of \emph{dependencies}.
Whenever the value of one of these dependencies changes between renders,
React schedules the effect function to execute.
The effect function may optionally return a cleanup function that runs before the effect is
re-executed and cleans up the previous render's effect.
This enables programs to react to state updates and cause side effects in
response.

\paragraph{\textbf{An example: Moving Dot}}
The react model based on components and hooks provides flexibility but reasoning about renders with
side effects is difficult.
Effects depend on state changes that occur in previous renders,
and multiple pieces of state may interact indirectly through chains of effects.
As applications grow larger, understanding the temporal behavior of a
program becomes increasingly challenging.
To illustrate these issues, consider the following (simplified) program, written
using \kw{useEffect} to display a dot that follows the
location of the user's click, together with a checkbox controlling whether movement is enabled.
%

\begin{lstlisting}[mathescape = false]
export default function MovingDot() {
  const [position, setPosition] = useState({ x: 0, y: 0 });
  const [canMove, setCanMove] = useState(true);
  let handleClick = (e) => {
    setPosition(_ => ({ x: e.clientX, y: e.clientY }));  };
    useEffect(() => {
      if (canMove) { document.addEventListener("click", handleClick); }
      return () => document.removeEventListener("click", handleClick);
    }, [canMove]);
    return ( <div>
            <input type="checkbox" checked={canMove}
                   onChange={(e) => setCanMove(_ => e.target.checked)} />
            <div style={transform: `${position.x}, ${position.y}`} ... />
           </div> ); }
\end{lstlisting}

\emph{Lines 2--3: State Declarations.}
The component declares two state variables along with their setters:
\kw{position} holds the dot's current coordinates, initialized to \kw{\{x:0, y:0\}},
and
\kw{canMove} is a boolean flag, initialized to \kw{true}, that controls whether
a click should be allowed to move the dot.

\emph{Lines 4--6: Click Handler.}
The \kw{handleClick} function that takes a DOM event \kw{e}, calls
\kw{setPosition} in the updater form: the argument
\kw{\_ => \{x: e.clientX, y: e.clientY\}} ignores the previous position
and returns the click coordinates.
Because this is a setter call, the update is \emph{not} applied immediately;
it is queued and will take effect at the start of the next render.

\emph{Lines 7--10: Effect Block.}
\kw{useEffect} declares a side-effecting block with dependency list \kw{[canMove]}.
React runs this block after any render in which \kw{canMove} has changed
(and once on the initial mount).
Each time the block runs, it first executes the cleanup function
\kw{removeEventListener} returned
by the \emph{previous} invocation, if any,  to
remove the record of the old \kw{handleClick} before setting up a new one.
The body then conditionally re-records \kw{handleClick}: if \kw{canMove}
is \kw{true}, \kw{addEventListener} attaches \kw{handleClick} to the document's
click event; if \kw{canMove} is \kw{false}, no new listener is attached.

\emph{Lines 11--15: Rendered Output.}
%
The \kw{<input type="checkbox">} renders the toggle; its \kw{onChange}
handler calls \kw{setCanMove} with the checkbox's new checked value,
scheduling a \kw{canMove} update for the next render, which will in turn
trigger a re-run of the effect block.
The positioned \kw{<div>} renders the dot itself, placed at the coordinates
stored in \kw{position}.

The correctness of this example relies on several temporal assumptions:
\begin{itemize}
  \item[i.] When the user clicks the page, the registered listener calls
  \kw{setPosition}, but \kw{position} only updates on the \emph{next}
  render, not during the event.
  \item[ii.] When \kw{canMove} flips, the effect block runs first
  the cleanup function to remove the old listener; only after that the body decides whether to
  attach a new one. Forgetting the cleanup would correspond to a program logic
  mismatch and a potential memory leak.
  \item[iii.] The points in time when a click can move the dot are not
  determined by the code itself; it is determined by when the effect block
  last ran and what \kw{canMove}'s value was at that point.
\end{itemize}

%
This kind of effect block programming is a common idiom in modern \react,
and, unfortunately, also a common source of bugs.
To use it correctly, a programmer must track which renders re-run the effect,
which of those runs trigger the cleanup, and in what order the resulting side
effects occur relative to one another and to the surrounding renders. This is
infeasible in large applications.
%
\paragraph{\textbf{Introducing \lang}}
To address the challenges discussed above, \lang introduces a novel type-and-effect system
that tracks how state variables may change over renders.
We encode the timing semantics of \lang in this effect system, so the above list
of assumptions can be known at compile time.
The key idea is to represent potential state updates explicitly as effects
describing \emph{when} (i.e., after how many renders) an update
may occur.
For a variable \kw{x}, the base effect \stch{x} indicates that \kw{x}
\emph{may change}, and $\Next{1r}{\stch{x}}$ indicates that \kw{x} may have a
different value in exactly the next render.
We write $\cdot$ for the \emph{empty effect}, indicating that no state
update may occur.

As usual in effect systems, function types carry an effect.
We use the notation $\tyArrow{\tau_1}{\tau_2}[F]$ to describe
a function whose argument type is $\tau_1$, return type is $\tau_2$,
and whose effect is $F$.
Function effects are only triggered when a function is \emph{called}; until then
they are encapsulated in the type.
%
%
Setters are primitive to \lang and we can assign them a type based on their
timing behavior:
$\tySetter{x}$.
Setters have \Next{1r}{\stch{x}} as their effect to describe semantically
that setters are evaluated between renders, and in the render after a setter is
called its variable may have been altered by the closure given to the setter.
%
%
%
Also note that the inner function of type $\tau \to \tau$ has the empty effect
$\cdot$, since the setter of \kw{x} may only update \kw{x}.

Multiple effects can be composed using the sequencing $*$ operator.
Concretely, $F_1 * F_2$ denotes that \emph{both} effects $F_1$ and $F_2$ occur
in sequence.
For instance, consider the expression: \texttt{setX(f); setY(f)}.
The effect for this program would be written as
$\Next{1r}{\stch x} * \Next{1r}{\stch y}$.
Intuitively, this can be read as: in the next render, $x$ may have changed
\emph{and also} in the next render, $y$ may have changed.
%

Effects can be composed with the $+$ operator to handle branches:
$F_1 + F_2$ indicates that either $F_1$ or $F_2$ may happen.
For instance, consider
$\kw{if b then setX(f) else setY(f)}$.
The effect for this program is written as $\Next{1r}{\stch x} + \Next{1r}{\stch y}$,
since either $x$ or $y$ may change in the next render.
%

The constructs above describe \emph{synchronous} state changes: setter calls
that are queued during a render and flushed deterministically before the next
one begins, requiring no external trigger.
\emph{Asynchronous} events, like a DOM click, a timer expiry, or a network
response may arrive at any point in time, interleaved between renders, or not at all.
Their occurrence is contingent on the outside world, not on the program's
own execution.
To describe these, the effect language provides a second layer of
modalities indexed by \emph{event labels} rather than by render counts.
The two layers meet at effect of the form
$\always{e}(\Next{1r}{\stch x})$: contingent on an external event $e$
firing, $x$ may change in the render that follows.

%


To see how \lang{}'s type-and-effect system can help in practice let
us revisit the MovingDot example. In \lang we can write it as follows.
\begin{lstlisting}
comp MovingDot () : html {
  state position, setPosition default (0, 0);
  state canMove, setCanMove default true;
  let handleClick = $\lambda$e. setPosition($\lambda$_. (e.clientX, e.clientY));
  on canMove do {
    remove $\eventEff{click}{\#doc}$;
    if canMove then (bind $\eventEff{click}{\#doc}$ handleClick) else () };
  return ( ... /* $\text{same html as in React}$ */ ); }
\end{lstlisting}
The \lang implementation is similar to the one in \react we presented
before with a few differences.
State declarations specify the default initial value explicitly, and we use
$\lambda$-expressions in place of \kw{(e) => ...} closures.
The \kw{useEffect} block of \react is replaced by an \emph{on-block} of the form
\kw{on canMove do
\{...\}} which fires whenever its watched
variable \kw{canMove} changes.
An on-block in \lang executes \emph{cleanup at the beginning}, if
needed. In our example \kw{remove} is called first, and only after
that the rest of the body is executed. Notice that the conditional
registration is now the body of the on-block, not a check inside
\kw{handleClick}.


The argument \eventEff{click}{\#doc} is an \emph{event label}: a tagged
identifier for a class of external events.
The label has the form
$\eventEffBar{\ell}{v}$, where $\ell$ names the kind of event
(\kw{click}, \kw{timeout}, \kw{req}, \dots) and $\overline{v}$ is a tuple
of statically-known values identifying \emph{which} event of that kind
(here, the document node \kw{\#doc}; for a request, a URL).
Only statically-known values may appear in the tuple: dynamic values such
as a freshly-generated timer id cannot enter the effect layer, so events
like \eventEff{timeout}{} carry an empty tuple.
%

The type-and-effect system reveals both how user interactions influence program
state and how the lifecycle of the click subscription itself is managed.
The relevant effects that \lang infers for the MovingDot example are the following:
\begin{align*}
  \kw{bind}\;\eventEff{click}{\#doc}\;\kw{handleClick}
    \;&:\; \always{\eventEff{click}{\#doc}}\bigl(\Next{1r}{\stch{\kw{position}}}\bigr) \\
  \kw{remove}\;\eventEff{click}{\#doc}
    \;&:\; \remove{\eventEff{click}{\#doc}} \\
  \text{\kw{on canMove}~block body}
    \;&:\; \remove{\eventEff{click}{\#doc}} \;*\; \bigl(\always{\eventEff{click}{\#doc}}\bigl(\Next{1r}{\stch{\kw{position}}}\bigr) + (\cdot)\bigr) \\
  \kw{onChange}
    \;&:\; \Next{1r}{\stch{\kw{canMove}}}
\end{align*}


The construct \kw{bind} registers a
\emph{persistent} handler against an event label.
Every time the event
fires, the handler runs.
We describe this with the graded square modality
$\always{\eventEff{click}{\#doc}}(F_h)$, read ``every time
\eventEff{click}{\#doc} fires, $F_h$ happens.''
Here $F_h$ is the effect of the
handler body. In our example, \kw{handleClick} calls \kw{setPosition}, so $F_h =
\Next{1r}{\stch{\kw{position}}}$.
%
The construct \kw{remove} unregisters
all handlers for the named event. We describe this with the effect
$\remove{\eventEff{click}{\#doc}}$, read ``the registration for
\eventEff{click}{\#doc} no longer fires its handler.''
%
%
%
%
\lang's type-and-effect system encodes directly the timing assumptions we
mentioned earlier from \react.
\begin{enumerate}
  \item[i.] We see in the effect of click listener that it modifies
  \kw{position} on the next render because \Next{1r}\stch{\kw{position}}
  signifies that "in the next render position may change."
  \item[ii.] When \kw{canMove} changes we see that
  \remove{\eventEff{click}{\#doc}} happens before the new event handler is
  registered, so we know that stale event handlers are cleaned up.
  \item[iii.] We see in the effect of the on-block that the event handler is
  only sometimes registered
  $\bigl(\always{\eventEff{click}{\#doc}}\bigl(\Next{1r}{\stch{\kw{position}}}\bigr)
  + (\cdot)\bigr)$, which indicates that our handler is conditional.
\end{enumerate}
%
By exposing this structure, \lang helps programmers reason
about how reactive programs evolve.

%


\paragraph{\textbf{A second example: Debounce.}}
To illustrate one-time listeners in \lang{}, we consider a \kw{Debounce} component,
which takes a quickly-changing input \kw{value} and exposes a \kw{slow} value
that updates only when \kw{value} has not changed for $100$ms.
A debounce protects expensive downstream work from executing too frequently.
This can be implemented in \lang as follows:
\begin{lstlisting}
comp Debounce (value: any) : any {
  state slow, setSlow default value;
  on value do {
    clearTimeout();
    setTimeout($\lambda$_. setSlow($\lambda$_. value))
  };
  return slow;
}
\end{lstlisting}

We first initialize a \kw{slow} stateful variable, which will store our slowed
down value.
We bind an on-block to the fast \emph{value}. In this block,
we create a timeout (running for {100\kw{ms}) to update the \kw{slow} state, but first we cancel our previous timeout.
%

The builtins \kw{setTimeout} and \kw{clearTimeout} have types:
\begin{align*}
  \kw{setTimeout f}
    \;&:\; \eventually{\eventEff{timeout}{}}(F_f) \;*\; \Next{100\kw{ms}}{\eventEff{timeout}{}} \\
  \kw{clearTimeout()}
    \;&:\; \cancelEv{\eventEff{timeout}{}} \;*\; \remove{\eventEff{timeout}{}}
\end{align*}

Unlike \kw{bind}, \kw{setTimeout}
registers a \emph{one-shot} handler.
When \eventEff{timeout}{} fires the callback runs once
and the registration is
consumed.
We describe this behavior with the graded diamond\\
$\eventually{\eventEff{timeout}{}}(F)$, read ``contingent on
\eventEff{timeout}{} firing, $F$ happens once.''
Here $F$ is the effect of the
callback body: \kw{setSlow} produces $\Next{1r}{\stch{\kw{slow}}}$.
The accompanying $\Next{100\kw{ms}}{e}$ records that the event
\eventEff{timeout}{} itself fires after $100$ms of wall-clock time.
\lang's $\Next{}{}$ modality is indexed not only by renders
\kw{r}, but by any declared time unit --- renders, milliseconds, network
round-trips --- so the system can describe asynchronous schedules
alongside synchronous render counts.

For timeouts, a \eventEff{timeout}{} event is scheduled and
a handler is registered.
To clean up both, \kw{clearTimeout} produces cancel and remove effects
$\cancelEv{\eventEff{timeout}{}} * \remove{\eventEff{timeout}{}}$.
The cancel effect will prevent the \eventEff{timeout}{} event from firing, and
the remove effect clears all current event listeners to the event.
Putting it together, the cascading effect of \kw{value} in \kw{Debounce}
is
\[
  \cancelEv{\eventEff{timeout}{}} \;*\; \remove{\eventEff{timeout}{}} \;*\;
  \eventually{\eventEff{timeout}{}}\bigl(\Next{1r}{\stch{\kw{slow}}}\bigr) \;*\;
  \Next{100\kw{ms}}{\eventEff{timeout}{}}
\]
which reads: when \kw{value} changes (1) any pending \eventEff{timeout}{} is
cancelled and its handlers are unregistered (2) a fresh \eventEff{timeout}{} is
scheduled to fire after $100$ms (3) contingent on it firing, in the next render
\kw{slow} may change.

The payoff of the debounce example is that the protection survives
composition.
Suppose a parent component watches \kw{slow} from a
\kw{Debounce} child and issues a network request:
\begin{align*}
&\synSubComp{\text{slowInput}}{\kw{Debounce}(\text{input})} \\
&\synOn{\text{slowInput}}{\kw{fetch}(\text{``/api/search?q=''} + \text{slowInput}, \ldots)}
\end{align*}
\noindent Then the effect of \kw{input} carries the prefix
$\eventually{\eventEff{timeout}{}}(\Next{1r}{(\Next{1n}{F})})$,
so the network call $F$ is visibly behind a debounce timer.
If \kw{input} is changed by a frequent event (e.g.\ an \kw{onChange} on
an \kw{<input/>}), a parent or library author can see in the type that
the downstream \kw{fetch} is protected; if a debounce were missing, the
type would be missing the leading $\eventually{}$ and \lang's
post-hoc analyses could warn the programmer.

\paragraph{\textbf{Temporal Dependency Graphs}}
\react's \kw{useEffect} is intentionally an ``escape hatch'' from the
declarative paradigm (see David
Khourshid's ``Goodbye useEffect'' \cite{khourshid2022goodbye}).
For most of React we rely on predictable functional reactive
programming, and then carve out a section of our program where that
predictability no longer holds.
Synchronizing with external systems can have unexpected behavior, effects may
clash with sibling effects, and cleanup happens some unknowable amount of time
in the future.

\lang's type-and-effect system can be used to explore the \emph{temporal
dependency graph}. In this graph, nodes represent potential state changes,
and edges model delays between two nodes.
This way, the graph visualizes effects over \emph{multiple renders}, even though
our effects only capture dependencies over a single render.
By recursively expanding these effects over the graph, we can compute the \emph{full
effect} of any event.

%
%
Temporal dependency graphs can be used to automatically detect inter-render loops,
analyze the performance of event handlers and an app's initial render, and ensure stale event
handlers are not left behind.
Each of these analyses can be performed via off-the-shelf graph walking algorithms.
The simplest example that the effect system catches is a state-change loop, showcased in
the following example and the corresponding looping graph.
%
\begin{center}
\begin{minipage}[c]{0.56\linewidth}
\begin{lstlisting}
  comp MutualRecursion (clock: int) : unit {
    state x, setX default 0;
    state y, setY default 0;
    on x do { setY addOne };
    on y do { setX addOne };
    return ();
  }
\end{lstlisting}
\end{minipage}\hfill
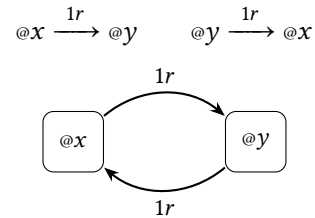
\begin{minipage}[c]{0.40\linewidth}
\centering
$\stch{x} \xrightarrow{\;1r\;} \stch{y}
 \qquad
 \stch{y} \xrightarrow{\;1r\;} \stch{x}$

\medskip

\begin{tikzpicture}[>={Stealth[length=2mm]}, node distance=16mm,
    every node/.style={font=\small},
    nd/.style={draw, rounded corners, inner sep=4pt, minimum size=8mm}]
  \node[nd] (x) {$\stch{x}$};
  \node[nd, right=of x] (y) {$\stch{y}$};
  \draw[->, thick] (x) to[bend left=40] node[above] {$1r$} (y);
  \draw[->, thick] (y) to[bend left=40] node[below] {$1r$} (x);
\end{tikzpicture}
\end{minipage}
\end{center}
The effect of changing \kw{x} is $\Next{1r}{\stch{y}}$; and vice versa,
the effect for changing \kw{y} is $\Next{1r}{\stch{x}}$.
From the graph, we can infer that the full effect of modifying \kw{x} is
an infinite chain of modifications of \kw{x} and \kw{y}.
\lang can auto-detect this kind of error and warn the programmer at compile time.
%

\section{Formal Syntax}
\label{sec:syntax}

\begin{figure}[t]
\label{fig:syntax-of-willow}
\begin{tabbing}
  Declarations\quad\=$p$ \quad\= $::=$\quad\= \kill
  Components \>$C$   \> $::=$  \> \synComponent{A}{x_1 : \tau_1, x_2 : \tau_2, \ldots, x_n : \tau_n}{\tau}{p_1 \; p_2 \; \ldots \; p_n \; \synReturn{x}}\\[4pt]
  Declarations \>$p$ \> $::=$  \> \synCompLet{x}{e}
                    $\mid$ \synState{x}{e}
                    \\\>\>\>   \synOn{x_1, x_2, \ldots, x_n}{e}
                    $\mid$ \synSubComp{x}{A(x_1, x_2, \ldots, x_n)}\\[4pt]
  Expressions \>$e$  \> $::=$  \> $x$
                    $\mid$ $c \in \{\m{true}, \m{false}, ()\}$
                    $\mid$ $e_1\;e_2$
                    $\mid$ $e_1;\;e_2$
                    $\mid$ \synFn{x}{e}
                    \\\>\>\> \synIf{x}{e_1}{e_2}
                    $\mid$ $(e_1,e_2)$
                    $\mid$ $\kw{fst}\;e$
                    $\mid$ $\kw{snd}\;e$
                    \\\>\>\> $\kw{bind}\;\eventEffBar{\ell}{v}\;e$
                    $\mid$ $\kw{once}\;\eventEffBar{\ell}{v}\;e$
                    $\mid$ $\kw{cancel}\;\eventEffBar{\ell}{v}$
                    $\mid$ $\kw{remove}\;\eventEffBar{\ell}{v}$
\end{tabbing}
\caption{Syntax of \lang Programs}
\label{fig:syntax}
\Description{Syntax}
\end{figure}

\lang's syntax is inspired by \react with a few simplifications.
Like \react, \lang programs are composed of components.
The last declared component is considered as the main toplevel component of the
program.
Term $C$ denotes a component named $A$ that takes a sequence of arguments
(similar to a function), and contains a series of declarations, terminated by
return of a variable.

Declarations $p$ have four possibilities: a simple let declaration
$\synCompLet{x}{e}$ binds variable $x$ to the value of $e$.
A state declaration binds a state variable $x$ and its setter $\mi{setX}$, and
is required to include a default initial value in the form of an expression.
Effect blocks take a list of variables that are monitored such that when any of
them change value compared to the previous render, the effect block's expression
$e$ is executed.
Finally, we also support \emph{subcomponents} that enable a programmer to
use any of the previously declared components.
The declaration $\synSubComp{x}{A(x_1, x_2, \ldots, x_n)}$ executes component
$A$ on arguments $x_1, \ldots x_n$ and binds the returned value to $x$.

Expressions in \lang appear in let, state, and effect block declarations.
Most of the expressions are standard, namely variables, constants,
function applications, and lambda-expressions.
If-expressions are written as $\synIf{x}{e_1}{e_2}$ which branches on the value of $x$.
We support pairs of the form $(e_1, e_2)$ which can be projected out using
\kw{fst} and \kw{snd} respectively.
%
Finally, \lang provides primitives for interacting with the event layer.
The expressions $\kw{bind}\;\eventEffBar{\ell}{v}\;e$ and $\kw{once}\;\eventEffBar{\ell}{v}\;e$
register the closure $e$ as a handler for the event $\eventEffBar{\ell}{v}$:
\kw{bind} installs a persistent listener that runs on every firing,
while \kw{once} installs a one-shot listener that runs at most once.
The expressions $\kw{cancel}\;\eventEffBar{\ell}{v}$ and $\kw{remove}\;\eventEffBar{\ell}{v}$
tear down event work: \kw{cancel} suppresses a single pending firing
of $\eventEffBar{\ell}{v}$, while \kw{remove} unregisters every listener
attached to $\eventEffBar{\ell}{v}$.


%

\section{Semantics}\label{sec:semantics}

We have created a formal render-based semantics for \lang inspired by \react and
many other frontend \js libraries/frameworks.
Render-based semantics are generally used in the context of frontend user
interface coding, so we discuss our semantics in terms relevant to
this paradigm.

Render-based semantics work by running a program's code every time a piece of
program state changes.
Whatever the program returns is considered the output for that render.
In frontend web development the output of a program is \html.
In many frameworks there is a virtual copy of the \html visible to the user
called the "Virtual DOM."
The returned \html is diffed with this Virtual DOM to find the exact changes
that need to be made to the visible \html, and these changes are made in one
atomic step.
To keep \lang environment-agnostic we don't formalize what happens with
the output of a \lang program.
We assume that there is a higher runtime that \lang exists within and it
handles the returned outputs on each render.

Real-world render-based frameworks make heavy use of caching and memoization in
order to improve performance.
We ignore these optimizations because pieces of software such as React Compiler \cite{react_compiler}
have substantially reduced the need to introduce these optimizations manually.

The declarative render-based programming style revolves around the existence of
an event queue.
For example in the \js runtime, to initiate a network request a programmer
calls \texttt{fetch} with a URL and request information and attaches a closure
via \texttt{.then} to run when the response arrives.
The in-flight request is handled outside \js by the host environment, and when
it completes the closure is queued to run on the next turn of the event loop.
%
The resulting execution order, with microtasks interleaved against the task
queue and \react's own priority-lane scheduler on top, can be as difficult to
predict as understanding any given \js runtime.
We simplify these render-based semantics with a more concrete execution order.
%

In our formal semantics, \lang's internal event queue holds closures associated
with setter calls; event interactions introduced by the event
primitives \kw{bind}, \kw{once}, \kw{cancel}, \kw{remove}; and events being
fired \eventEffBar{\ell}{v}.
%
%
%
\lang's queue is therefore concerned only with the events its primitives
introduce.
Events from the external environment move into \lang's event queue via the
"external scheduler" which is a queue of events that \lang flushes between
renders.

There are two phases to \lang program execution:
(1) The render phase, in which programmer-written code executes and the event queue is built,
(2) The housekeeping phase, in which the scheduler and event queue are flushed and program
state is updated.
These two phases are looped over and over as new inputs arrive from the external
environment.

%
%
%

%



\subsection{Semantics Derivations}
In all the semantic derivations there is a program signature $\Sigma$ that does
not change throughout program execution.
$\Sigma$ is a mapping from component names to a tuple of the full component's
code, its effect environment $\Delta$, and type environment $\Gamma$ which we
discuss in \S~\ref{sec:types}.
%
%
$$A: (\synComponent{A}{\overline{x:\tau}}{\tau_r}{p}, \Delta, \Gamma)$$

\emph{Component judgment.}
\semDerivTopStep
  {\mathcal{L}}
  {V_e}
  {V_s}
  {\mathcal{C}}
  {\synComponent{C}{\overline{x = v}}{\tau}{p}}
  {s}
  {U}
takes one step of component evaluation.
Here $\mathcal{L}$: listener map from event labels to registered listeners
tagged $\always{}$ persistent or $\eventually{}$ one-shot, $V_e$: snapshot of
the value environment from  the end of the previous render, $V_s$:
state-variable environment seeded from $V_e$ and mutated by setter calls
during housekeeping, $\mathcal{C}$: cancellation multiset with one entry per
pending \kw{cancel} suppressing the next dispatch of the named event, \kw s: the
execution status either \kw{rendered} if the event queue is not finished
processing or \kw{waiting} if housekeeping is over, and $U$: \lang's event
queue of setter closures and event-primitive operations awaiting housekeeping.


\emph{Declaration judgment.}
\semDerivDecl{\Sigma}{V_e}{V_s}{V}{c}{p}{v}{U}{V'} \;
evaluates declarations $p$ at component path $c$ under working
environment $V$, threading $V_e$ and $V_s$
through unchanged.
It produces the returned value $v$, the update queue $U$
accumulated as a side effect, and the final value environment $V'$
witnessed at return.
Here we introduce $V$: working environment of variable values in scope during
the current render, $c$: dotted component path used as the key for setter
dispatch, $V'$: value environment produced after all declarations in $p$ are
evaluated.

\emph{Expression judgment.}
\semDerivExpr{V}{e}{v}{U}
evaluates expression $e$ under working environment $V$ to value $v$, producing
update queue $U$ as a side effect.

\emph{Effect vs.\ state environments.}
The pair $V_e$, $V_s$ carries values between renders.
During the render phase they are read-only and during the housekeeping phase they are
modified.
$V_s$ is read to evaluate the \kw{state} declaration, the state's value is taken
from $V_s$ and bound with its name in $V$.
\lang writes to $V_s$ while flushing setters from the event queue during
housekeeping.
$V_e$ is a frozen snapshot of the value environment taken just before
housekeeping, when choosing whether to execute an effect block the bound
variables are checked for differences between $V_e$ and $V$.
$V_s$ begins each housekeeping phase equal to $V_e$ and is mutated as setter
calls are popped from the update queue.


\emph{Component paths.}
Setter evaluation happens between renders, so each setter must be paired with
the state variable it targets.
However, because \lang allows components to be used multiple times as
subcomponents, we need to be able to address each piece of state with a globally
unique name.
The component path $c$ provides this unique identifier.
We write $c.x$ for the path obtained by appending $x$ to $c$, with
``\texttt{.}'' as separator and no leading dot.
At the top level the path is empty and $c.x = x$, while inside a subcomponent
bound to variable $y$ in the parent component the path is $c'.y$ where $c'$ is
the parent's path and $c.x = c'.y.x$.
The path attached to each setter is the key for that state variable in $V_s$ and
$V_e$ at the top level, so setter dispatch is exact-key lookup.

\subsection{Component}

%
The component-level rules drive (1) the housekeeping phase in which
the event queue is flushed, and (2) the transition between renders
in which fresh arguments arrive and a new render is executed.
We show the most interesting rules in this section collected in
Figure~\ref{fig:sem-component} and leave the rest to the
appendix. 
Rules at this level only deal with executing the main component of the program.
We evaluate subcomponents at the declaration level of the semantics.
We make this choice so that all variable names are relative to the top main
component in this rules level, which makes handling state variables more
straightforward.

We start with showing how setters are evaluated between renders.
When a setter is called at the expression level, it does not immediately change
the value of its corresponding state variable but instead is entered into the
event queue to be evaluated during housekeeping.
In \rulename{flush first update} we pop off a setter call event from the queue
and execute its closure.

Before evaluating its closure we first remove all setters from the environment,
because setters are not allowed to have any effects.
The value returned from the closure is used as the new value for the $y$ state,
which is achieved by updating $y$ in $V_s$.
If $v \neq v'$ then $y$'s value will differ between $V_s$ and $V_e$, so for the next render $y$
will trigger the effects it is bound to.




\begin{figure}[t]
\begin{center}
\Rule{flush first update}
  {
    \semDerivTopStep
      {\mathcal{L}}
      {V_e}
      {V_s, y=v}
      {\mathcal{C}}
      {\synComponent{C}{\overline{x = v}}{\tau}{p}}
      {rendered}
      {setter_y(\Closure{y}{e}{V}),U'}
    \\\\
    V' = removeSetters(V) \\
    \semDerivExpr{V',y=v}{e}{v'}{\cdot}
  }{
    \semDerivTopStep
      {\mathcal{L}}
      {V_e}
      {V_s, y=v'}
      {\mathcal{C}}
      {\synComponent{C}{\overline{x = v}}{\tau}{p}}
      {rendered}
      {U'}
  }
  \\[10pt]
\Rule{flush external events}
  {
    \semDerivTopStep
      {\mathcal{L}}{V_e}{V_s}{\mathcal{C}}
      {\synComponent{C}{\overline{x = v}}{\tau}{p}}
      {rendered}{U}
    \\
    \mathcal{E} \;\Rightarrow_{\mathsf{ext}}\; U_{\mathcal{E}}
  }{
    \semDerivTopStep
      {\mathcal{L}}{V_e}{V_s}{\mathcal{C}}
      {\synComponent{C}{\overline{x = v}}{\tau}{p}}
      {rendered}{U_{\mathcal{E}},\,U}
  }
  \\[10pt]
\Rule{pop listen}
  {
    \semDerivTopStep
      {\mathcal{L}}{V_e}{V_s}{\mathcal{C}}
      {\synComponent{C}{\overline{x = v_1}}{\tau}{p}}
      {rendered}{\mathsf{listen}(\ell\langle\overline{v_2}\rangle,\, c,\, m),\, U'}
    \\\\
    m \in \{\eventually{}, \always{}\}
    \\
    \mathcal{L}' = \mathcal{L}[\ell\langle\overline{v_2}\rangle \mapsto \mathcal{L}(\ell\langle\overline{v_2}\rangle) \cup \{(c, m)\}]
  }{
    \semDerivTopStep
      {\mathcal{L}'}{V_e}{V_s}{\mathcal{C}}
      {\synComponent{C}{\overline{x = v_1}}{\tau}{p}}
      {rendered}{U'}
  }
  \\[10pt]
\Rule{pop event}
  {
    \semDerivTopStep
      {\mathcal{L}}{V_e}{V_s}{\mathcal{C}}
      {\synComponent{C}{\overline{x = v}}{\tau}{p}}
      {rendered}{(\ell\langle\overline{v_1}\rangle,\, v_2),\, U'}
    \\\\
    \mathcal{C}(\ell\langle\overline{v_1}\rangle) = 0
    \\
    \mathcal{L}(\ell\langle\overline{v_1}\rangle)
      = \{(\Closure{x_1}{e_1}{V_1}, m_1), \ldots, (\Closure{x_k}{e_k}{V_k}, m_k)\}
    \\
    \overline{\semDerivExpr{V_i,\; x_i = v_2}{e_i}{()}{U_i}}
    \and
    \overline{\listenersafe{U_i}}
    \\
    \mathcal{L}' = \mathcal{L}\bigl[\ell\langle\overline{v_1}\rangle \mapsto
      \{(\Closure{x_i}{e_i}{V_i}, m_i, F_i) \mid m_i = \always{}\}\bigr]
  }{
    \semDerivTopStep
      {\mathcal{L}'}{V_e}{V_s}{\mathcal{C}}
      {\synComponent{C}{\overline{x = v}}{\tau}{p}}
      {rendered}{U',\, U_1, \ldots, U_k}
  }
  \\[10pt]
\Rule{waiting to rendered}
  {
    \semDerivTopStep
      {\mathcal{L}}
      {V_e}
      {V_s}
      {\mathcal{C}}
      {\synComponent{C}{\overline{x = v}}{\tau}{p}}
      {waiting}
      {\cdot}
    \\\\
    \semDerivDecl{\Sigma}{V_e}{V_s}{\overline{x = v'}}{""}{p}{v_r}{U'}{V'} \\
    \exists i \;s.t. \; v_i \neq v'_i \\
  }{
    \semDerivTopStep
      {\mathcal{L}}
      {V'}
      {V'}
      {\mathcal{C}}
      {\synComponent{C}{\overline{x = v'}}{\tau}{p}}
      {rendered}
      {U'}
  }
\end{center}
\caption{
  Selected component-level semantics rules
}
\label{fig:sem-component}
\Description{Selected component-level operational semantics rules for \lang.}
\end{figure}


The \rulename{flush external events} rule shows how event-effects are flushed
from the external scheduler and enter \lang's event queue.
We take $\mathcal{E}$ as a given object \lang has no control over.
$\mathcal E$ is allowed to fill with labelled events such as mouse clicks,
network request successes or failures, etc, and \lang handles these once a
render at the start of housekeeping.
Here $U_{\mathcal{E}}$ is a series of pairs $(\ell\langle\overline{v_1}\rangle,\, v_2)$
of labelled events with the values they carry.
An example could be $(\kw{click[doc]}, (100, 200))$ for a document click at
position x=100 y=200.


The \rulename{pop listen} rule shows how event listeners are added.
The \kw{bind} and \kw{once} expressions do not instantly bind a closure to an
event, but instead add a
$\mathsf{listen}(\ell\langle\overline{v_2}\rangle,\,c,\, m)$ event to the queue.
The elements of the listen event's tuple are $\ell\langle\overline{v_2}\rangle$
the event to be listened to, $c$ the closure to execute on fire, and $m$ the
modality either \always{} or \eventually{} depending on if the listener was
created with \kw{bind} or \kw{once}.
The closure and modality are then added to the ordered set of event handlers


The \rulename{pop event} rule shows how an event handler is executed in response
to an event firing $(\ell\langle\overline{v_1}\rangle,\, v_2)$.
First, we require that the event has not been cancelled so we enforce that the
cancellation multiset does not contain any instances of
$\ell\langle\overline{v_1}\rangle$.
Next we retrieve all listeners of the event
$\overline {(\Closure{x_1}{e_1}{V_1}, m_1)}$
from the listener mapping
$\mathcal L$.
The closures are all executed and given $v_2$ as their argument, these are
allowed to have any side effect.
We enforce that the queues are "listener safe" which means that they are not
allowed to output a new labelled event $\ell'\langle\overline{v'}\rangle$ in
their event queue.
The event queues produced by executing the listeners are appended to the
program event queue to be evaluated.
Finally, only listeners marked with the \always{} modality are kept in the
listener mapping.


Once the event queue is fully flushed the state transitions from
\kw{rendered} to \kw{waiting} where \lang waits for new top-level arguments.
Now we review where a render occurs.
In \rulename{waiting to rendered}, \lang receives new top-level arguments from
the external program.
To write this in the semantics we require that the previous render had arguments
$v$ and the next render will have arguments $v'$ where at least one argument
differs.
Next we can evaluate the main component's declarations with the declaration
evaluation judgement, receive a new event queue to be processed, and start
housekeeping over again.
%
%
This rule is where we see the value environment from the end of the previous
render $V'$ become the initial value for $V_e$ and $V_s$.


\subsection{Declarations}

Declarations evaluate under the working environment $V$.
At this layer we read from the effect and state environments $V_e$ and $V_s$
and accumulate an update queue as a side effect.
We review re-render behavior for space purposes because the first
render behavior is straightforward; the rules appear in
Figure~\ref{fig:sem-decl}.

In \rulename{state rerender}, a re-rendered state variable declaration's
\kw{default} expression is ignored.
The \kw{default} expression is used on the first render to bootstrap an initial
value, but now we use the value in $V_s$ that has been modified by setter calls
during housekeeping.
To $V$ we add $x=v_0$ from $V_s$, and a setter value with label $c.x$.
Setters appear to act like normal functions in the type system, but in the
semantics they are a dummy value that holds the global path to its state.
This is the bridge between renders: setters dispatched during the previous
housekeeping have already updated $V_s$, so the next render sees their
effects without re-running initialization.

\begin{figure}[t]
\begin{center}
\Rule{state rerender}
  {
    \semDerivDecl{\Sigma}{V_e}{V_s}{V, x=v_0, \setter{x}=setter_{c.x}}{c}{p}{v}{U}{V'} \\
    x=v_0 \in V_s, x=v_e \in V_e
  }
  {
    \semDerivDecl
      {\Sigma}{V_e}{V_s}{V}
      {c}
      {\synState{x}{e}
        p}
      {v}{U}{V'}
  }
  \\[10pt]
\Rule{effect rerender, no changes}
  {
    \semDerivDecl{\Sigma}{V_e}{V_s}{V}{c}{p}{v}{U}{V'} \\
    \overline{x=v} \in V_e \\
    \overline{x=v'} \in V \\
    \forall i, v_i = v_i' \\
  }
  {
    \semDerivDecl
      {\Sigma}{V_e}{V_s}{V}
      {c}
      {\synOn{x_1, x_2,\ldots,x_n}{e}
        p}
      {v}{U}{V'}
  }
  \\[10pt]
\Rule{effect rerender, yes changes}
  {
    \semDerivDecl{\Sigma}{V_e}{V_s}{V}{c}{p}{v}{U'}{V'} \\
    \semDerivExpr{V}{e}{v_e}{U} \\
    \overline{x=v} \in V_e \\
    \overline{x=v'} \in V \\
    \exists i, v_i \ne v_i' \\
  }
  {
    \semDerivDecl
      {\Sigma}{V_e}{V_s}{V}
      {c}
      {\synOn{x_1, x_2,\ldots,x_n}{e}
        p}
      {v}{U,\, U'}{V'}
  }
  \\[10pt]
\Rule{subcomp}{
    \semDerivDecl{\Sigma}{V_e'}{V_s'}{\overline{x'=v}}{c.y}{p_A}{v_A}{U_y}{V_y'} \\
    \semDerivDecl{\Sigma}{V_e}{V_s}{V, y= v_A}{c}{p}{v}{U}{V'} \\
    V_e' = \text{removePrefix}("y.",\text{startsWith}("y.",V_e)) \\
    V_s' = \text{removePrefix}("y.",\text{startsWith}("y.",V_s)) \\
    V_y'' = \text{addPrefix}("y.", V_y') \\
    \overline{x=v} \in V \\
    A: (\synComponent{A}{\overline{x'}}{\tau}{p_A}, \Delta, \Gamma) \in \Sigma\\
  }{
    \semDerivDecl
      {\Sigma}
      {V_e}{V_s}{V}{c}
      {\synSubComp{y}{A(\overline{x})}p}
      {v}{U_y, U}{V'\cup V_y''}
  }
\end{center}
\caption{Re-render declaration-level semantics rules.}
\label{fig:sem-decl}
\Description{Declaration-level operational semantics rules for re-rendering in \lang.}
\end{figure}


Next we see how effect blocks are executed.
This is the only declaration that's allowed to perform side-effectful functions
in \lang.
Here we check for differing values of the watched variables $\overline{x}$
between $V_e$ and $V$.
If there are no differences then we are in the \rulename{effect rerender, no
changes} case and we skip the declaration.
If at least one variable has changed then we are in \rulename{effect rerender,
yes changes} and move to evaluate the effect block's body
\semDerivExpr{V}{e}{v_e}{U}.
This expression can have a side effect, which we add to the event queue to be
processed in housekeeping, the return value of this expression is discarded.


The subcomponent accepts the relevant subsets of $V_e$ and $V_s$ from its
parent.
The subcomponent returns a value, an update queue, and its final value
environment, which need to be prefixed with the parent's path.
Reading \rulename{subcomp} top to bottom:
\begin{enumerate}
  \item Retrieve the subcomponent's code $p_A$ from $\Sigma$.
  \item Calculate $V_e'$ and $V_s'$ by restricting to entries whose keys begin with
  ``$y.$'' and strip the prefix.
  \item Evaluate $p_A$ under $V_e'$, $V_s'$, and a working
  environment of the subcomponent's arguments.
  \item Re-prefix the final value environment $V_y'$ with ``$y.$'',
  yielding $V_y''$.
  \item Bind $y$ to the subcomponent's return $v_A$, then
  evaluate the remaining parent declarations.
  \item Place the subcomponent's queue $U_y$ before the parent's queue
  $U$ and union $V_y''$ into $V'$ so that the subcomponent's state
  is stored for the next render.
\end{enumerate}


\subsection{Expression} \;

Most of the expression semantics of \lang are unchanged from a normal
expression semantics of any other language.
\lang expressions create an update queue as a side effect which is ordered in
execution order the same as the update queue for declarations.
Take for example the \rulename{seq} rule in Figure~\ref{fig:sem-expr}.
\begin{figure}[t]
\begin{center}
\Rule{seq}
  {
    \semDerivExpr{V}{e_1}{v_1}{U_1} \\
    \semDerivExpr{V}{e_2}{v_2}{U_2} \\
  }
  {
    \semDerivExpr{V}{e_1\;;\;e_2}{v_2}{U_1, U_2}
  }
\Rule{function app}
  {
    \semDerivExpr{V}{e_1}{\Closure{x}{e_3}{V'}}{U_1} \\
    \semDerivExpr{V}{e_2}{v_2}{U_2} \\
    \semDerivExpr{V', x: v_2}{e_3}{v_3}{U_3} \\
  }
  {
    \semDerivExpr{V}{e_1\;e_2}{v_3}{U_1, U_2, U_3}
  }
\Rule{setter app}
  {
    \semDerivExpr{V}{e_1}{setter_x}{U_1} \\
    \semDerivExpr{V}{e_2}{\Closure{y}{e_3}{V'}}{U_2}
  }
  {
    \semDerivExpr{V}{e_1\;e_2}{()}{U_1, U_2, setter_x(\Closure{y}{e_3}{V'})}
  }
\Rule{bind}
  {
    \semDerivExpr{V}{e}{\Closure{x}{e'}{V'}}{U}
  }
  {
    \semDerivExpr{V}{\kw{bind}\;\ell\langle\overline{v}\rangle\;e}{()}
      {U,\, \mathsf{listen}(\ell\langle\overline{v}\rangle,\, \Closure{x}{e'}{V'},\, \always{})}
  }
\Rule{cancel}
  {
  }
  {
    \semDerivExpr{V}{\kw{cancel}\;\ell\langle\overline{v}\rangle}{()}
      {\mathsf{cancel}(\ell\langle\overline{v}\rangle)}
  }
\Rule{remove}
  {
  }
  {
    \semDerivExpr{V}{\kw{remove}\;\ell\langle\overline{v}\rangle}{()}
      {\mathsf{remove}(\ell\langle\overline{v}\rangle)}
  }
\end{center}
\caption{Selected expression-level semantics rules.}
\label{fig:sem-expr}
\Description{Selected expression-level operational semantics rules for \lang.}
\end{figure}

The first and second expressions are evaluated, they return values, and produce
update queues $U_1$ and $U_2$ as a side effect.
In this particular rule we discard $v_1$, return $v_2$, and as a side effect
produce the queue $U_1, U_2$ because the events in $U_1$ happened first.

The \rulename{function app} rule is normal with added event queue handling.
When a function application is a normal closure value then the closure is
evaluated immediately and its effect is placed in the event queue.
In the \rulename{setter app} rule we see how the setter dummy value is used.
This rule requires its argument to be a closure value, and places the closure
onto the event queue rather than executing it right now.
As we saw earlier this closure will be used to modify the value of $x$ between
renders.


Finally we see the \rulename{bind}, \rulename{cancel}, and \rulename{remove}
rules; \rulename{once} is exactly like bind but outputs an \eventually{} label
instead of an \always{}.
These each have their own corresponding event that they place on the queue.
We reviewed the listen event earlier, and review the cancel and remove events in
the appendix.
\kw{bind} and \kw{once} both evaluate their argument expression to a closure and
emit a listen event tagged with their corresponding modality.
\kw{cancel} and \kw{remove} take no expression argument and unconditionally emit
their corresponding queue item.

\section{Type-and-Effect System}
\label{sec:types}
\begin{figure}[t]
\begin{tabbing}
  Event Labels\quad\=$\eventEffBar{\ell}{v}$ \quad\= $::=$\quad\= \kill
  Types \>$\tau$  \> $::=$  \> $\tyUnit$
                    $\mid$ $\tyBool$
                    $\mid$ $\tau_1 \times \tau_2$
                    $\mid$ $(\tyArrow{\tau_1}{\tau_2}[F])$\\[4pt]
  Effects \>$F$  \> $::=$  \> $\cdot$
                    $\mid$ $\stch{x}$
                    $\mid$ $\eventEffBar{\ell}{v}$
                    $\mid$ $\Next{Nu}{F}$
                    $\mid$ $F_1 * F_2$
                    $\mid$ $F_1 + F_2$
                    \\\>\>\> $\eventually{\eventEffBar{\ell}{v}}(F)$
                    $\mid$ $\always{\eventEffBar{\ell}{v}}(F)$
                    $\mid$ $\cancelEv{\eventEffBar{\ell}{v}}$
                    $\mid$ $\remove{\eventEffBar{\ell}{v}}$\\[4pt]
  Event Labels \>$\eventEffBar{\ell}{v}$  \> $::=$  \> $\ell\langle \overline{v} \rangle$ \quad where $\overline{v}$ is a tuple of base values\\[4pt]
  Time Unit \>$u$  \> $::=$  \> $r \mid n \mid \kw{ms} \mid \ldots$
\end{tabbing}
\caption{Grammar of \lang Types and Effects}
\label{fig:types-effects}
\Description{Grammar of Types and Effects}
\end{figure}
\lang's type-and-effect system answers \emph{when}-questions about a
render-based program: when can a state change cascade into another, when
are event-handlers registered and removed, when does a feedback loop appear, and
which effects fire on the first render.

\lang has two base effects, the state change effect \stch x and the event effect
\eventEffBar{\ell}{v}.
The state change effect signifies that a state's setter function is called
so its variable $x$ may change.
An event effect corresponds to an event in the environment \lang is embedded
in.
On the web this can be a keypress, a mouse movement, a network request.

The primary vehicle of timing analysis is the temporal next modality
$\Next{N}{}$ whose superscript $N$ counts time units, letting effects describe
not only \emph{what} may happen but \emph{when}.
Units include renders ($r$), network requests ($n$), wall-clock
milliseconds ($\m{ms}$), and any unit the host environment cares to
introduce.
The syntax for effects is based upon the next operator in temporal
logic~\cite{Pnueli77SFCS}.

The event handler modalities \eventually{\eventEffBar{\ell}{v}}~"eventually" and
\always{\eventEffBar{\ell}{v}}~"always" describe when an event handler will be
registered.
The superscripts given to \textit{eventually} and \textit{always} are the name
of the event being bound to: a click, a key press, a promise's resolution or error.
The modalities \cancelEv{\eventEffBar{\ell}{v}}~"cancel" and
\remove{\eventEffBar{\ell}{v}}~"remove" balance with these to ensure no stale
event handlers remain after each render.
\lang's type-and-effect system creates a unified abstraction over both event
handlers and pending events via these modalities.

Effects move between two positions in \lang.
The arrow type in \lang $(\tyArrow{\tau_1}{\tau_2}[F])$ is extended by a latent
effect annotation F.
This annotation is the effect that may occur when the arrow's closure is
executed.
The typing of an expression \tyDerivExpr{\Gamma}{e}{\tau}{F}
reads ``under $\Gamma$, expression $e$ has type $\tau$ and produces effect $F$.''
We say that when an expression is wrapped in a closure as a function body then
the latent effect of the expression is captured into the type system, and when a
function is called then its effect is released.
%

Sequential effects are denoted with $*$ and branches in programs are denoted
with $+$.
The LHS and RHS of a sequencing operator denote which effect will be put on
\lang's event queue first (in $F_1 * F_2$, $F_1$ happens first).
This means that in general effects are not commutative across $*$.

Effects are introduced to the \lang type-and-effect system through type
annotations on built-in functions.
Most syntax constructs combine the effects of their children using the
sequencing and branching constructors.
The built-ins that introduce effects are setters, network calls, and the
event-layer primitives \kw{bind}, \kw{once}, \kw{cancel}, and
\kw{remove}.
%

\begin{center}
  \ensuremath{
  \inferrule*[left=(app)]{
    \tyDerivExpr{\Gamma}{e_1}{(\tyArrow{\tau_1}{\tau_2}[F])}{F_1}\\
    \tyDerivExpr{\Gamma}{e_2}{\tau_1}{F_2}
  }{
    \tyDerivExpr{\Gamma}{e_1\;e_2}{\tau_2}{F_1 * F_2 * F}
  }
  \inferrule*[left=(fn)]{
    \tyDerivExpr{\Gamma, x :\tau_1}{e}{\tau_2}{F}
  }{
    \tyDerivExpr{\Gamma}{\synFn{x}{e}}{(\tyArrow{\tau_1}{\tau_2}[F])}{\cdot}
  }}
\end{center}
Both \rulename{app} and \rulename{fn} are the typical function introduction
and application rules, except extended with our effect system.
The application's effect combines three pieces via the sequencing combinator
$*$: $F_1$ from evaluating $e_1$, $F_2$ from evaluating $e_2$, and the closure's
latent effect $F$.
In the function introduction rule \rulename{fn} we see how the effect of the
function body moves from the latent position in its typing judgement to the
latent position in the arrow type.
A closure expression has no effect on its own so the closure expression has the
empty effect.

Time units appear as a superscript on the next modality $\Next{}{}$ and
only identical units collapse, so
$\Next{1r}{(\Next{1r}{F})}$ can be written as $\Next{2r}{F}$.
We deliberately keep units in \lang incompatible: a network request is more taxing
than one extra render, and some milliseconds of delay matter for perceived
responsiveness but have no performance implication.
Nesting modalities sequences in time: $\Next{1r}(\Next{1n}{F})$ kicks off a
network request during the next render and $F$ happens after it returns.
Sequencing and branching place each child effect relative to a
shared ``zero time'': in $(\Next{1r}{F_1} * \Next{1n}{F_2})$,
$F_1$ runs whenever the next render occurs, and $F_2$
runs once a network request resolves.
$(\Next{1r}{F_1} * \Next{1r}{F_2})$ can be simplified to $\Next{1r}{(F_1 * F_2)}$.
New time units are introduced by built-in functions.
In a real-world scenario programmers could type-alias their own units to make
library \api{}s read better.

Below we give the type for a possible function \kw{asyncCompute}, a function
that takes two closures: one to run on success and one on failure.
The function returns unit immediately and after one unit of computation time,
will either succeed or fail.

\noindent\begin{minipage}{\linewidth}
\begin{lstlisting}[numbers=none]
asyncCompute : $\forall F1, F2$. ($\tyArrow{\kw{int}}{\tyUnit}[F1]$) -> ($\tyArrow{\kw{int}}{\tyUnit}[F2]$) -> unit
             | $\Next{1u}$(comp[suc] + comp[err])
             * $\eventually{\text{comp[suc]}}$($F1$ * $\remove{\text{comp[err]}}$) * $\eventually{\text{comp[err]}}$($F2$ * $\remove{\text{comp[suc]}}$)
\end{lstlisting}
\end{minipage}
%
%
%
The effect of \kw{asyncCompute} shows our unknown unit of computation as "1u" in
the superscript of a next modality.
After one unit of time we can see that either the \kw{comp[suc]} event or the \kw{comp[err]}
event is promised to fire.
At the time of function call, two single-shot event handler registrations are
put onto the event queue.
They are bound to the \texttt{comp[suc]} success or \texttt{comp[err]} error events.
On success the \texttt{comp[suc]} event will fire so the first argument will be
called, its effect F1 will happen, and all handlers for \texttt{comp[err]} will
be removed.
The reverse goes for the \kw{comp[err]} event.
Through analysis of \lang's effect for \texttt{asyncCompute} we can see
(1) when each event will occur, (2) every event is handled, and (3) all handlers
are cleaned up.
%
%

\emph{Immediate, cascading, and full effects.}
\lang distinguishes three kinds of effect.
The \emph{immediate} effect of a function is its own latent side effect.

A variable's \emph{cascading} effect describes what watching effect blocks
will do when that variable changes.
For example, $\synOn{x}{\setter{y}\;\m{addOne}}$ contributes $\Next{1r}{\stch y}$
to the cascading effect of $x$.
When $x$ changes the block fires (during the next render), the setter enqueues a
closure, and one render later $y$ may change.
Cascading effects live in $\Delta$ alongside the variable's dependencies,
written $\deltaEntry{x}{\m{deps}}{\m{cascade}}$.
State variables have no dependencies; let-bound variables do.
To determine a variable's dependencies we parameterize \lang over a
given dataflow function \df.
\df takes an expression as input, and returns a set of variable names that may
alter the value of the expression when they change.
The simplest instantiation of \df is just the set of free variables.

The \emph{full} effect of a variable or function collects every effect
that may be triggered, transitively, by a single mutation or event firing.
It is recovered post-typecheck by a graph expansion over $\Delta$ that
tracks visited nodes to mark loops.
For example, given
$\deltaEntry{x}{}{\Next{1r}{\stch z}},\;
\deltaEntry{y}{z}{\cdot},\;
\deltaEntry{z}{}{\Next{1r}{\stch x}}$
--- ``$x$ triggers a state change in $z$ one render later, $y$ depends on
$z$ so it is recomputed whenever $z$ changes, and $z$ may mutate $x$ after
one render'' --- the full effect of $x$ changing is
$\Next{1r}{(\stch{z} * \stch{y} * \Next{1r}{\m{loop}[x]})}$,
where $\m{loop}[x]$ marks the point at which the expansion revisits $x$.
Loops are surfaced rather than swept away.
Under a reactive paradigm a cycle in the cascading graph is usually a bug.
Preservation (the metatheory section) makes the full effect honest, any
well-typed program that steps continues to type with an effect bounded by
the original under subeffecting, so the full effect computed at type-check
time is a genuine over-approximation of what runs.

\emph{Subeffecting.}
The relation $F \le F'$ used in the rules above is \lang's subeffecting
relation.
$F'$ is the more informative side, it makes more commitments about
what may happen than $F$ does.
Under this reading $+$ acts as a meet: a branch over-approximation can
be replaced by either side, so $F_1 + F_2 \le F_i$.
Dually, $*$ acts as a join, with $F \le F_1 * F_2$ whenever $F$ refines
either side.
The next modality $\Next{N}{}$ is monotone and obeys $\Next{0}{F} = F$ and
$\Next{N_1}{}\Next{N_2}{} = \Next{N_1+N_2}{}$, the graded-monad laws.
Subeffecting on $\stch{x}$, $\eventEffBar{\ell}{v}$, and the event-guarded
constructs is similarly familiar: the empty effect sits below each, and
the body of $\always{}$ and $\eventually{}$ is monotone.
The full rules, including the dual side for $\stch{x}$ and the four
event-guarded leaves, are in Appendix.

\subsection{Declaration-Level Typing Rules}
\label{sec:decl-typing}
%

%
The declaration typing derivation is
\tyDerivDecl{\Sigma}{\Delta}{\Gamma_1}{p}{\Gamma_2},
where $\Sigma$ is the program signature, which maps each component name
to its code, effect environment, and typing environment.
$\Delta$ is the current component's effect environment; it maps variables
to their dependencies and cascading effect $\deltaEntry{x}{\m{deps}}{\m{cascade}}$.
We treat $\Delta$ as given, the typing rules check that effects in the program
are subeffects of the effects given in $\Delta$.
The difficulty in finding a $\Delta$ candidate for a program is finding a
minimal cascading effect.
Typing rules enforce that each variable's cascading entry in $\Delta$ is consistent
using the "causation" judgement
\causesderiv{\Delta}{x}{F}.
Our prototype inference/typechecker recovers $\Delta$ programmatically (\S\ref{sec:examples}).


\begin{center}
\ensuremath{
  \inferrule*[left=(x causes F)]{
    \deltaEntry{x}{\ldots}{F'} \in \Delta
    \and
    F \le F'
  }{
    \causesderiv{\Delta}{x}{F}
  }
  \qquad
  \inferrule*[left=(x causes y)]{
    \deltaEntry{y}{\ldots,x,\ldots}{F'} \in \Delta
  }{
    \causesderiv{\Delta}{x}{\stch y}
  }
}
\end{center}

\causesderiv{\Delta}{x}{F} reads ``a change in $x$ is a sufficient cause for the
effect $F$, given the cascading entries in $\Delta$.''
The first rule is the base case: $x$ causes $F$ directly when $\Delta$
records a cascade $F'$ for $x$ that is at least as informative as $F$
under subeffecting.
The second is the dependency case: $x$ causes $\stch y$ when $y$'s
dependencies include $x$, because if $x$ changes $y$ is recomputed, which
is itself a state-change-like event.

\begin{figure}[t]
\begin{center}
\Rule{ty state decl}{
    \tyDerivExpr{\Gamma}{e}{\tau}{\cdot}\\
    \deltaEntry{x}{}{F} \in \Delta\\
    \Gamma' = \Gamma, x:\tau, \setter{x}: (\tySetter[\tau]{x})
  }{
    \tyDerivDecl
      {\Sigma}{\Gamma}{\Delta}
      {\synState{x}{e}}
      {\Gamma'}
  }
\Rule{ty let decl}{
    \tyDerivExpr{\Gamma}{e}{\tau}{\cdot} \\
    \overline{y} = df(e) \\
    \deltaEntry{x}{\overline y}{F} \in \Delta
  }{
    \tyDerivDecl{\Sigma}{\Gamma}{\Delta}{\synCompLet{x}{e}}{\Gamma, x: \tau}
  }
\Rule{ty on decl}{
    \tyDerivExpr{\Gamma}{e}{\tyUnit}{F} \\
    \forall i \in [n], \causesderiv{\Delta}{x_i}{F}
  }{
    \tyDerivDecl{\Sigma}{\Gamma}{\Delta}{\synOn{x_1, x_2, \ldots, x_n}{e}}{\Gamma}
  }
\Rule{ty subcomp decl}{
  \overline{\causesderiv{\Delta}{x_i}{\stch y.z_i}}\\
  \causesderiv{\Delta}{y.return}{\stch y}\\
  A: (\tinyComp{A}(\overline{z_i:\tau_i}):\tau_r\;\{p\}, \Delta_A, \Gamma_A) \in \Sigma\\
  \overline{x_i: \tau_i}\in \Gamma \\
  \overline{v_\Delta} = fv(\Delta_A)\\
  \Delta_A\overline{[y.v_{\Delta i}/v_{\Delta i}]} \subseteq \Delta\\
  \overline{v_\Gamma} = fv(\Gamma_A)\\
  }{
    \tyDerivDecl{\Sigma}{\Gamma}{\Delta}{
      \synSubComp{y}{A(x_1, x_2, \ldots, x_n)}
    }{\Gamma, \Gamma_A\overline{[y.v_{\Gamma i}/v_{\Gamma i}]}, y: \tau_r}
  }
\end{center}
\caption{Declaration-level typing rules.}
\label{fig:ty-decl}
\Description{The declaration-level typing rules for \lang: state, let, on, and subcomponent declarations.}
\end{figure}


\rulename{ty state decl} (Figure~\ref{fig:ty-decl}) requires the default
expression $e$ to have no effect.
State variables have no dependencies, so their entry in $\Delta$ looks like
$\deltaEntry{x}{}{F}$.
The cascade $F$ records what effect blocks watching $x$ will do when $x$
changes.
%
\rulename{ty let decl} similarly looks up $x$'s entry, but with
dependencies $\overline{y}$ computed by \df.
A let-bound variable is assumed to change whenever any dependency does.


\rulename{ty on decl} types the body $e$ with some effect $F$, then checks
that each watched variable $x_i$ ``causes'' $F$ --- i.e., $F$ is
accounted for by $x_i$'s cascading entry in $\Delta$.
This is the static side of \lang's preservation story: the cascade promised in
$\Delta$ at type-check accounts for the effect block's body.


\rulename{ty subcomp decl} appears to do a lot but is mostly managing
$\alpha$-renaming of variables. This rule retrieves the subcomponent's entry
from $\Sigma$, then enforces variable aliases.
We enforce a timing guarantee that the supplied arguments to $A$ are
dependencies that change $A$'s parameters in the same render.
We also do the same for $A$'s return value back to the variable bound to the
result of $A$.
This means that arguments and return values are passed as normal with no
event-queue or timing interaction.
We enforce that $A$'s effect environment $\Delta_A$ is a sub-environment of
the current $\Delta$ (with some renaming).
$A$'s typing environment $\Gamma_A$ is not required of the caller; instead its
renamed copy is added to the output so the subcomponent's internal bindings
become accessible as $y.v$ in the parent.
All declaration-level variables in subcomponents are assumed to be unique,
so we can use $y.$ as a unique prefix for the variables of this subcomponent
without clashing with another instance of $A$ in this scope.

\subsection{Expression-Level Typing Rules}

Earlier we reviewed the function application rule, we now apply three other
important expression typing rules, shown in Figure~\ref{fig:ty-expr}.

\begin{figure}[t]
\begin{center}
\Rule{seq}{
    \tyDerivExpr{\Gamma}{e_1}{\tau_1}{F_1} \\
    \tyDerivExpr{\Gamma}{e_2}{\tau_2}{F_2}
  }{
    \tyDerivExpr{\Gamma}{e_1;e_2}{\tau_2}{F_1 * F_2}
  }
\Rule{branch}{
    \tyDerivExpr{\Gamma}{e_1}{\tyBool}{F_1} \\
    \tyDerivExpr{\Gamma}{e_2}{\tau}{F_2}\\
    \tyDerivExpr{\Gamma}{e_3}{\tau}{F_3}
  }{
    \tyDerivExpr{\Gamma}{\synIf{e_1}{e_2}{e_3}}{\tau}{F_1 * (F_2 + F_3)}
  }
\Rule{ty bind}{
    \tyDerivExpr{\Gamma}{e}{(\tyArrow{\tau}{\tyUnit}[F])}{F_e}\\
    \tau = \Sigma_E(\eventEffBar{\ell}{v})
  }{
    \tyDerivExpr
      {\Gamma}
      {\kw{bind}\;\eventEffBar{\ell}{v}\;e}
      {\tyUnit}
      {F_e * \always{\eventEffBar{\ell}{v}}(F)}
  }
\Rule{ty once}{
    \tyDerivExpr{\Gamma}{e}{(\tyArrow{\tau}{\tyUnit}[F])}{F_e}\\
    \tau = \Sigma_E(\eventEffBar{\ell}{v})
  }{
    \tyDerivExpr
      {\Gamma}
      {\kw{once}\;\eventEffBar{\ell}{v}\;e}
      {\tyUnit}
      {F_e * \eventually{\eventEffBar{\ell}{v}}(F)}
  }
\Rule{ty cancel}{ }{
    \tyDerivExpr
      {\Gamma}
      {\kw{cancel}\;\eventEffBar{\ell}{v}}
      {\tyUnit}
      {\cancelEv{\eventEffBar{\ell}{v}}}
  }
\Rule{ty remove}{ }{
    \tyDerivExpr
      {\Gamma}
      {\kw{remove}\;\eventEffBar{\ell}{v}}
      {\tyUnit}
      {\remove{\eventEffBar{\ell}{v}}}
  }
\end{center}
\caption{Selected expression-level typing rules, including the four event
primitives.}
\label{fig:ty-expr}
\Description{Expression-level typing rules for \lang: sequencing, branching, and the bind, once, cancel, and remove event primitives.}
\end{figure}


The return type of $e_1$ is discarded and the return type of $e_2$ is returned.
The effects of $e_1$ and $e_2$ are sequenced with $*$.
\rulename{branch} sequences the effect of the condition with a $+$ joining the
effects of the branches.
Exactly one of $F_2$ or $F_3$ may happen.

Shown in Figure~\ref{fig:ty-expr}, the \kw{bind} and \kw{once} rules require
their argument to be a closure
whose latent effect $F$ becomes the body of the event-guarded modality,
and whose parameter type matches the payload type
$\Sigma_E(\eventEffBar{\ell}{v})$ carried by the event.
The mode tag ($\always{}$ versus $\eventually{}$) matches the
listener tag the semantics installs in the listener map $\mathcal{L}$.
The \kw{cancel} and \kw{remove} rules have no premise: the event label is
provided syntactically, and the semantics treats both operations as
unconditional, so the effect annotation describes an \emph{attempt} at
teardown that may be a no-op at runtime.

\subsection{Effect Polymorphism}
\label{subsec:polymorphism}
\lang's prototype type-and-effect inference algorithm implements effect
polymorphism through a basic Hindley-Milner style type inference algorithm used
for effects.
For \lang this means that both functions and components are allowed to
parameterize effect variables that can be filled in at function call or at
subcomponent declaration.
%
\lang's formal type-and-effect system does not implement this feature for
simplicity of the type system.
Because \lang does not feature any form of recursion we do not have to handle
effect-polymorphic recursion, so this makes our implementation easier.
It is appropriate to use an effect in \lang when a function or component takes a
function with arbitrary effect as a parameter.
\begin{lstlisting}[numbers=none]
  both : $\forall F.$ $\tyArrow{(\tyArrow{\kw{int}}{\tyUnit}[F]) \to (\kw{int} \times \kw{int})}{\tyUnit}[F*F]$
  comp TextInput<F> (init: string, onChange: $\tyArrow{\kw{string}}{\tyUnit}[F]$) : html {...}
\end{lstlisting}
Here we see a utility function \kw{both} that applies its first
argument to both elements of an integer pair.
The effect produced is $F*F$ because the function parameter is called twice.
The second small example is the type of a utility component TextInput, which one
can imagine is given a typical \react implementation of a text input element.
An initial value for the input is supplied, and when the input is changed by the
user, onChange would be called.
When $TextInput$ is used as a subcomponent $F$ can be inferred as a concrete
effect, and \lang propagates the effect of the function given for \kw{onChange}.
Due to the way the subcomponent rule ensures that subcomponent effect
environments ($\Delta_A$) are contained within their parent environments
($\Delta$) this does not cause any namespace clashes.
In examples, we freely use effect polymorphism to prevent code duplication and
improve readability.

\subsection{Post-Typecheck Analysis}
\label{subsec:post-analysis}

Once type-and-effect inference finishes, effects in both $\Delta$ and \html can
be viewed as graph nodes for standard graph analysis.


As mentioned earlier, \lang can automatically do loop detection, analyze the
performance of event handlers, ensure stale event handlers are not left
behind, and analyze the performance of an app's initial render.
Each of these algorithms starts with identifying key expressions in a program
and calculating the full effect of those expressions.

To turn an effect into a directional graph we view each effect as its
own node in the graph, and use the effect environment as the list of edges in the graph.
We can view the dependency array as a list of incoming edges, and view the
cascading effect as the outgoing edges.
The timing and event annotations on next, eventually, and always modalities can
be viewed as weights in this graph.
Earlier we outlined how to calculate the full effect of changing a variable, but
this process can be used on the effect of any expression.
The full effect is recovered by a simple graph expansion: if a node $n$ is a
descendant of its own expansion we place a \m{loop}[$n$] node and do not expand
again.
This is how we earlier recovered the full effect
$\Next{1r}{(\stch{z} * \stch{y} * \Next{1r}{\m{loop}[x]})}$ of $x$.
%

\emph{Effect Summary}
We can make the full effect of a variable more readable to a programmer by
hiding state change effects corresponding to variables in child subcomponents'
scopes.
Calculate the full effect of an expression, hide all variables that start
with a "$y.$" prefix, collapse all next modalities' superscripts so that
\Next{1r}{\Next{1r}{F}} becomes \Next{2r}{F}.

\emph{Loop detection}
A common bug in declarative reactive programs is inter-render loops caused by
two or more effect blocks triggering each other's effects each time they run.
To the uninitiated this may sound as simple as not having "on x do {setY ...}; on y
do {setX ...}" however this problem goes deeper.
The reason effect blocks exist is to synchronize with external state, so effect
blocks are where declarative reactive programs interface with programs of a
different paradigm.
When effect blocks only synchronize by either retrieving or pushing updates but
not both then there is no problem.
Bugs can arise when two way synchronization occurs between the client and
server-side.

Take a case as simple as keeping user data up to date on the client side.
A user's data may be updated by the current application or externally in a separate
instance of the same app.
This means a common pattern is to have a web socket listening for user profile
updates so the app can update proactively to reflect changes.
This socket should only update the client-side data though; if receiving data on
the client side causes an update to be pushed to server-side then the client no
longer controls whether it will receive that push back to itself, and a loop
could form.
With \lang we can give external \api{s} temporal effects so \lang's loop
detection can catch these cases.


%

\emph{Event Checking Analysis}
Every type of event fires with a different distribution over time.
Keyboard presses happen in short bursts, mouse movements are frequent and
continual, network requests only resolve a single time.
Knowing these distributions, we can warn programmers of performance intensive
event handlers placed on key events.

\emph{Event handlers cleaned up}
Given the full effect of a variable, walk the effect's graph and ensure that
all event handlers are matched by a corresponding \kw{remove} modality.
This includes walking both branches of the $+$ constructor separately and
looking inside of the \always{} and \eventually{} modalities.
Using this simple graph walk we can revisit the \kw{asyncCompute} example.
We see, just with \lang's type-and-effect system, that no event
handlers are left after the function executes.

\emph{First Render Analysis}
First render is the part of program behavior users feel.
Initial paint dominates perceived performance, Google's Core Web Vitals now
codify Largest Contentful Paint as a first-class web performance
metric~\cite{GoogleLCP}.

A component's first render is the case where every variable
counts as "changed," so every effect block fires.
This means that all program logic is evaluated at the time few pieces of state
have initial values yet.
In the initial render the application has not had time to fetch up-to-date data
from the server, and has not even been able to pull data from local cache.
Permission to play sound/video or use the microphone/camera is not established
yet, whether previously given or not.
This means that the initial render starts more asynchronous computations than
any other time in an app's lifecycle while in a completely unpopulated state.
For programmers this can be a difficult mental shift, as usually programs are
read in the context of steady-state execution.
\lang's type system can give programmers a static understanding of the first
paint's timing separate from the steady-state behavior.

\section{Metatheory}
\label{sec:metatheory}

We present a proof of preservation for \lang.
Our preservation proof's main guarantee can be summarized as follows:
"when a variable changes, a sub-effect of its cascade given in $\Delta$ will
appear in the next housekeeping phase."
To prove this we rely upon an instrumented semantics, which is extended with
helpful side conditions and employs the use of a trace of update queues rather
than a single update queue.
The trace's main function is to couple event queues with the effect that
caused them.
Given $\Delta$ has an entry $\deltaEntry{x}{y}{F}$, then preservation tells
us that when $x$ changes we should see that $y$ had changed.
If there is an effect block bound to $x$ then we should see it produce an update
queue that fulfills some subeffect of $F$.

As a brief aside, the main extension to the language we require is that
\kw{bind} and \kw{once} expressions are annotated with their listener's effect.
This information is already available in the typing of a program so we just also
require it to be in the instrumented semantics.
We use this to extend the listener mapping to be from an event label key to a
tuple of $(c, m, F)$ closure $c$, modality $m$, and $F$ the effect of $c$ as
determined by the type system.
This helps greatly when ensuring that event handlers have the effect they are
typed to have.


Returning to traces: the records that make up traces are
\begin{itemize}
  \item \semPfTraceArgs{\overline x}{\overline v_1}{\overline v_2}
  records which of the main component's arguments changed this render.
  \item \semPfTraceLet{x}{\overline{y}} records a let binding and which
  variables it was affected by.
  \item \semPfTraceOn{\overline{x}}{U} records that a list of variables
  $\overline x$ changed between this render and last render and triggered the
  event queue $U$ to be added to the global event queue.
  \item \semPfTraceReturn{x} records that the variable $x$ was returned by a
  component.
  \item \semPfTraceExt{(\eventEffBar{\ell_1}{v_{11}},\, v_{21}), \ldots,
        (\eventEffBar{\ell_n}{v_{1n}},\, v_{2n})} marks the receipt of a
        batch of external event firings from the scheduler $\mathcal{E}$ and
        delimits a block of per-event processing records that follow.
  \item \semPfTraceFire{(\eventEffBar{\ell}{v_1},\, v_2)} records an
        event firing is waiting to be dispatched.
  \item \semPfTraceFireCxl{(\eventEffBar{\ell}{v_1},\, v_2)} records that a
        firing was suppressed by the cancellation multiset
        $\mathcal{C}$.
  \item \semPfTraceFireSuc{(\eventEffBar{\ell}{v_1},\, v_2)}
        {(U_1, F_1), \ldots, (U_k, F_k)} records that a firing was
        successfully dispatched: for each registered listener $i$, the listener
        body produced update queue $U_i$, and the record pairs $U_i$ with the
        effect $F_i$ stored in the listener map $\mathcal{L}$ at registration
        time.  Recording $F_i$ in the trace is what allows the post-hoc Full
        Effect analysis to walk listener-registered effects through the
        operational trace.
\end{itemize}

The instrumented component semantic judgement is a well-typed configuration.
It contains typing environments $\Sigma, \Delta, \Gamma_s$.
$\Sigma$ and $\Delta$ have already been explained, $\Gamma_s$ is the entire
typing environment of the main component that is returned from type checking.
$\mathcal S$ is a state tuple containing $(X_1, X_2,  V_s,  \mathcal{L},
\mathcal{C})$.
$X_1$ and $X_2$ are the sets of variables that were modified during the
housekeeping phases before the previous render and the most recent render respectively.
$\mathcal L$ is extended in the way we discussed earlier.
$V_s$ and $\mathcal C$ are unchanged from the regular semantics.
$A$ is the name of the main component.
$\kw{status}$ is either \kw{waiting} or \kw{rendered}.
$T_1$ is the already processed part of the trace produced by the most recent render.
$T_2$ is the trace that is yet to be processed.
$$\confPfTop
  {\mathcal{S}}
  {A}{\overline{x=v}}
  {status}{T_1}{T_2}$$




We now give the theorem statement for Preservation of Configurations and then
describe relevant judgements directly after.
The theorem formally states that given a well-typed program with main
component A, at each step of the housekeeping phase (1) the effect environment
and set of states/arguments modified in the previous housekeeping phase explain
the trace produced by the most recent render (2) the already-processed section
of the trace explains each variable that is marked as having been changed for
next render.

\begin{theorem}[Preservation for Configurations] \;
\label{thm:preservation-configurations}
Let $\mathcal{K} = \confPfTop{\mathcal{S}}{A}{\overline{x{=}v}}{status$_1$}{T_1}{T_2}$ with
$\mathcal{S} = (X_1, X_2, V_{s1}, \mathcal{L}, \mathcal{C})$, and suppose $\mathcal{K}$ is well-typed, i.e.
\[
\begin{array}{l}
\tyDerivComp{\Sigma}{\tinyComp{A}}{(\Delta, \Gamma_s)} 
\qquad
\justifies{\Delta, X_1}{T_1,T_2}
  \qquad \justifiesVarSet{\Delta,X_1}{T_1}{X_2} \\[2pt]
\valTyEnvDeriv{\Gamma_s}{V_{s1}}
  \qquad \overline{\valTyDeriv{v}{\tau}} 
\qquad
\text{for every } (c,m,F) \in \mathcal{L}(\eventEffBar{\ell}{v_1}):\
  \valTyDeriv{c}{\tyArrow{\Sigma_E(\eventEffBar{\ell}{v_1})}{\tyUnit}[F]} \\[2pt]
\text{for every firing } (\eventEffBar{\ell}{v_1},v_2)
  \text{ in } \mathcal{E} \text{ or any } \semPfTraceFire{} \text{ record of } T_2:\
  \valTyDeriv{v_2}{\Sigma_E(\eventEffBar{\ell}{v_1})}.
\end{array}
\]
If $\mathcal{K} \to \mathcal{K}'$ with
$\mathcal{K}' = \confPfTop{\mathcal{S}'}{A}{\overline{x{=}v'}}{status$_2$}{T_3}{T_4}$ and
$\mathcal{S}' = (X_3, X_4, V_{s2}, \mathcal{L}', \mathcal{C}')$,
then the analogous well-typedness conditions hold of $\mathcal{K}'$.
\end{theorem}

\emph{Component typing}
$\tyDerivComp{\Sigma}{\tinyComp{A}}{(\Delta,\Gamma_s)}$ states that under the program
signature $\Sigma$, the component $A$ is typeable under effect environment $\Delta$ and
typing environment $\Gamma_s$.

\emph{Trace justification} $\justifies{\Delta,X}{T_1, T_2}$ states that for each entry
in the trace we can point to a variable in $X$ and an effect in $\Delta$ that justify it.
%
%
$\Delta$ contains effects that the type system expects to happen after variables
change and $X$ is the list of variables that changed, so this judgement requires
that $T_1,T_2$ reflects the type-and-effect system's allowed effects.

\emph{Trace-driven variable-change set} $\justifiesVarSet{\Delta,X}{T_1}{X'}$
This judgement states that not only should $\Delta,X$ justify $T_1$ but all three
together should be able to justify each variable that is marked as having
changed for the next render.
%

\emph{Value well-typedness} $\valTyEnvDeriv{\Gamma}{V}$ signifies that all variable-value mappings in $V$ are
well-typed under the variable-type mappings in $\Gamma$.
$\valTyDeriv{v}{\tau}$ states that value $v$ is well-typed under $\tau$.

\emph{Listener-map well-formedness} This requirement on the listener map
tells us the closures stored in the map are well-typed under the arrow type
$\tyArrow{\Sigma_E(\eventEffBar{\ell}{v})}{\tyUnit}[F]$.
The input to the arrow must be of the type that is carried by the event the
closure is bound to.
So for a closure bound to the event \eventEffBar{\ell}{v} the value carried by
that event firing will have type $\Sigma_E(\eventEffBar{\ell}{v})$.
%

\emph{Payload typing for in-flight firings} Here we require that the values
associated with in-flight events are of the correct type
$\Sigma_E(\eventEffBar{\ell}{v})$.

We do not constrain the scheduler or give guarantees about scheduled events
being on time.
This is by the nature of the scheduler being external, we must assume the
scheduler is well-behaved.
There is no constant unit conversion from milliseconds to renders or any other
units, so any well-timed property on the scheduler would follow by
tautology.

\usetikzlibrary{shapes.geometric}

\section{Implementation and Evaluation}\label{sec:examples}

We implement a prototype type-and-effect inference algorithm for \lang
written in Haskell, and accompany it with the suite of analysis algorithms (Section~\ref{subsec:post-analysis})
to explore how \lang could give programmers more information about their programs
before runtime.
Our inference algorithm requires minimal type annotations, only needing function
and component arguments to be type-and-effect annotated.
The effects of all programs in the paper are automatically inferred by the \lang{} inference
checker.

To demonstrate how \lang's type-and-effect system can help find timing bugs
before runtime, we show a focused example of the username input on a
signup form for a website.
%
%
Our example will show a text input field where the user will write their desired
username.
The username will then be sent to the server over the network and the client
will receive a boolean response of whether the username is available or taken.
By dynamically checking the availability we make our form feel more responsive
than if we required the user to validate all inputs at once when attempting to
submit the form.

\begin{lstlisting}
comp UsernameInput () : (bool, html) {
  state username,       setUsername       default "";
  state availableNames, setAvailableNames default new Set();
  state status, setStatus default "idle";
  comp slowUsername = Debounce(username);

  let handleUsernameSuc = $\lambda$(name, isAvailable).
    if isAvailable then setAvailableNames($\lambda$s. s.add(name)) else ()
  let handleUsernameErr = $\lambda$_. setStatus($\lambda$_. "error");
  on slowUsername do {
    if availableNames.has(slowUsername) then
      setStatus($\lambda$_. "idle")
    else (
      setStatus($\lambda$_. "checking");
      fetchUsernameCheck(slowUsername, handleUsernameSuc, handleUsernameErr))};
  return (availableNames.has(username), /* html */)}
\end{lstlisting}

In lines 2-4 we declare our state variables.
We need to track the user's in-progress username, so we store that in the
\kw{username} state string.
\kw{availableNames} is a set we use to hold username candidates marked
as available by the server.
\kw{status} will store the network request status of our username
checks.
On line 5 we declare a debounced version of \kw{username}, this will be helpful
to slow down the number of network requests to our username availability \api
endpoint.
%
Internally the debounce schedules a timeout to fire its slow update; we refer to
that timeout event by the label \eventEff{db}{} in the effects below.

On lines 7 and 8 we declare a handler for when a username check request
succeeds.
We receive a tuple of the checked name and a bool of whether it is available or
not.
If available we add the name to the availableNames set, and if not we do nothing.
There are two failure states for a username check: first the network
request itself could fail and second the name could be unavailable.
We handled the second case already but we handle the network failure on line 9
by setting the status to "error".

Whenever the user stops typing in the text input field for the username, the
debounce will resolve, \kw{slowUsername} will update to its new value, and the
on-block on line 10 will execute.
If we already have checked this username and it's available then we can set the
network status to idle (lines 11-12).
Otherwise we set our \kw{status} to "checking" on line 14 and initiate a new \api
call to check availability on line 15.
Here, \kw{fetchUsernameCheck} has a similar type to \kw{asyncCompute} from
\S\ref{sec:types}, except its events are \eventEff{req}{check,suc} and
\eventEff{req}{check,err}.

Finally we return (1) a bool of whether the currently typed username is available and
(2) the html of the username input which shows username availability,
network request status, and an input.
\begin{center}
  \lstinline|<input id="name" value={username} onChange={$\lambda$e. setUsername($\lambda$_. e.target.value)}/>|
\end{center}
Now we can see what \lang tells us about this program.
%

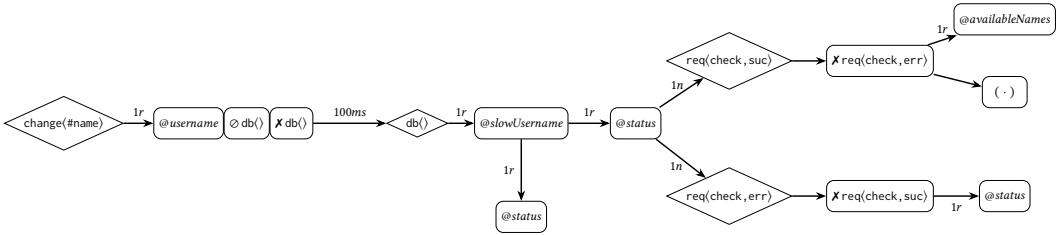
\begin{center}
\resizebox{\linewidth}{!}{%
\begin{tikzpicture}[>={Stealth[length=2mm]},
    every node/.style={font=\scriptsize},
    nd/.style={draw, rounded corners, inner sep=3pt, minimum height=6mm},
    ev/.style={draw, diamond, aspect=2.2, inner sep=1pt, font=\scriptsize},
    cas/.style={->, thick}]
  \node[ev] (chg) at (0,4)      {$\eventEff{change}{\#name}$};
  \node[nd] (un)  at (2.4,4)    {$\stch{\textit{username}}$};
  \node[nd, anchor=west] (can) at (un.east)  {$\cancelEv{\eventEff{db}{}}$};
  \node[nd, anchor=west] (rem) at (can.east) {$\remove{\eventEff{db}{}}$};
  \node[ev] (db)  at (6.8,4)    {$\eventEff{db}{}$};
  \node[nd] (su)  at (8.8,4)     {$\stch{\textit{slowUsername}}$};
  \node[nd] (st1) at (8.8,2.2)   {$\stch{\textit{status}}$};
  \node[nd] (st2) at (11.0,4)     {$\stch{\textit{status}}$};
  \node[ev] (rqs) at (12.8,5.2)   {$\eventEff{req}{check,suc}$};
  \node[ev] (rqe) at (12.8,2.6)   {$\eventEff{req}{check,err}$};
  \node[nd] (remerr) at (15.7,5.2) {$\remove{\eventEff{req}{check,err}}$};
  \node[nd] (an)  at (18.1,6)   {$\stch{\textit{availableNames}}$};
  \node[nd] (cdot) at (18.1,4.6) {\;\;$(\;\cdot\;)$\;\;};
  \node[nd] (ustrem) at (15.7,2.6) {$\remove{\eventEff{req}{check,suc}}$};
  \node[nd] (ust) at (18.1,2.6) {$\stch{\textit{status}}$};
  \draw[cas] (chg) -- node[above]{$1r$} (un);
  \draw[cas] (rem) -- node[above]{$100ms$} (db);
  \draw[cas] (db)  -- node[above]{$1r$} (su);
  \draw[cas] (su)  -- node[left]{$1r$} (st1);
  \draw[cas] (su)  -- node[above]{$1r$} (st2);
  \draw[cas] (st2) -- node[above,pos=0.4]{$1n$} (rqs);
  \draw[cas] (st2) -- node[below,pos=0.4]{$1n$} (rqe);
  \draw[cas] (rqs) -- (remerr);
  \draw[cas] (remerr) -- node[above]{$1r$} (an);
  \draw[cas] (remerr) -- (cdot);
  \draw[cas] (rqe) -- (ustrem);
  \draw[cas] (ustrem) -- node[below]{$1r$} (ust);
\end{tikzpicture}}
\end{center}


Manual inspection of the full effect of our text input changing shows us a bug
in our program.
At every point in the network request we update \textit{status} except for when
the request is successful.
Immediately when \textit{slowUsername} is changed, both branches of our on-block
update the status.
Then after the network request the error case updates the status, but the
success case does not.
This means that on success the loading spinner would confusingly continue to
spin while telling the user that their username is available and they don't have
anything to wait for.
This can be fixed by updating the \textit{status} in the
handleUsernameSuc closure.

\lstinline|$\lambda$(name, isAvailable). (setStatus($\lambda$_."idle");  if isAvailable...)|

Looking at the types also gives the realization that we're allowing multiple
network requests to share the same \textit{status} variable.
%
%
This means that we have created a race condition on the status of multiple
network requests.
In a traditional setting like \react the solution would use pointers but
in the stripped down setting of \lang we slow down.
We can bind our username check status to check if there is already a network request in progress and do nothing if so.
\begin{lstlisting}
  on slowUsername, status do {
    if availableNames.has(slowUsername) || status == "checking" then () else ...
\end{lstlisting}
This prevents us from kicking off more than one network request at a time and
removes the race condition.

The last problem \lang flags is via loop detection.
%
%
The full effect analysis automatically flags a loop in our code, so we have to
do manual analysis to see if this is okay or not.
%

Now our analysis.
When availableNames \emph{does not include} slowUsername and status == "idle" then we
want to check the username.
This case leads us to the \kw{else} branch which checks the username.
The case in which availableNames \emph{includes} slowUsername and status ==
"idle" happens when the user has settled on a username and the username has been
checked.
This case leads to the \kw{then} branch which does nothing.
The case in which availableNames \emph{does not include} slowUsername and status
== "error" brings us to the \kw{else} branch and triggers another check.
This is desirable because it's a built-in retry mechanism for the network error
failure case.
So \lang tells us there is a temporal loop in our code, and there is, but it is
acting as our error-handling case so we can leave the code as is.


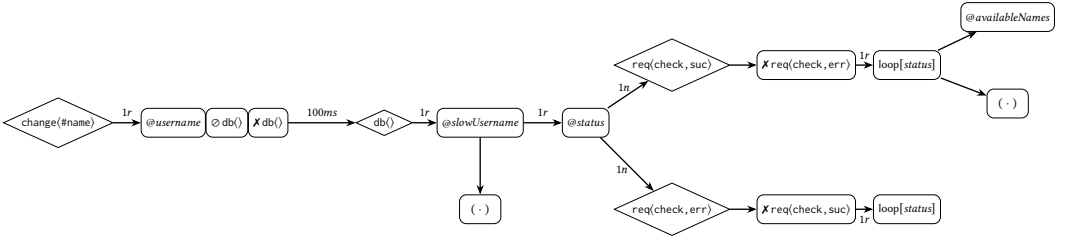
\begin{center}
\resizebox{\linewidth}{!}{%
\begin{tikzpicture}[>={Stealth[length=2mm]},
    every node/.style={font=\scriptsize},
    nd/.style={draw, rounded corners, inner sep=3pt, minimum height=6mm},
    ev/.style={draw, diamond, aspect=2.2, inner sep=1pt, font=\scriptsize},
    cas/.style={->, thick}]
  \node[ev] (chg) at (0,4)      {$\eventEff{change}{\#name}$};
  \node[nd] (un)  at (2.4,4)    {$\stch{\textit{username}}$};
  \node[nd, anchor=west] (can) at (un.east)  {$\cancelEv{\eventEff{db}{}}$};
  \node[nd, anchor=west] (rem) at (can.east) {$\remove{\eventEff{db}{}}$};
  \node[ev] (db)  at (6.8,4)    {$\eventEff{db}{}$};
  \node[nd] (su)  at (8.8,4)     {$\stch{\textit{slowUsername}}$};
  \node[nd] (st1) at (8.8,2.2)   {\;\;$(\;\cdot\;)$\;\;};
  \node[nd] (st2) at (11.0,4)     {$\stch{\textit{status}}$};
  \node[ev] (rqs) at (12.8,5.2)   {$\eventEff{req}{check,suc}$};
  \node[ev] (rqe) at (12.8,2.2)   {$\eventEff{req}{check,err}$};
  \node[nd] (remerr) at (15.6,5.2) {$\remove{\eventEff{req}{check,err}}$};
  \node[nd] (sts) at (17.7,5.2) {$\m{loop}[\textit{status}]$};
  \node[nd] (an)  at (19.8,6.2) {$\stch{\textit{availableNames}}$};
  \node[nd] (cdot) at (19.8,4.4) {\;\;$(\;\cdot\;)$\;\;};
  \node[nd] (ustrem) at (15.6,2.2) {$\remove{\eventEff{req}{check,suc}}$};
  \node[nd] (ust) at (17.7,2.2) {$\m{loop}[\textit{status}]$};
  \draw[cas] (chg) -- node[above]{$1r$} (un);
  \draw[cas] (rem) -- node[above]{$100ms$} (db);
  \draw[cas] (db)  -- node[above]{$1r$} (su);
  \draw[cas] (su)  -- node[left]{} (st1);
  \draw[cas] (su)  -- node[above]{$1r$} (st2);
  \draw[cas] (st2) -- node[above,pos=0.4]{$1n$} (rqs);
  \draw[cas] (st2) -- node[below,pos=0.4]{$1n$} (rqe);
  \draw[cas] (rqs) -- node[above]{} (remerr);
  \draw[cas] (remerr) -- node[above]{$1r$} (sts);
  \draw[cas] (sts) -- (an);
  \draw[cas] (sts) -- (cdot);
  \draw[cas] (rqe) -- (ustrem);
  \draw[cas] (ustrem) -- node[below]{$1r$} (ust);
\end{tikzpicture}}
\end{center}

\section{Related Work}
\label{sec:related}

\emph{React, and GUI Programming}

Functional reactive programming (FRP) began with
Fran~\cite{elliottFunctionalReactiveAnimation1997}, which modeled time-varying
\emph{behaviors} and discrete \emph{events} as first-class values.
This value-centric view of change was carried onto the web by
Flapjax~\cite{Meyerovich09OOPSLA} and FrTime~\cite{Cooper06ESOP}, which build
dataflow graphs that automatically propagate updates, and the lineage continues
through ultrametric and modal accounts of higher-order
FRP~\cite{krishnaswamiUltrametricSemanticsReactive2011}.
Elm~\cite{czaplickiElmConcurrentFRP2012} specialized FRP for GUIs and
popularized the model--view--update loop that later shaped Redux and the React
ecosystem.
Demetrescu et al.~\cite{demetrescuReactiveImperativeProgramming2011} bring
dataflow constraints to imperative C/C++.


Most closely related are formal semantics for React itself.
$\lambda_{\text{react}}$~\cite{madsenSemanticsEssenceReact2020} gives a
semantics for legacy class-based React, while
React-tRace~\cite{leeReacttRaceSemanticsUnderstanding2025} models the modern
hooks API, including the subtle ordering of effects and state updates.
%
%
\lang's semantics differ from Lee et al.'s to focus on the timing analysis.
We attempt to make \lang's temporal dependency graphs widely applicable to
reactive programming while looking to \react to see what the problem space is like.
React-tRace fully encodes the component lifecycle of mounting and unmounting, as
well as the highly asynchronous order of setter, effect, and subcomponent
evaluation.
If we were to directly use the tRace semantics we would have to introduce an
entire layer of dynamic naming that would clutter our type system's readability.
We would like to implement \lang's type-and-effect system for React tRace in the
future if possible.

Adjoint Reactive GUI programming~\cite{graulundAdjointReactiveGUI2021} is the nearest
type-theoretic neighbor on the GUI side.
It introduces a modality for receiving a value from an event and reasons about
the time needed to obtain a value of a given type.
However that work uses a coeffect discipline over values, \lang uses an effect
discipline over computations.

A separate strand integrates session types with GUI and reactive programming.
Event-Driven Multiparty Session
Actors~\cite{fowlerEventDrivenMultipartySession2017} use flow-sensitive effect
typing for actors participating in several sessions at once, and
Model-View-Update-Communicate~\cite{fowlerModelViewUpdateCommunicateSessionTypes2019}
folds session typing into an Elm-style architecture.
%

\emph{Timing Analysis}
\lang's graded next modality $\Next{}{}$ is most directly inspired by temporal
session types.
Das et al.~\cite{dasParallelComplexityAnalysis2018} add timing modalities to
session types to bound the parallel complexity of message-passing programs.
The next operator itself is syntactically the next operator of temporal
logic~\cite{Pnueli77SFCS}.
ESTEREL~\cite{Berry92SCP} and LUSTRE~\cite{Halbwachs91IEEE} compute over
discrete logical instants and explicit clocks, building on the dataflow
semantics of Kahn process networks~\cite{Kahn74IFIP}.
\lang's render is a comparable discrete tick, but \lang targets the
asynchronous, externally-scheduled world of UI events and over-approximates
branching as \emph{may}-effects rather than assuming a single fixed synchronous
clock.
Modal FRP calculi such as Simply RaTT~\cite{bahrSimplyRaTTFitchstyle2019} use a
Fitch-style $\Next{}{}$ modality to stage reactive computation across time steps
and to rule out space leaks.
%
%
Ahman~\cite{ahmanWhenProgramsHave2023} studies termination and resource locking
for non-declarative functional reactive programs.
%
The actor-reactor model~\cite{vandenvonderTacklingAwkwardSquad2017}
catalogs the pitfalls --- glitches, stale reads, unintended feedback --- that
reactive languages routinely fall into; \lang attacks the same hazards
statically, surfacing cyclic cascades and stale-listener bugs as type-level
effects.
Time-indexed hardware systems suggest that \lang's render model could generalize
beyond software: the ephemeral history
register~\cite{EphemeralHistoryRegister2004} is a circuit-design construct
similar to React's scheduled state.
Timeline types~\cite{nigamModularHardwareDesign2023} index a type system by
when signals arrive, albeit for statically-scheduled hardware rather than
dynamically-scheduled UI events.

\emph{Effect systems}
\lang's next modality $\Next{}{}$ is a graded modality in the sense of graded monads~\cite{Katsumata14}, and
bounded linear logic~\cite{girardBoundedLinearLogic1992} and its generalization
to resource semirings~\cite{BrunelGMZ14,ghicaBoundedLinearTypes2014}, which annotate types
with quantities drawn from an algebraic structure.
%
The link between graded temporal modalities and effects has been
developed recently by Sekiyama et al., who give an algebraic axiomatization of
temporal effects for higher-order recursive
programs~\cite{sekiyamaAlgebraicTemporalEffects2025} and a dependent temporal
type-and-effect system with answer-effect modification for delimited
continuations~\cite{sekiyamaTemporalVerificationAnswerEffect2023}.
%
%
\lang's effects also build on the broader effect-system tradition:
algebraic effects and handlers~\cite{kammarHandlersAction2013a} and the
effect-polymorphic language Frank~\cite{lindleyBeBe2017} inform our treatment of
effect polymorphism.
%
Two systems for typing parallel and concurrent structure are
methodologically close: Muller's static prediction of parallel computation
graphs~\cite{mullerStaticPredictionParallel2022} and the graph types for ADTs
with futures of Rinaldi et al.~\cite{rinaldi_pipelines_2024}, whose vertex
structures describe dependencies among futures.
We took heavy inspiration from this system while designing \lang, and we
anticipate needing to incorporate vertex structures in the future if we are to
allow dynamic mounting and unmounting of subcomponents.
More broadly, \lang's temporal dependency graphs connect to a long line of work
on dependency analysis, from the program dependence
graph~\cite{ferranteProgramDependenceGraph1987} to semantics-based accounts of
dependency~\cite{cousotAbstractSemanticDependency2019} and their roots in secure
information flow~\cite{denningCertificationProgramsSecure1977}.
\lang differs in tracking \emph{when} dependencies fire across renders rather
than only which computations depend on which.

\section{Conclusion}



We have presented \lang, a core calculus for render-based reactive programming whose
type-and-effect system makes the temporal behavior of reactive programs
explicit.
By treating the \emph{render} as the fundamental unit of time and tracking state
changes, event registrations, and scheduled events as graded effects, \lang
turns timing properties that are normally buried in framework runtimes into
static, checkable artifacts.
The temporal next modality $\Next{}{}$ records \emph{when}, in renders or any
declared time unit, an observable update may occur, while the
$\always{}$, $\eventually{}$, $\cancelEv{}$, and $\remove{}$ modalities track
the registration and teardown lifecycle of event listeners.
We gave \lang a two-phase, time-aware operational semantics, proved preservation
of the effect system against an instrumented version of that semantics, and
showed that the inferred effects form \emph{temporal dependency graphs} that
standard graph algorithms can analyze for render cascades, inter-render
loops, stale listeners, high-frequency handlers, and first-render timing.
A prototype inference algorithm checks every example in the paper, and our
signup-form case study shows \lang statically surfacing a stuck loading state, a
request race condition, and an update loop.

%
\lang does not yet model the mounting and unmounting of
components that Lee et al.~\cite{leeReacttRaceSemanticsUnderstanding2025}
formalize.
Lifting this restriction would require a dependent type system built
on the vertex structures of Rinaldi et al.~\cite{rinaldi_pipelines_2024}.
This would let \lang's analysis apply directly to pre-existing render-based
frameworks.
%
%
%
By making the \emph{when} of reactive computation a first-class,
statically-analyzable property, \lang gives user interface programmers a
principled way to understand the responsiveness of their applications at compile
time and to spend less effort on manual temporal reasoning.

\bibliographystyle{ACM-Reference-Format}
\bibliography{bib}


\begin{thebibliography}{42}


\ifx \showCODEN    \undefined \def \showCODEN     #1{\unskip}     \fi
\ifx \showISBNx    \undefined \def \showISBNx     #1{\unskip}     \fi
\ifx \showISBNxiii \undefined \def \showISBNxiii  #1{\unskip}     \fi
\ifx \showISSN     \undefined \def \showISSN      #1{\unskip}     \fi
\ifx \showLCCN     \undefined \def \showLCCN      #1{\unskip}     \fi
\ifx \shownote     \undefined \def \shownote      #1{#1}          \fi
\ifx \showarticletitle \undefined \def \showarticletitle #1{#1}   \fi
\ifx \showURL      \undefined \def \showURL       {\relax}        \fi
\providecommand\bibfield[2]{#2}
\providecommand\bibinfo[2]{#2}
\providecommand\natexlab[1]{#1}
\providecommand\showeprint[2][]{arXiv:#2}

\bibitem[Ahman(2023)]%
        {ahmanWhenProgramsHave2023}
\bibfield{author}{\bibinfo{person}{Danel Ahman}.}
  \bibinfo{year}{2023}\natexlab{}.
\newblock \showarticletitle{When {Programs} {Have} to {Watch} {Paint} {Dry}}.
  In \bibinfo{booktitle}{\emph{Foundations of {Software} {Science} and
  {Computation} {Structures}}}, \bibfield{editor}{\bibinfo{person}{Orna
  Kupferman} {and} \bibinfo{person}{Pawel Sobocinski}} (Eds.).
  \bibinfo{publisher}{Springer Nature Switzerland}, \bibinfo{address}{Cham},
  \bibinfo{pages}{1--23}.
\newblock
\showISBNx{978-3-031-30829-1}
\href{https://doi.org/10.1007/978-3-031-30829-1_1}{doi:\nolinkurl{10.1007/978-3-031-30829-1_1}}


\bibitem[Bahr et~al\mbox{.}(2019)]%
        {bahrSimplyRaTTFitchstyle2019}
\bibfield{author}{\bibinfo{person}{Patrick Bahr},
  \bibinfo{person}{Christian~Uldal Graulund}, {and}
  \bibinfo{person}{Rasmus~Ejlers Møgelberg}.} \bibinfo{year}{2019}\natexlab{}.
\newblock \showarticletitle{Simply {RaTT}: a fitch-style modal calculus for
  reactive programming without space leaks}.
\newblock \bibinfo{journal}{\emph{Proceedings of the ACM on Programming
  Languages}} \bibinfo{volume}{3}, \bibinfo{number}{ICFP} (\bibinfo{date}{July}
  \bibinfo{year}{2019}), \bibinfo{pages}{1--27}.
\newblock
\showISSN{2475-1421}
\href{https://doi.org/10.1145/3341713}{doi:\nolinkurl{10.1145/3341713}}


\bibitem[Banerjee et~al\mbox{.}(2013)]%
        {banerjee_graphical_2013}
\bibfield{author}{\bibinfo{person}{Ishan Banerjee}, \bibinfo{person}{Bao
  Nguyen}, \bibinfo{person}{Vahid Garousi}, {and} \bibinfo{person}{Atif
  Memon}.} \bibinfo{year}{2013}\natexlab{}.
\newblock \showarticletitle{Graphical user interface ({GUI}) testing:
  Systematic mapping and repository}.
\newblock \bibinfo{journal}{\emph{Information and Software Technology}}
  \bibinfo{volume}{55}, \bibinfo{number}{10} (\bibinfo{date}{Oct.}
  \bibinfo{year}{2013}), \bibinfo{pages}{1679--1694}.
\newblock
\showISSN{0950-5849}
\href{https://doi.org/10.1016/j.infsof.2013.03.004}{doi:\nolinkurl{10.1016/j.infsof.2013.03.004}}


\bibitem[Berry and Gonthier(1992)]%
        {Berry92SCP}
\bibfield{author}{\bibinfo{person}{G\'{e}rard Berry} {and}
  \bibinfo{person}{Georges Gonthier}.} \bibinfo{year}{1992}\natexlab{}.
\newblock \showarticletitle{The ESTEREL synchronous programming language:
  design, semantics, implementation}.
\newblock \bibinfo{journal}{\emph{Sci. Comput. Program.}} \bibinfo{volume}{19},
  \bibinfo{number}{2} (\bibinfo{date}{Nov.} \bibinfo{year}{1992}),
  \bibinfo{pages}{87–152}.
\newblock
\showISSN{0167-6423}
\href{https://doi.org/10.1016/0167-6423(92)90005-V}{doi:\nolinkurl{10.1016/0167-6423(92)90005-V}}


\bibitem[Brunel et~al\mbox{.}(2014)]%
        {BrunelGMZ14}
\bibfield{author}{\bibinfo{person}{Alo{\"{\i}}s Brunel}, \bibinfo{person}{Marco
  Gaboardi}, \bibinfo{person}{Damiano Mazza}, {and} \bibinfo{person}{Steve
  Zdancewic}.} \bibinfo{year}{2014}\natexlab{}.
\newblock \showarticletitle{A Core Quantitative Coeffect Calculus}. In
  \bibinfo{booktitle}{\emph{Programming Languages and Systems - 23rd European
  Symposium on Programming, {ESOP} 2014, Held as Part of the European Joint
  Conferences on Theory and Practice of Software, {ETAPS} 2014, Grenoble,
  France, April 5-13, 2014, Proceedings}} \emph{(\bibinfo{series}{Lecture Notes
  in Computer Science}, Vol.~\bibinfo{volume}{8410})},
  \bibfield{editor}{\bibinfo{person}{Zhong Shao}} (Ed.).
  \bibinfo{publisher}{Springer}, \bibinfo{pages}{351--370}.
\newblock
\href{https://doi.org/10.1007/978-3-642-54833-8\_19}{doi:\nolinkurl{10.1007/978-3-642-54833-8\_19}}


\bibitem[Clark and Team(2017)]%
        {react-fiber}
\bibfield{author}{\bibinfo{person}{Lin Clark} {and} \bibinfo{person}{React
  Team}.} \bibinfo{year}{2017}\natexlab{}.
\newblock \bibinfo{title}{React Fiber Architecture}.
\newblock
  \bibinfo{howpublished}{\url{https://github.com/acdlite/react-fiber-architecture}}.
\newblock


\bibitem[Cooper and Krishnamurthi(2006)]%
        {Cooper06ESOP}
\bibfield{author}{\bibinfo{person}{Gregory~H. Cooper} {and}
  \bibinfo{person}{Shriram Krishnamurthi}.} \bibinfo{year}{2006}\natexlab{}.
\newblock \showarticletitle{Embedding dynamic dataflow in a call-by-value
  language}. In \bibinfo{booktitle}{\emph{Proceedings of the 15th European
  Conference on Programming Languages and Systems}} (Vienna, Austria)
  \emph{(\bibinfo{series}{ESOP'06})}. \bibinfo{publisher}{Springer-Verlag},
  \bibinfo{address}{Berlin, Heidelberg}, \bibinfo{pages}{294–308}.
\newblock
\showISBNx{354033095X}
\href{https://doi.org/10.1007/11693024_20}{doi:\nolinkurl{10.1007/11693024_20}}


\bibitem[Cousot(2019)]%
        {cousotAbstractSemanticDependency2019}
\bibfield{author}{\bibinfo{person}{Patrick Cousot}.}
  \bibinfo{year}{2019}\natexlab{}.
\newblock \showarticletitle{Abstract Semantic Dependency}. In
  \bibinfo{booktitle}{\emph{Static Analysis - 26th International Symposium, SAS
  2019, Porto, Portugal, October 8-11, 2019, Proceedings}}
  \emph{(\bibinfo{series}{Lecture Notes in Computer Science},
  Vol.~\bibinfo{volume}{11822})},
  \bibfield{editor}{\bibinfo{person}{Bor-Yuh~Evan Chang}} (Ed.).
  \bibinfo{publisher}{Springer}, \bibinfo{pages}{389--410}.
\newblock
\href{https://doi.org/10.1007/978-3-030-32304-2\_19}{doi:\nolinkurl{10.1007/978-3-030-32304-2\_19}}


\bibitem[Czaplicki(2012)]%
        {czaplickiElmConcurrentFRP2012}
\bibfield{author}{\bibinfo{person}{Evan Czaplicki}.}
  \bibinfo{year}{2012}\natexlab{}.
\newblock \showarticletitle{Elm : {Concurrent} {FRP} for {Functional} {GUIs}}.
\newblock
\urldef\tempurl%
\url{https://www.semanticscholar.org/paper/Elm-%3A-Concurrent-FRP-for-Functional-GUIs-Czaplicki/1791a8a278b83c54425d7581cb45320feba5f4b0}
\showURL{%
\tempurl}


\bibitem[Das et~al\mbox{.}(2018)]%
        {dasParallelComplexityAnalysis2018}
\bibfield{author}{\bibinfo{person}{Ankush Das}, \bibinfo{person}{Jan Hoffmann},
  {and} \bibinfo{person}{Frank Pfenning}.} \bibinfo{year}{2018}\natexlab{}.
\newblock \showarticletitle{Parallel complexity analysis with temporal session
  types}.
\newblock \bibinfo{journal}{\emph{Proc. ACM Program. Lang.}}
  \bibinfo{volume}{2}, \bibinfo{number}{ICFP} (\bibinfo{date}{July}
  \bibinfo{year}{2018}), \bibinfo{pages}{91:1--91:30}.
\newblock
\href{https://doi.org/10.1145/3236786}{doi:\nolinkurl{10.1145/3236786}}


\bibitem[Demetrescu et~al\mbox{.}(2011)]%
        {demetrescuReactiveImperativeProgramming2011}
\bibfield{author}{\bibinfo{person}{Camil Demetrescu}, \bibinfo{person}{Irene
  Finocchi}, {and} \bibinfo{person}{Andrea Ribichini}.}
  \bibinfo{year}{2011}\natexlab{}.
\newblock \showarticletitle{Reactive imperative programming with dataflow
  constraints}. In \bibinfo{booktitle}{\emph{Proceedings of the 2011 {ACM}
  international conference on {Object} oriented programming systems languages
  and applications}} \emph{(\bibinfo{series}{{OOPSLA} '11})}.
  \bibinfo{publisher}{Association for Computing Machinery},
  \bibinfo{address}{New York, NY, USA}, \bibinfo{pages}{407--426}.
\newblock
\showISBNx{978-1-4503-0940-0}
\href{https://doi.org/10.1145/2048066.2048100}{doi:\nolinkurl{10.1145/2048066.2048100}}


\bibitem[Denning and Denning(1977)]%
        {denningCertificationProgramsSecure1977}
\bibfield{author}{\bibinfo{person}{Dorothy~E. Denning} {and}
  \bibinfo{person}{Peter~J. Denning}.} \bibinfo{year}{1977}\natexlab{}.
\newblock \showarticletitle{Certification of Programs for Secure Information
  Flow}.
\newblock \bibinfo{journal}{\emph{Commun. ACM}} \bibinfo{volume}{20},
  \bibinfo{number}{7} (\bibinfo{year}{1977}), \bibinfo{pages}{504--513}.
\newblock
\href{https://doi.org/10.1145/359636.359712}{doi:\nolinkurl{10.1145/359636.359712}}


\bibitem[Elliott and Hudak(1997)]%
        {elliottFunctionalReactiveAnimation1997}
\bibfield{author}{\bibinfo{person}{Conal Elliott} {and} \bibinfo{person}{Paul
  Hudak}.} \bibinfo{year}{1997}\natexlab{}.
\newblock \showarticletitle{Functional reactive animation}. In
  \bibinfo{booktitle}{\emph{Proceedings of the second {ACM} {SIGPLAN}
  international conference on {Functional} programming}}
  \emph{(\bibinfo{series}{{ICFP} '97})}. \bibinfo{publisher}{Association for
  Computing Machinery}, \bibinfo{address}{New York, NY, USA},
  \bibinfo{pages}{263--273}.
\newblock
\showISBNx{978-0-89791-918-0}
\href{https://doi.org/10.1145/258948.258973}{doi:\nolinkurl{10.1145/258948.258973}}


\bibitem[Ferrante et~al\mbox{.}(1987)]%
        {ferranteProgramDependenceGraph1987}
\bibfield{author}{\bibinfo{person}{Jeanne Ferrante}, \bibinfo{person}{Karl~J.
  Ottenstein}, {and} \bibinfo{person}{Joe~D. Warren}.}
  \bibinfo{year}{1987}\natexlab{}.
\newblock \showarticletitle{The Program Dependence Graph and Its Use in
  Optimization}.
\newblock \bibinfo{journal}{\emph{ACM Transactions on Programming Languages and
  Systems}} \bibinfo{volume}{9}, \bibinfo{number}{3} (\bibinfo{year}{1987}),
  \bibinfo{pages}{319--349}.
\newblock
\href{https://doi.org/10.1145/24039.24041}{doi:\nolinkurl{10.1145/24039.24041}}


\bibitem[Fowler(2019)]%
        {fowlerModelViewUpdateCommunicateSessionTypes2019}
\bibfield{author}{\bibinfo{person}{Simon Fowler}.}
  \bibinfo{year}{2019}\natexlab{}.
\newblock \bibinfo{title}{Model-{View}-{Update}-{Communicate}: {Session}
  {Types} meet the {Elm} {Architecture}}.
\newblock
\href{https://doi.org/10.48550/ARXIV.1910.11108}{doi:\nolinkurl{10.48550/ARXIV.1910.11108}}
\newblock
\shownote{Version Number: 3}.


\bibitem[Fowler and Hu(2017)]%
        {fowlerEventDrivenMultipartySession2017}
\bibfield{author}{\bibinfo{person}{Simon Fowler} {and} \bibinfo{person}{Raymond
  Hu}.} \bibinfo{year}{2017}\natexlab{}.
\newblock \bibinfo{title}{Event-{Driven} {Multiparty} {Session} {Actors}}.
\newblock \bibinfo{howpublished}{Presented at the 11th ACM SIGPLAN Workshop on
  Higher-Order Programming with Effects (HOPE 2023)}.
\newblock
\newblock
\shownote{No proceedings}.


\bibitem[Ghica and Smith(2014)]%
        {ghicaBoundedLinearTypes2014}
\bibfield{author}{\bibinfo{person}{Dan~R. Ghica} {and} \bibinfo{person}{Alex~I.
  Smith}.} \bibinfo{year}{2014}\natexlab{}.
\newblock \showarticletitle{Bounded {Linear} {Types} in a {Resource}
  {Semiring}}. In \bibinfo{booktitle}{\emph{Programming {Languages} and
  {Systems}}}, \bibfield{editor}{\bibinfo{person}{Zhong Shao}} (Ed.).
  \bibinfo{publisher}{Springer}, \bibinfo{address}{Berlin, Heidelberg},
  \bibinfo{pages}{331--350}.
\newblock
\showISBNx{978-3-642-54833-8}
\href{https://doi.org/10.1007/978-3-642-54833-8_18}{doi:\nolinkurl{10.1007/978-3-642-54833-8_18}}


\bibitem[Girard et~al\mbox{.}(1992)]%
        {girardBoundedLinearLogic1992}
\bibfield{author}{\bibinfo{person}{Jean-Yves Girard}, \bibinfo{person}{Andre
  Scedrov}, {and} \bibinfo{person}{Philip~J. Scott}.}
  \bibinfo{year}{1992}\natexlab{}.
\newblock \showarticletitle{Bounded linear logic: a modular approach to
  polynomial-time computability}.
\newblock \bibinfo{journal}{\emph{Theoretical Computer Science}}
  \bibinfo{volume}{97}, \bibinfo{number}{1} (\bibinfo{date}{April}
  \bibinfo{year}{1992}), \bibinfo{pages}{1--66}.
\newblock
\showISSN{0304-3975}
\href{https://doi.org/10.1016/0304-3975(92)90386-T}{doi:\nolinkurl{10.1016/0304-3975(92)90386-T}}


\bibitem[{Google}(2020)]%
        {GoogleLCP}
\bibfield{author}{\bibinfo{person}{{Google}}.} \bibinfo{year}{2020}\natexlab{}.
\newblock \bibinfo{title}{Largest Contentful Paint (LCP)}.
\newblock \bibinfo{howpublished}{web.dev. \url{https://web.dev/articles/lcp}}.
\newblock
\newblock
\shownote{Core Web Vitals metric for perceived load speed; sites should target
  LCP of 2.5 seconds or less.}.


\bibitem[Graulund et~al\mbox{.}(2021)]%
        {graulundAdjointReactiveGUI2021}
\bibfield{author}{\bibinfo{person}{Christian~Uldal Graulund},
  \bibinfo{person}{Dmitrij Szamozvancev}, {and} \bibinfo{person}{Neel
  Krishnaswami}.} \bibinfo{year}{2021}\natexlab{}.
\newblock \showarticletitle{Adjoint {Reactive} {GUI} {Programming}}.
\newblock In \bibinfo{booktitle}{\emph{Foundations of {Software} {Science} and
  {Computation} {Structures}}}, \bibfield{editor}{\bibinfo{person}{Stefan
  Kiefer} {and} \bibinfo{person}{Christine Tasson}} (Eds.).
  Vol.~\bibinfo{volume}{12650}. \bibinfo{publisher}{Springer International
  Publishing}, \bibinfo{address}{Cham}, \bibinfo{pages}{289--309}.
\newblock
\showISBNx{978-3-030-71994-4 978-3-030-71995-1}
\href{https://doi.org/10.1007/978-3-030-71995-1_15}{doi:\nolinkurl{10.1007/978-3-030-71995-1_15}}
\newblock
\shownote{Series Title: Lecture Notes in Computer Science}.


\bibitem[Halbwachs et~al\mbox{.}(1991)]%
        {Halbwachs91IEEE}
\bibfield{author}{\bibinfo{person}{N. Halbwachs}, \bibinfo{person}{P. Caspi},
  \bibinfo{person}{P. Raymond}, {and} \bibinfo{person}{D. Pilaud}.}
  \bibinfo{year}{1991}\natexlab{}.
\newblock \showarticletitle{The synchronous data flow programming language
  LUSTRE}.
\newblock \bibinfo{journal}{\emph{Proc. IEEE}} \bibinfo{volume}{79},
  \bibinfo{number}{9} (\bibinfo{year}{1991}), \bibinfo{pages}{1305--1320}.
\newblock
\href{https://doi.org/10.1109/5.97300}{doi:\nolinkurl{10.1109/5.97300}}


\bibitem[Kahn(1974)]%
        {Kahn74IFIP}
\bibfield{author}{\bibinfo{person}{Gilles Kahn}.}
  \bibinfo{year}{1974}\natexlab{}.
\newblock \showarticletitle{The Semantics of a Simple Language for Parallel
  Programming}. In \bibinfo{booktitle}{\emph{IFIP Congress}}.
\newblock
\urldef\tempurl%
\url{https://api.semanticscholar.org/CorpusID:18030506}
\showURL{%
\tempurl}


\bibitem[Kammar et~al\mbox{.}(2013)]%
        {kammarHandlersAction2013a}
\bibfield{author}{\bibinfo{person}{Ohad Kammar}, \bibinfo{person}{Sam Lindley},
  {and} \bibinfo{person}{Nicolas Oury}.} \bibinfo{year}{2013}\natexlab{}.
\newblock \showarticletitle{Handlers in action}. In
  \bibinfo{booktitle}{\emph{Proceedings of the 18th {ACM} {SIGPLAN}
  international conference on {Functional} programming}}
  \emph{(\bibinfo{series}{{ICFP} '13})}. \bibinfo{publisher}{Association for
  Computing Machinery}, \bibinfo{address}{New York, NY, USA},
  \bibinfo{pages}{145--158}.
\newblock
\showISBNx{978-1-4503-2326-0}
\href{https://doi.org/10.1145/2500365.2500590}{doi:\nolinkurl{10.1145/2500365.2500590}}


\bibitem[Katsumata(2014)]%
        {Katsumata14}
\bibfield{author}{\bibinfo{person}{Shin{-}ya Katsumata}.}
  \bibinfo{year}{2014}\natexlab{}.
\newblock \showarticletitle{Parametric effect monads and semantics of effect
  systems}. In \bibinfo{booktitle}{\emph{The 41st Annual {ACM} {SIGPLAN-SIGACT}
  Symposium on Principles of Programming Languages, {POPL} '14, San Diego, CA,
  USA, January 20-21, 2014}}, \bibfield{editor}{\bibinfo{person}{Suresh
  Jagannathan} {and} \bibinfo{person}{Peter Sewell}} (Eds.).
  \bibinfo{publisher}{{ACM}}, \bibinfo{pages}{633--646}.
\newblock
\href{https://doi.org/10.1145/2535838.2535846}{doi:\nolinkurl{10.1145/2535838.2535846}}


\bibitem[Khourshid(2022)]%
        {khourshid2022goodbye}
\bibfield{author}{\bibinfo{person}{David Khourshid}.}
  \bibinfo{year}{2022}\natexlab{}.
\newblock \bibinfo{title}{Goodbye, useEffect}.
\newblock \bibinfo{howpublished}{Reactathon 2022 Conference Talk, Real World
  React (YouTube)}.
\newblock
\urldef\tempurl%
\url{https://www.youtube.com/watch?v=HPoC-k7Rxwo}
\showURL{%
\tempurl}


\bibitem[Krishnaswami and Benton(2011)]%
        {krishnaswamiUltrametricSemanticsReactive2011}
\bibfield{author}{\bibinfo{person}{Neelakantan~R. Krishnaswami} {and}
  \bibinfo{person}{Nick Benton}.} \bibinfo{year}{2011}\natexlab{}.
\newblock \showarticletitle{Ultrametric {Semantics} of {Reactive} {Programs}}.
  In \bibinfo{booktitle}{\emph{2011 {IEEE} 26th {Annual} {Symposium} on {Logic}
  in {Computer} {Science}}}. \bibinfo{publisher}{IEEE},
  \bibinfo{address}{Toronto, ON, Canada}, \bibinfo{pages}{257--266}.
\newblock
\showISBNx{978-1-4577-0451-2}
\href{https://doi.org/10.1109/LICS.2011.38}{doi:\nolinkurl{10.1109/LICS.2011.38}}


\bibitem[Lee et~al\mbox{.}(2025)]%
        {leeReacttRaceSemanticsUnderstanding2025}
\bibfield{author}{\bibinfo{person}{Jay Lee}, \bibinfo{person}{Joongwon Ahn},
  {and} \bibinfo{person}{Kwangkeun Yi}.} \bibinfo{year}{2025}\natexlab{}.
\newblock \bibinfo{title}{React-{tRace}: {A} {Semantics} for {Understanding}
  {React} {Hooks}}.
\newblock
\href{https://doi.org/10.1145/3763067}{doi:\nolinkurl{10.1145/3763067}}
\newblock
\shownote{arXiv:2507.05234 [cs]}.


\bibitem[Lindley et~al\mbox{.}(2017)]%
        {lindleyBeBe2017}
\bibfield{author}{\bibinfo{person}{Sam Lindley}, \bibinfo{person}{Conor
  McBride}, {and} \bibinfo{person}{Craig McLaughlin}.}
  \bibinfo{year}{2017}\natexlab{}.
\newblock \bibinfo{title}{Do be do be do}.
\newblock
\href{https://doi.org/10.48550/arXiv.1611.09259}{doi:\nolinkurl{10.48550/arXiv.1611.09259}}
\newblock
\shownote{arXiv:1611.09259 [cs]}.


\bibitem[Madsen et~al\mbox{.}(2020)]%
        {madsenSemanticsEssenceReact2020}
\bibfield{author}{\bibinfo{person}{Magnus Madsen}, \bibinfo{person}{Ondřej
  Lhoták}, {and} \bibinfo{person}{Frank Tip}.}
  \bibinfo{year}{2020}\natexlab{}.
\newblock \showarticletitle{A {Semantics} for the {Essence} of {React}}. In
  \bibinfo{booktitle}{\emph{{LIPIcs}, {Volume} 166, {ECOOP} 2020}},
  Vol.~\bibinfo{volume}{166}. \bibinfo{publisher}{Schloss Dagstuhl –
  Leibniz-Zentrum für Informatik}, \bibinfo{pages}{12:1--12:26}.
\newblock
\showISBNx{978-3-95977-154-2}
\showISSN{1868-8969}
\href{https://doi.org/10.4230/LIPICS.ECOOP.2020.12}{doi:\nolinkurl{10.4230/LIPICS.ECOOP.2020.12}}
\newblock
\shownote{Artwork Size: 26 pages, 814986 bytes Medium: application/pdf}.


\bibitem[Memon(2002)]%
        {memon_gui_2002}
\bibfield{author}{\bibinfo{person}{A.M. Memon}.}
  \bibinfo{year}{2002}\natexlab{}.
\newblock \showarticletitle{{GUI} testing: pitfalls and process}.
\newblock \bibinfo{journal}{\emph{Computer}} \bibinfo{volume}{35},
  \bibinfo{number}{8} (\bibinfo{date}{Aug.} \bibinfo{year}{2002}),
  \bibinfo{pages}{87--88}.
\newblock
\showISSN{1558-0814}
\href{https://doi.org/10.1109/MC.2002.1023795}{doi:\nolinkurl{10.1109/MC.2002.1023795}}


\bibitem[{Meta Platforms}(2013)]%
        {ReactMeta}
\bibfield{author}{\bibinfo{person}{{Meta Platforms}}.}
  \bibinfo{year}{2013}\natexlab{}.
\newblock \bibinfo{title}{React: A JavaScript Library for Building User
  Interfaces}.
\newblock \bibinfo{howpublished}{\url{https://react.dev}}.
\newblock


\bibitem[{Meta Platforms}(2024)]%
        {react_compiler}
\bibfield{author}{\bibinfo{person}{{Meta Platforms}}.}
  \bibinfo{year}{2024}\natexlab{}.
\newblock \bibinfo{title}{React Compiler}.
\newblock
  \bibinfo{howpublished}{\url{https://react.dev/learn/react-compiler/introduction}}.
\newblock


\bibitem[Meyerovich et~al\mbox{.}(2009)]%
        {Meyerovich09OOPSLA}
\bibfield{author}{\bibinfo{person}{Leo~A. Meyerovich}, \bibinfo{person}{Arjun
  Guha}, \bibinfo{person}{Jacob Baskin}, \bibinfo{person}{Gregory~H. Cooper},
  \bibinfo{person}{Michael Greenberg}, \bibinfo{person}{Aleks Bromfield}, {and}
  \bibinfo{person}{Shriram Krishnamurthi}.} \bibinfo{year}{2009}\natexlab{}.
\newblock \showarticletitle{Flapjax: a programming language for Ajax
  applications}. In \bibinfo{booktitle}{\emph{Proceedings of the 24th ACM
  SIGPLAN Conference on Object Oriented Programming Systems Languages and
  Applications}} (Orlando, Florida, USA) \emph{(\bibinfo{series}{OOPSLA '09})}.
  \bibinfo{publisher}{Association for Computing Machinery},
  \bibinfo{address}{New York, NY, USA}, \bibinfo{pages}{1–20}.
\newblock
\showISBNx{9781605587660}
\href{https://doi.org/10.1145/1640089.1640091}{doi:\nolinkurl{10.1145/1640089.1640091}}


\bibitem[Muller(2022)]%
        {mullerStaticPredictionParallel2022}
\bibfield{author}{\bibinfo{person}{Stefan~K. Muller}.}
  \bibinfo{year}{2022}\natexlab{}.
\newblock \showarticletitle{Static prediction of parallel computation graphs}.
\newblock \bibinfo{journal}{\emph{Proc. ACM Program. Lang.}}
  \bibinfo{volume}{6}, \bibinfo{number}{POPL} (\bibinfo{date}{Jan.}
  \bibinfo{year}{2022}), \bibinfo{pages}{46:1--46:31}.
\newblock
\href{https://doi.org/10.1145/3498708}{doi:\nolinkurl{10.1145/3498708}}


\bibitem[Nigam et~al\mbox{.}(2023)]%
        {nigamModularHardwareDesign2023}
\bibfield{author}{\bibinfo{person}{Rachit Nigam},
  \bibinfo{person}{Pedro~Henrique Azevedo~de Amorim}, {and}
  \bibinfo{person}{Adrian Sampson}.} \bibinfo{year}{2023}\natexlab{}.
\newblock \showarticletitle{Modular {Hardware} {Design} with {Timeline}
  {Types}}.
\newblock \bibinfo{journal}{\emph{Proc. ACM Program. Lang.}}
  \bibinfo{volume}{7}, \bibinfo{number}{PLDI} (\bibinfo{date}{June}
  \bibinfo{year}{2023}), \bibinfo{pages}{120:343--120:367}.
\newblock
\href{https://doi.org/10.1145/3591234}{doi:\nolinkurl{10.1145/3591234}}


\bibitem[Pnueli(1977)]%
        {Pnueli77SFCS}
\bibfield{author}{\bibinfo{person}{Amir Pnueli}.}
  \bibinfo{year}{1977}\natexlab{}.
\newblock \showarticletitle{The temporal logic of programs}. In
  \bibinfo{booktitle}{\emph{Proceedings of the 18th Annual Symposium on
  Foundations of Computer Science}} \emph{(\bibinfo{series}{SFCS '77})}.
  \bibinfo{publisher}{IEEE Computer Society}, \bibinfo{address}{USA},
  \bibinfo{pages}{46–57}.
\newblock
\href{https://doi.org/10.1109/SFCS.1977.32}{doi:\nolinkurl{10.1109/SFCS.1977.32}}


\bibitem[Rinaldi et~al\mbox{.}(2024)]%
        {rinaldi_pipelines_2024}
\bibfield{author}{\bibinfo{person}{Francis Rinaldi}, \bibinfo{person}{june
  wunder}, \bibinfo{person}{Arthur Azevedo~de Amorim}, {and}
  \bibinfo{person}{Stefan~K. Muller}.} \bibinfo{year}{2024}\natexlab{}.
\newblock \showarticletitle{Pipelines and Beyond: Graph Types for {ADTs} with
  Futures}.
\newblock \bibinfo{journal}{\emph{Proc. ACM Program. Lang.}}
  \bibinfo{volume}{8}, \bibinfo{number}{POPL} (\bibinfo{date}{Jan.}
  \bibinfo{year}{2024}), \bibinfo{pages}{17:482--17:511}.
\newblock
\href{https://doi.org/10.1145/3632859}{doi:\nolinkurl{10.1145/3632859}}


\bibitem[Rosenband(2004)]%
        {EphemeralHistoryRegister2004}
\bibfield{author}{\bibinfo{person}{Daniel~L. Rosenband}.}
  \bibinfo{year}{2004}\natexlab{}.
\newblock \showarticletitle{The ephemeral history register: flexible scheduling
  for rule-based designs}. In \bibinfo{booktitle}{\emph{Proceedings of the
  {Second} {ACM}/{IEEE} {International} {Conference} on {Formal} {Methods} and
  {Models} for {Co}-{Design}}} \emph{(\bibinfo{series}{{MEMOCODE} '04})}.
  \bibinfo{publisher}{IEEE Computer Society}, \bibinfo{address}{USA},
  \bibinfo{pages}{189--198}.
\newblock
\showISBNx{978-0-7803-8509-2}
\href{https://doi.org/10.1109/MEMCOD.2004.1459853}{doi:\nolinkurl{10.1109/MEMCOD.2004.1459853}}


\bibitem[Sekiyama and Unno(2023)]%
        {sekiyamaTemporalVerificationAnswerEffect2023}
\bibfield{author}{\bibinfo{person}{Taro Sekiyama} {and}
  \bibinfo{person}{Hiroshi Unno}.} \bibinfo{year}{2023}\natexlab{}.
\newblock \showarticletitle{Temporal {Verification} with {Answer}-{Effect}
  {Modification}: {Dependent} {Temporal} {Type}-and-{Effect} {System} with
  {Delimited} {Continuations}}.
\newblock \bibinfo{journal}{\emph{Proc. ACM Program. Lang.}}
  \bibinfo{volume}{7}, \bibinfo{number}{POPL} (\bibinfo{date}{Jan.}
  \bibinfo{year}{2023}), \bibinfo{pages}{71:2079--71:2110}.
\newblock
\href{https://doi.org/10.1145/3571264}{doi:\nolinkurl{10.1145/3571264}}


\bibitem[Sekiyama and Unno(2025)]%
        {sekiyamaAlgebraicTemporalEffects2025}
\bibfield{author}{\bibinfo{person}{Taro Sekiyama} {and}
  \bibinfo{person}{Hiroshi Unno}.} \bibinfo{year}{2025}\natexlab{}.
\newblock \showarticletitle{Algebraic {Temporal} {Effects}: {Temporal}
  {Verification} of {Recursively} {Typed} {Higher}-{Order} {Programs}}.
\newblock \bibinfo{journal}{\emph{Proc. ACM Program. Lang.}}
  \bibinfo{volume}{9}, \bibinfo{number}{POPL} (\bibinfo{date}{Jan.}
  \bibinfo{year}{2025}), \bibinfo{pages}{78:2306--78:2336}.
\newblock
\href{https://doi.org/10.1145/3704914}{doi:\nolinkurl{10.1145/3704914}}


\bibitem[Van~den Vonder et~al\mbox{.}(2017)]%
        {vandenvonderTacklingAwkwardSquad2017}
\bibfield{author}{\bibinfo{person}{Sam Van~den Vonder}, \bibinfo{person}{Joeri
  De~Koster}, \bibinfo{person}{Florian Myter}, {and} \bibinfo{person}{Wolfgang
  De~Meuter}.} \bibinfo{year}{2017}\natexlab{}.
\newblock \showarticletitle{Tackling the awkward squad for reactive
  programming: the actor-reactor model}. In
  \bibinfo{booktitle}{\emph{Proceedings of the 4th {ACM} {SIGPLAN}
  {International} {Workshop} on {Reactive} and {Event}-{Based} {Languages} and
  {Systems}}} \emph{(\bibinfo{series}{{REBLS} 2017})}.
  \bibinfo{publisher}{Association for Computing Machinery},
  \bibinfo{address}{New York, NY, USA}, \bibinfo{pages}{27--33}.
\newblock
\showISBNx{978-1-4503-5515-5}
\href{https://doi.org/10.1145/3141858.3141863}{doi:\nolinkurl{10.1145/3141858.3141863}}


\bibitem[Wan and Hudak(2000)]%
        {Wan00PLDI}
\bibfield{author}{\bibinfo{person}{Zhanyong Wan} {and} \bibinfo{person}{Paul
  Hudak}.} \bibinfo{year}{2000}\natexlab{}.
\newblock \showarticletitle{Functional reactive programming from first
  principles}. In \bibinfo{booktitle}{\emph{Proceedings of the ACM SIGPLAN 2000
  Conference on Programming Language Design and Implementation}} (Vancouver,
  British Columbia, Canada) \emph{(\bibinfo{series}{PLDI '00})}.
  \bibinfo{publisher}{Association for Computing Machinery},
  \bibinfo{address}{New York, NY, USA}, \bibinfo{pages}{242–252}.
\newblock
\showISBNx{1581131992}
\href{https://doi.org/10.1145/349299.349331}{doi:\nolinkurl{10.1145/349299.349331}}


\end{thebibliography}

\appendix
\section{Collected Grammars}
\label{app:grammars}

For reference, we collect here the full grammars of \lang, reproduced from the
main text: the term syntax (Figure~\ref{app:fig:syntax}) and the grammar of
types and effects (Figure~\ref{app:fig:types-effects}).

\begin{figure}[h]
\begin{tabbing}
  Declarations\quad\=$p$ \quad\= $::=$\quad\= \kill
  Components \>$C$   \> $::=$  \> \synComponent{A}{x_1 : \tau_1, x_2 : \tau_2, \ldots, x_n : \tau_n}{\tau}{p_1 \; p_2 \; \ldots \; p_n \; \synReturn{x}}\\[4pt]
  Declarations \>$p$ \> $::=$  \> \synCompLet{x}{e}
                    $\mid$ \synState{x}{e}
                    \\\>\>\>   \synOn{x_1, x_2, \ldots, x_n}{e}
                    $\mid$ \synSubComp{x}{A(x_1, x_2, \ldots, x_n)}\\[4pt]
  Expressions \>$e$  \> $::=$  \> $x$
                    $\mid$ $c \in \{\m{true}, \m{false}, ()\}$
                    $\mid$ $e_1\;e_2$
                    $\mid$ $e_1;\;e_2$
                    $\mid$ \synFn{x}{e}
                    \\\>\>\> \synIf{x}{e_1}{e_2}
                    $\mid$ $(e_1,e_2)$
                    $\mid$ $\kw{fst}\;e$
                    $\mid$ $\kw{snd}\;e$
                    \\\>\>\> $\kw{bind}\;\eventEffBar{\ell}{v}\;e$
                    $\mid$ $\kw{once}\;\eventEffBar{\ell}{v}\;e$
                    $\mid$ $\kw{cancel}\;\eventEffBar{\ell}{v}$
                    $\mid$ $\kw{remove}\;\eventEffBar{\ell}{v}$
\end{tabbing}
\caption{Syntax of \lang Programs (reproduced from Figure~\ref{fig:syntax}).}
\label{app:fig:syntax}
\Description{Syntax of \lang programs}
\end{figure}

\begin{figure}[h]
\begin{tabbing}
  Event Labels\quad\=$\eventEffBar{\ell}{v}$ \quad\= $::=$\quad\= \kill
  Types \>$\tau$  \> $::=$  \> $\tyUnit$
                    $\mid$ $\tyBool$
                    $\mid$ $\tau_1 \times \tau_2$
                    $\mid$ $(\tyArrow{\tau_1}{\tau_2}[F])$\\[4pt]
  Effects \>$F$  \> $::=$  \> $\cdot$
                    $\mid$ $\stch{x}$
                    $\mid$ $\eventEffBar{\ell}{v}$
                    $\mid$ $\Next{Nu}{F}$
                    $\mid$ $F_1 * F_2$
                    $\mid$ $F_1 + F_2$
                    \\\>\>\> $\eventually{\eventEffBar{\ell}{v}}(F)$
                    $\mid$ $\always{\eventEffBar{\ell}{v}}(F)$
                    $\mid$ $\cancelEv{\eventEffBar{\ell}{v}}$
                    $\mid$ $\remove{\eventEffBar{\ell}{v}}$\\[4pt]
  Event Labels \>$\eventEffBar{\ell}{v}$  \> $::=$  \> $\ell\langle \overline{v} \rangle$ \quad where $\overline{v}$ is a tuple of base values\\[4pt]
  Time Unit \>$u$  \> $::=$  \> $r \mid n \mid \kw{ms} \mid \ldots$
\end{tabbing}
\caption{Grammar of \lang Types and Effects (reproduced from Figure~\ref{fig:types-effects}).}
\label{app:fig:types-effects}
\Description{Grammar of Types and Effects}
\end{figure}

\section{Types}

\subsection{Program} \;

These rules establish program-level well-formedness. The judgement
$\Sigma \Rightarrow \tinyComp{A_1}, \ldots, \tinyComp{A_n} \Rightarrow \Sigma'$
checks a sequence of component definitions against a starting signature
$\Sigma$, threading each type-checked component back into the signature so that
later components may refer to earlier ones, and yields the final signature
$\Sigma'$ that maps every component name to its code, effect environment
$\Delta$, and typing environment $\Gamma$.

\begin{placedfigure}
\Rule{wf comp}
  {
    \Sigma
      \Rightarrow \synComponent{A}{\overline{x_i : \tau_i}}{\tau}{p_1,p_2,\ldots, p_n}
      \Rightarrow (\Delta, \Gamma)
  }
  {
    \Sigma
      \Rightarrow \tinyComp{A}
      \Rightarrow \Sigma, A: (\tinyComp{A}, \Delta, \Gamma)
  }

\bigskip

\Rule{wf program}
  {
    \Sigma
      \Rightarrow \synComponent{A_1}{\overline{x_i : \tau_i}}{\tau}{p_1,p_2,\ldots, p_n}
      \Rightarrow (\Delta, \Gamma)
    \\\\
    \Sigma, A_1: (\tinyComp{A_1}, \Delta, \Gamma)
      \Rightarrow \tinyComp{A_2}, \ldots, \tinyComp{A_n}
      \Rightarrow \Sigma'
  }
  {
    \Sigma
      \Rightarrow \tinyComp{A_1}, \ldots, \tinyComp{A_n}
      \Rightarrow \Sigma'
  }
\captionof{figure}{Program-level well-formedness.}
\label{fig:app-ty-program}
\end{placedfigure}

\subsection{Component} \;

Component-level typing checks a single component against the program signature
$\Sigma$. The judgement $\tyDerivComp{\Sigma}{\tinyComp{A}}{(\Delta, \Gamma)}$
says that, given the declared effect environment $\Delta$, the component's
declaration block type-checks and produces the typing environment $\Gamma$ that
binds its state, let, and subcomponent names.

\begin{placedfigure}
\Rule{ty component}
  {
    \tyDerivDecl{\Sigma}{\overline{x_i : \tau_i}}{\Delta}{p_1,p_2,\ldots, p_n}{\Gamma}
  }
  {
    \tyDerivComp
      {\Sigma}
      {\synComponent{A}{\overline{x_i : \tau_i}}{\tau}{p_1,p_2,\ldots, p_n}}
      {(\Delta, \Gamma)}
  }
\captionof{figure}{Component-level typing.}
\label{fig:app-ty-component}
\end{placedfigure}

\subsection{Declaration} \;

Declaration typing threads a typing environment through a component's
declaration block. The judgement $\tyDerivDecl{\Sigma}{\Gamma}{\Delta}{p}{\Gamma'}$
checks a single declaration $p$ under signature $\Sigma$, effect environment
$\Delta$, and incoming environment $\Gamma$, extending it to $\Gamma'$; the
sequencing judgement
$\tyDerivDecls{\Sigma}{\Gamma}{\Delta}{p_1;\ldots;\synReturn{x}}{\tau}$
chains these across a whole block that ends in a return of type $\tau$. Each
rule additionally checks, via the ``causes'' judgement
(Section~\ref{app:causes}), that the effects a declaration performs are
permitted by $\Delta$.

\begin{placedfigure}
\Rule{ty return decl}{
    \causesderiv{\Delta}{x}{\stch return}
}{
    \tyDerivDecls{\Sigma}{\Gamma, x: \tau}{\Delta}{\synReturn{x}}{\tau}
  }



\Rule{declarations}{
    \tyDerivDecl{\Sigma}{\Gamma_1}{\Delta}{p_1}{\Gamma_2}
    \\
    \tyDerivDecls{\Sigma}{\Gamma_2}{\Delta}{p_2;\ldots; p_n; \synReturn{x}}{\tau}
  }
  {
    \tyDerivDecls{\Sigma}{\Gamma_1}{\Delta}{p_1;p_2;\ldots; p_n; \synReturn{x}}{\tau}
  }

\Rule{ty state decl}{
    \tyDerivExpr{\Gamma}{e}{\tau}{\cdot}\\
    \deltaEntry{x}{}{F} \in \Delta\\
    \Gamma' = \Gamma, x:\tau, \setter{x}: (\tySetter[\tau]{x})
  }{
    \tyDerivDecl
      {\Sigma}{\Gamma}{\Delta}
      {\synState{x}{e}}
      {\Gamma'}
  }



\Rule{ty let decl}{
    \tyDerivExpr{\Gamma}{e}{\tau}{\cdot} \\
    \overline{y} = \df(e) \\
    \deltaEntry{x}{\overline y}{F} \in \Delta
  }{
    \tyDerivDecl{\Sigma}{\Gamma}{\Delta}{\synCompLet{x}{e}}{\Gamma, x: \tau}
  }



\Rule{ty on decl}{
    \tyDerivExpr{\Gamma}{e}{\tyUnit}{F} \\
    \forall i \in [n], \causesderiv{\Delta}{x_i}{F}
  }{
    \tyDerivDecl{\Sigma}{\Gamma}{\Delta}{\synOn{x_1, x_2, \ldots, x_n}{e}}{\Gamma}
  }


\Rule{ty subcomp decl}{
  A: (\tinyComp{A}(\overline{z_i:\tau_i}):\tau_r\;\{p\}, \Delta_A, \Gamma_A) \in \Sigma\\
  \overline{\causesderiv{\Delta}{x_i}{\stch y.z_i}}\\
  \causesderiv{\Delta}{y.return}{\stch y}\\
  \overline{x_i: \tau_i}\in \Gamma \\\\
  \overline{v_\Delta} = fv(\Delta_A)\\
  \Delta_A\overline{[y.v_{\Delta i}/v_{\Delta i}]} \subseteq \Delta\\
  \overline{v_\Gamma} = fv(\Gamma_A)\\
  }{
    \tyDerivDecl{\Sigma}{\Gamma}{\Delta}{
      \synSubComp{y}{A(x_1, x_2, \ldots, x_n)}
    }{\Gamma, \Gamma_A\overline{[y.v_{\Gamma i}/v_{\Gamma i}]}, y: \tau_r}
  }
\captionof{figure}{Declaration-level typing rules.}
\label{fig:app-ty-decl}
\end{placedfigure}

\subsection{Expression} \;

Expression typing is a standard type-and-effect system. The judgement
$\tyDerivExpr{\Gamma}{e}{\tau}{F}$ assigns expression $e$ the type $\tau$ and
the effect $F$ describing the update-queue items its evaluation may emit ---
state changes, listener registrations, cancellations, and removals. Pure
expressions carry the empty effect $\cdot$, sequencing and application combine
sub-effects with $*$, branching combines them with $+$, and the event
primitives $\kw{bind}$, $\kw{once}$, $\kw{cancel}$, and $\kw{remove}$ introduce
the event-guarded modalities and teardown effects.

\begin{placedfigure}
\Rule{unit}{
    \and
  }{
    \tyDerivExpr{\Gamma}{()}{\tyUnit}{\cdot}
  }
\quad
\Rule{bool}{
    c \in \{true, false\}
  }{
    \tyDerivExpr{\Gamma}{c}{\tyBool}{\cdot}
  }
\quad
\Rule{base}{
    c \in \alpha
  }{
    \tyDerivExpr{\Gamma}{c}{\alpha}{\cdot}
  }
\quad
\Rule{var}{
  }{
    \tyDerivExpr{\Gamma, x: \tau}{x}{\tau}{\cdot}
  }

\Rule{app}{
    \tyDerivExpr{\Gamma}{e_1}{(\tyArrow{\tau_1}{\tau_2}[F])}{F_1}\\
    \tyDerivExpr{\Gamma}{e_2}{\tau_1}{F_2}
  }{
    \tyDerivExpr{\Gamma}{e_1\;e_2}{\tau_2}{F_1 * F_2 * F}
  }
\quad
\Rule{fn}{
    \tyDerivExpr{\Gamma, x :\tau_1}{e}{\tau_2}{F}
  }{
    \tyDerivExpr{\Gamma}{\synFn{x}{e}}{(\tyArrow{\tau_1}{\tau_2}[F])}{\cdot}
  }

\Rule{seq}{
    \tyDerivExpr{\Gamma}{e_1}{\tau_1}{F_1} \\
    \tyDerivExpr{\Gamma}{e_2}{\tau_2}{F_2}
  }{
    \tyDerivExpr{\Gamma}{e_1;e_2}{\tau_2}{F_1 * F_2}
  }
\quad
\Rule{branch}{
    \tyDerivExpr{\Gamma}{e_1}{\tyBool}{F_1} \\
    \tyDerivExpr{\Gamma}{e_2}{\tau}{F_2}\\
    \tyDerivExpr{\Gamma}{e_3}{\tau}{F_3}
  }{
    \tyDerivExpr{\Gamma}{\synIf{e_1}{e_2}{e_3}}{\tau}{F_1 * (F_2 + F_3)}
  }
\quad
\Rule{prod}{
    \tyDerivExpr{\Gamma}{e_1}{\tau_1}{F_1} \\
    \tyDerivExpr{\Gamma}{e_2}{\tau_2}{F_2}
  }{
    \tyDerivExpr{\Gamma}{(e_1, e_2)}{\tau_1 \times \tau_2}{F_1 * F_2}
  }
\quad
\Rule{fst}{
    \tyDerivExpr{\Gamma}{e_1}{\tau_1 \times \tau_2}{F}
  }{
    \tyDerivExpr{\Gamma}{\kw{fst}\;e_1}{\tau_1}{F}
  }
\quad
\Rule{snd}{
    \tyDerivExpr{\Gamma}{e_1}{\tau_1 \times \tau_2}{F}
  }{
    \tyDerivExpr{\Gamma}{\kw{snd}\;e_1}{\tau_1}{F}
  }
\quad
\Rule{ty bind}{
    \tyDerivExpr{\Gamma}{e}{(\tyArrow{\tau}{\tyUnit}[F])}{F_e}\\
    \tau = \Sigma_E(\eventEffBar{\ell}{v})
  }{
    \tyDerivExpr
      {\Gamma}
      {\kw{bind}\;\eventEffBar{\ell}{v}\;e}
      {\tyUnit}
      {F_e * \always{\eventEffBar{\ell}{v}}(F)}
  }
\quad
\Rule{ty once}{
    \tyDerivExpr{\Gamma}{e}{(\tyArrow{\tau}{\tyUnit}[F])}{F_e}\\
    \tau = \Sigma_E(\eventEffBar{\ell}{v})
  }{
    \tyDerivExpr
      {\Gamma}
      {\kw{once}\;\eventEffBar{\ell}{v}\;e}
      {\tyUnit}
      {F_e * \eventually{\eventEffBar{\ell}{v}}(F)}
  }
\quad
\Rule{ty cancel}{
    \\
  }{
    \tyDerivExpr
      {\Gamma}
      {\kw{cancel}\;\eventEffBar{\ell}{v}}
      {\tyUnit}
      {\cancelEv{\eventEffBar{\ell}{v}}}
  }
\quad
\Rule{ty remove}{
    \\
  }{
    \tyDerivExpr
      {\Gamma}
      {\kw{remove}\;\eventEffBar{\ell}{v}}
      {\tyUnit}
      {\remove{\eventEffBar{\ell}{v}}}
  }
\captionof{figure}{Expression typing rules.}
\label{fig:app-ty-expr}
\end{placedfigure}

\section{Side Conditions}
\subsection{Typing -- ``causes''} \;
\label{app:causes}

The ``causes'' judgement $\causesderiv{\Delta}{x}{F}$ reads ``a change in $x$ is
a sufficient cause for the effect $F$''. It is the side condition that ties the
declared effect environment $\Delta$ to what a declaration may do: a state, let,
effect, or subcomponent declaration type-checks only when every effect it
performs is caused, according to $\Delta$, by a variable it depends on. The
three rules cover the direct case (where $x$'s own entry in $\Delta$ dominates
$F$), the dependency case (a change in $x$ propagates to any $y$ that lists $x$
among its dependencies), and closure under the $*$ combinator.

\begin{placedfigure}
\Rule{x causes F}{
  \deltaEntry{x}{\ldots}{F'} \in \Delta
  \and
  F \le F'
}{
  \causesderiv{\Delta}{x}{F}
}
\quad
\Rule{x causes y}{
  \deltaEntry{y}{\ldots,x,\ldots}{F'} \in \Delta
}{
  \causesderiv{\Delta}{x}{\stch y}
}
\quad
\Rule{x causes $F_1 * F_2$}{
  \causesderiv{\Delta}{x}{F_1'}
  \and
  \causesderiv{\Delta}{x}{F_2'}
  \and
  F_1 \le F_1'
  \and
  F_2 \le F_2'
}{
  \causesderiv{\Delta}{x}{F_1 * F_2}
}
\captionof{figure}{The ``causes'' judgement.}
\label{fig:app-causes}
\end{placedfigure}

\subsection{Sub-Effecting}
\label{app:subeffecting}

Sub-effecting is the subtyping relation $F \le F'$ on effects, read ``$F$ is a
sub-effect of $F'$'': anywhere an effect $F'$ is expected, an $F$ with
$F \le F'$ may be supplied. The rules propagate sub-effecting through the $+$
and $*$ combinators, make the delay modality $\Next{t}{}$ monotone in its body
and splittable and additive in its grade $t$, provide reflexivity and
transitivity, and let the empty effect $\cdot$ be refined into any single leaf
effect (a state change, an event, an event-guarded modality, or a teardown), so
a computation may always be over-approximated by a larger effect.

\begin{placedfigure}
\Rule{se-plus-r}{
  F \le F_1\\
  F \le F_2
}{
  F \le F_1 + F_2
}
\quad
\Rule{se-plus-l1}{ }{
  F_1 + F_2 \le F_1
}
\quad
\Rule{se-plus-l2}{ }{
  F_1 + F_2 \le F_2
}
\quad
\Rule{se-mult}{
  F \le F_{i \in \{1,2\}}
}{
  F \le F_1 * F_2
}
\quad
\Rule{se-zero-l}{
  F \le F'
}{
  \Next{0} F \le {F'}
}
\quad
\Rule{se-zero-r}{
  F \le F'
}{
  F \le \Next{0} {F'}
}
\quad
\Rule{se-delay}{
  (\Next {t - t'} F) \le F' \\
  t \le t'
}{
  \Next {t} {F} \le \Next{t'} {F'}
}
\quad
\Rule{se-split-l}{ }{
  \Next{t_1 + t_2} {F} \le \Next {t_1} \Next {t_2} {F}
}
\quad
\Rule{se-split-r}{ }{
  \Next {t_1} \Next {t_2} {F} \le \Next{t_1 + t_2} {F}
}
\quad
\Rule{se-eq}{
  F = F'
}{
  F \le F'
}
\quad
\Rule{se-trans}{
  F \le F'\\F' \le F''
}{
  F \le F''
}
\quad
\Rule{se-delay-eq}{
  F \le F'
}{
  \Next {t} {F} \le \Next{t} {F'}
}
\quad
\Rule{se-subeffecting}{ }{
  \cdot \le \stch{x}
}
\label{sec:subeffecting-nse}
\quad
\Rule{se-subeffecting-nse}{ }{
  \cdot \le \eventEffBar{\ell}{v}
}

\quad\Rule{se-subeffecting-eventually}{ }{
  \cdot \le \eventually{\eventEffBar{\ell}{v}}(F)
}
\quad
\Rule{se-subeffecting-always}{ }{
  \cdot \le \always{\eventEffBar{\ell}{v}}(F)
}

\quad\Rule{se-eventually-body}{
  F \le F'
}{
  \eventually{\eventEffBar{\ell}{v}}(F) \le \eventually{\eventEffBar{\ell}{v}}(F')
}
\quad
\Rule{se-always-body}{
  F \le F'
}{
  \always{\eventEffBar{\ell}{v}}(F) \le \always{\eventEffBar{\ell}{v}}(F')
}

\quad\Rule{se-subeffecting-cancel}{ }{
  \cdot \le \cancelEv{\eventEffBar{\ell}{v}}
}
\quad
\Rule{se-subeffecting-remove}{ }{
  \cdot \le \remove{\eventEffBar{\ell}{v}}
}
\captionof{figure}{Sub-effecting.}
\label{fig:app-subeffecting}
\end{placedfigure}


\section{Semantics}
\label{app:semantics}



\subsection{External Scheduler}

\;

Three new objects are added to the component-level housekeeping semantics.
\begin{itemize}
  \item $\mathcal{E}$, the \emph{external scheduler}: an external scheduler presents
    an ordered list of \emph{event firings}, each of the form
    $(\ell\langle\overline{v_1}\rangle,\, v_2)$ where
    $\ell\langle\overline{v_1}\rangle$ is the full event label (label
    $\ell$ together with its discriminant value tuple $\overline{v_1}$) and
    $v_2$ is the \emph{payload value} delivered to listeners when this
    firing is dispatched. The same event may occur more than once in
    $\mathcal{E}$, possibly with different payloads, so a single render can
    see an event fire multiple times. New elements may be present in
    $\mathcal{E}$ at the beginning of each housekeeping phase. $\mathcal{E}$
    is treated as an unknown black box: well-typedness of configurations
    requires only that each emitted payload $v_2$ inhabit
    $\Sigma_E(\ell\langle\overline{v_1}\rangle)$, the data type declared for
    that event in the global event signature.
  \item $\mathcal{L}$, the \emph{listener map}: a partial function
    $\mathcal{L} : \ell\langle\overline{v}\rangle \rightharpoonup
    \mathcal{P}_{\mathit{fin}}(\mathit{Closure} \times \{\eventually{}, \always{}\} \times \mathit{Effect})$
    mapping exact events to finite sets of registered listener records. Each
    record $(c, m, F)$ pairs the closure $c$ with a mode tag --- $\eventually{}$
    for a one-shot listener and $\always{}$ for a persistent one --- and the
    closure's body effect $F$, recorded at the point of registration. The
    effect component $F$ is an instrumented annotation only: it is used by
    well-typedness of configurations to witness the produced update queue when
    the listener fires, and is erased in bare execution. If
    $\ell\langle\overline{v}\rangle \notin \mathrm{dom}(\mathcal{L})$ we treat
    $\mathcal{L}(\ell\langle\overline{v}\rangle) = \emptyset$.
  \item $\mathcal{C}$, the \emph{cancellation multiset}: a finite multiset of
    exact events $\ell\langle\overline{v}\rangle$. When \rulename{pop event}
    is about to dispatch an event, it first consults $\mathcal{C}$: if that
    event has multiplicity $\geq 1$ in $\mathcal{C}$, the dispatch is dropped
    and the event's multiplicity in $\mathcal{C}$ is decremented by one;
    otherwise the event is dispatched normally. Because $\mathcal{C}$ is a
    multiset, the same event may be cancelled more than once, in which case
    each cancellation drops one subsequent dispatch. We write
    $\{\!\!\{\,e\,\}\!\!\}$ for a singleton multiset, $\uplus$ for multiset
    union (which adds multiplicities), and $\setminus$ for multiset difference
    (which subtracts multiplicities).
\end{itemize}

\subsubsection{External task finished} \;

This judgement, $\mathcal{E} \Rightarrow_{\mathsf{ext}} U_{\mathcal{E}}$, reads
the pending event firings out of the external scheduler $\mathcal{E}$ as a queue
of update-queue items to be prepended at the start of a housekeeping phase.
Because $\mathcal{E}$ is an opaque black box, the single rule simply exposes
whatever ordered list of firings it currently holds.

\begin{placedfigure}
\Rule{external task finished}
  {
    \mathcal{E} = (\ell_1\langle\overline{v_{11}}\rangle,\, v_{21}), \ldots,
      (\ell_n\langle\overline{v_{1n}}\rangle,\, v_{2n})
  }
  {
    \mathcal{E} \;\Rightarrow_{\mathsf{ext}}\;
    (\ell_1\langle\overline{v_{11}}\rangle,\, v_{21}), \ldots,
    (\ell_n\langle\overline{v_{1n}}\rangle,\, v_{2n})
  }
\captionof{figure}{External tasks finishing.}
\label{fig:app-sem-ext-finished}
\end{placedfigure}

\subsubsection{Listener-safe queues} \;

A listener fires in response to an event, so it must not be able to inject a
fresh labelled event back into the queue. It may, however, enqueue the internal
housekeeping items: state updates, cancellations, listener removals, and new
listener registrations. The judgement $\listenersafe{U}$ holds exactly when $U$
contains no labelled events.

\begin{placedfigure}
\Rule{lsafe empty}
  { }
  { \listenersafe{\cdot} }
\quad
\Rule{lsafe setter}
  { \listenersafe{U} }
  { \listenersafe{setter_y(f),\, U} }
\quad
\Rule{lsafe cancel}
  { \listenersafe{U} }
  { \listenersafe{\mathsf{cancel}(\ell\langle\overline{v}\rangle),\, U} }

\Rule{lsafe remove}
  { \listenersafe{U} }
  { \listenersafe{\mathsf{remove}(\ell\langle\overline{v}\rangle),\, U} }
\quad
\Rule{lsafe listen}
  { \listenersafe{U} }
  { \listenersafe{\mathsf{listen}(\ell\langle\overline{v}\rangle,\, c,\, m,\, F),\, U} }
\captionof{figure}{Listener-safe queues.}
\label{fig:app-sem-lsafe}
\end{placedfigure}

\textbf{Lemma (Expression evaluation is listener-safe).}
If \semDerivExpr{V}{e}{v}{U}, then $\listenersafe{U}$.

\textit{Proof.}
By induction on the expression-level semantic derivation.  We check that
every rule in \rulename{Expression} above either emits the empty queue
$\cdot$ or extends with one of the four item kinds admitted by the
listener-safe judgement:
\begin{itemize}
  \item \rulename{var}, \rulename{constant}, \rulename{closure create}:
        $U = \cdot$; apply \rulename{lsafe empty}.
  \item \rulename{fst}, \rulename{snd}: $U$ is the queue of a subderivation;
        apply the IH to that subderivation.
  \item \rulename{seq}, \rulename{prod}, \rulename{branch1},
        \rulename{branch2}, \rulename{function app}: $U$ is the
        concatenation of subderivation queues; closure of the
        listener-safe judgement under concatenation (each \rulename{lsafe} rule extends a
        listener-safe tail with one admitted item, so an induction on the
        first queue gives concatenation closure) combined with the IH on
        each subderivation discharges the case.
  \item \rulename{setter app}: $U = U_1, U_2, setter_x(\Closure{y}{e_3}{V'})$.
        $U_1$ and $U_2$ are listener-safe by the IH; the trailing item is
        admitted by \rulename{lsafe setter}.
  \item \rulename{cancel}, \rulename{remove}: $U$ is a single $\mathsf{cancel}$
        or $\mathsf{remove}$ item; admitted by \rulename{lsafe cancel} /
        \rulename{lsafe remove}.
  \item \rulename{bind}, \rulename{once}: $U = U', \mathsf{listen}(\ldots)$.
        $U'$ is listener-safe by the IH; the trailing $\mathsf{listen}$ item
        is admitted by \rulename{lsafe listen}.
\end{itemize}
No expression-level rule emits an event-firing item
$(\eventEffBar{\ell}{v_1},\, v_2)$; labelled events enter the working queue
only via \rulename{flush external events} at the component level, never by
user code.  \qed

\subsection{Component}
\label{app:semantics:component}
\;

The extended judgment
$\semDerivTopStep{\mathcal{L}}{V_e}{V_s}{\mathcal{C}}{C}{status}{U'}$
threads $\mathcal{L}$ and $\mathcal{C}$ as store-passed state through all component rules.
Updated values of $\mathcal{C}$ or $\mathcal{L}$ appear in the conclusion's corresponding
argument positions; rules that do not modify them pass the same value through unchanged.

\begin{placedfigure}
\Rule{init}
  {
    \semDerivDecl{\Sigma}{\cdot}{\cdot}{x_1:v_1,\ldots, x_n:v_n}{``"}{p}{v_r}{U}{V}
  }{
    \semDerivTopStep
      {\emptyset}{V}{V}{\emptyset}
      {\synComponent{C}{\overline{x = v}}{\tau}{p}}
      {rendered}{U}
  }


\Rule{no updates}
  {
    \semDerivTopStep
      {\mathcal{L}}{V_e}{V_s}{\mathcal{C}}
      {\synComponent{C}{\overline{x = v}}{\tau}{p}}
      {rendered}{\cdot}
  }{
    \semDerivTopStep
      {\mathcal{L}}{V_e}{V_s}{\mathcal{C}}
      {\synComponent{C}{\overline{x = v}}{\tau}{p}}
      {waiting}{\cdot}
  }

\Rule{flush external events}
  {
    \semDerivTopStep
      {\mathcal{L}}{V_e}{V_s}{\mathcal{C}}
      {\synComponent{C}{\overline{x = v}}{\tau}{p}}
      {rendered}{U}
    \\\\
    \mathcal{E} \;\Rightarrow_{\mathsf{ext}}\; U_{\mathcal{E}}
  }{
    \semDerivTopStep
      {\mathcal{L}}{V_e}{V_s}{\mathcal{C}}
      {\synComponent{C}{\overline{x = v}}{\tau}{p}}
      {rendered}{U_{\mathcal{E}},\,U}
  }

\Rule{pop setter}
  {
    \semDerivTopStep
      {\mathcal{L}}{V_e}{V_s}{\mathcal{C}}
      {\synComponent{C}{\overline{x = v}}{\tau}{p}}
      {rendered}{setter_y(f),\, U'}
    \\\\
    \setterderiv{V_s}{setter_y(f)}{V_s'}
  }{
    \semDerivTopStep
      {\mathcal{L}}{V_e}{V_s'}{\mathcal{C}}
      {\synComponent{C}{\overline{x = v}}{\tau}{p}}
      {rendered}{U'}
  }

\Rule{pop event}
  {
    \semDerivTopStep
      {\mathcal{L}}{V_e}{V_s}{\mathcal{C}}
      {\synComponent{C}{\overline{x = v}}{\tau}{p}}
      {rendered}{(\ell\langle\overline{v_1}\rangle,\, v_2),\, U'}
    \\\\
    \mathcal{C}(\ell\langle\overline{v_1}\rangle) = 0
    \\\\
    \mathcal{L}(\ell\langle\overline{v_1}\rangle)
      = \{(\Closure{x_1}{e_1}{V_1}, m_1, F_1), \ldots, (\Closure{x_k}{e_k}{V_k}, m_k, F_k)\}
    \and
    \overline{m_i \in \{\eventually{}, \always{}\}}
    \\\\
    \overline{\semDerivExpr{V_i,\; x_i = v_2}{e_i}{()}{U_i}}
    \and
    \overline{\listenersafe{U_i}}
    \\\\
    \mathcal{L}' = \mathcal{L}\bigl[\ell\langle\overline{v_1}\rangle \mapsto
      \{(\Closure{x_i}{e_i}{V_i}, m_i, F_i) \mid m_i = \always{}\}\bigr]
  }{
    \semDerivTopStep
      {\mathcal{L}'}{V_e}{V_s}{\mathcal{C}}
      {\synComponent{C}{\overline{x = v}}{\tau}{p}}
      {rendered}{U',\, U_1, \ldots, U_k}
  }

\Rule{pop event cancelled}
  {
    \semDerivTopStep
      {\mathcal{L}}{V_e}{V_s}{\mathcal{C}}
      {\synComponent{C}{\overline{x = v}}{\tau}{p}}
      {rendered}{(\ell\langle\overline{v_1}\rangle,\, v_2),\, U'}
    \\\\
    \mathcal{C}(\ell\langle\overline{v_1}\rangle) \geq 1
    \and
    \mathcal{C}' = \mathcal{C} \setminus \{\!\!\{\,\ell\langle\overline{v_1}\rangle\,\}\!\!\}
  }{
    \semDerivTopStep
      {\mathcal{L}}{V_e}{V_s}{\mathcal{C}'}
      {\synComponent{C}{\overline{x = v}}{\tau}{p}}
      {rendered}{U'}
  }

\Rule{pop cancel}
  {
    \semDerivTopStep
      {\mathcal{L}}{V_e}{V_s}{\mathcal{C}}
      {\synComponent{C}{\overline{x = v_1}}{\tau}{p}}
      {rendered}{\mathsf{cancel}(\ell\langle\overline{v_2}\rangle),\, U'}
  }{
    \semDerivTopStep
      {\mathcal{L}}{V_e}{V_s}
      {\mathcal{C} \uplus \{\!\!\{\,\ell\langle\overline{v_2}\rangle\,\}\!\!\}}
      {\synComponent{C}{\overline{x = v_1}}{\tau}{p}}
      {rendered}{U'}
  }

\Rule{pop listen}
  {
    \semDerivTopStep
      {\mathcal{L}}{V_e}{V_s}{\mathcal{C}}
      {\synComponent{C}{\overline{x = v_1}}{\tau}{p}}
      {rendered}{\mathsf{listen}(\ell\langle\overline{v_2}\rangle,\, c,\, m,\, F),\, U'}
    \\\\
    m \in \{\eventually{}, \always{}\}
    \\\\
    \mathcal{L}' = \mathcal{L}[\ell\langle\overline{v_2}\rangle \mapsto \mathcal{L}(\ell\langle\overline{v_2}\rangle) \cup \{(c, m, F)\}]
  }{
    \semDerivTopStep
      {\mathcal{L}'}{V_e}{V_s}{\mathcal{C}}
      {\synComponent{C}{\overline{x = v_1}}{\tau}{p}}
      {rendered}{U'}
  }

\Rule{pop remove listener}
  {
    \semDerivTopStep
      {\mathcal{L}}{V_e}{V_s}{\mathcal{C}}
      {\synComponent{C}{\overline{x = v_1}}{\tau}{p}}
      {rendered}{\mathsf{remove}(\ell\langle\overline{v_2}\rangle),\, U'}
    \\\\
    \mathcal{L}' = \mathcal{L}[\ell\langle\overline{v_2}\rangle \mapsto \emptyset]
  }{
    \semDerivTopStep
      {\mathcal{L}'}{V_e}{V_s}{\mathcal{C}}
      {\synComponent{C}{\overline{x = v_1}}{\tau}{p}}
      {rendered}{U'}
  }

\Rule{waiting to rendered}
  {
    \semDerivTopStep
      {\mathcal{L}}{V_e}{V_s}{\mathcal{C}}
      {\synComponent{C}{\overline{x = v}}{\tau}{p}}
      {waiting}{\cdot}
    \\\\
    \semDerivDecl{\Sigma}{V_e}{V_s}{\overline{x = v'}}{``"}{p}{v_r}{U'}{V'} \\
    \exists i \;s.t. \; v_i \neq v'_i \\
  }{
    \semDerivTopStep
      {\mathcal{L}}{V'}{V'}{\mathcal{C}}
      {\synComponent{C}{\overline{x = v'}}{\tau}{p}}
      {rendered}{U'}
  }
\captionof{figure}{Component-level stepping.}
\label{fig:app-sem-component}
\end{placedfigure}

\breathe{}

\begin{placedfigure}
\Rule{valid setter}
  {
    V' = removeSetters(V) \\
    \semDerivExpr{V',x=v}{e}{v'}{\cdot}
  }{
    \setterderiv{V_s,x=v}{setter_x(\Closure{x}{e}{V})}{V_s,x=v'}
  }
\captionof{figure}{Setter validity.}
\label{fig:app-sem-valid-setter}
\end{placedfigure}

\subsection{Declarations} \;

The declaration evaluation judgement
$\semDerivDecl{\Sigma}{V_e}{V_s}{V}{c}{p}{v}{U}{V'}$ is the big-step operational
semantics of a declaration block. Under signature $\Sigma$, previous-render
environment $V_e$, and setter store $V_s$, it evaluates the block $p$ in the
local environment $V$ to a return value $v$, emitting an update queue $U$ and a
final environment $V'$. The two rules here handle a plain \rulename{var}
(let-bound) declaration and the terminating \rulename{return}; the
\rulename{Initialize} and \rulename{Rerender} rules below cover the remaining
declaration forms in their two evaluation modes.

\begin{placedfigure}
\Rule{var}
  {
    \semDerivExpr{V}{e}{v_0}{\cdot} \\
    \semDerivDecl{\Sigma}{V_e}{V_s}{V, x=v_0}{c}{p}{v}{U}{V'} \\
  }
  {
    \semDerivDecl{\Sigma}{V_e}{V_s}{V}{c}{\synCompLet{x}{e} p}{v}{U}{V'}
  }

\quad
\Rule{return}
  {
    \semDerivExpr{V}{x}{v}{\cdot} \\
  }
  {
    \semDerivDecl{\Sigma}{V_e}{V_s}{V}{c}{\synReturn{x}}{v}{\cdot}{V}
  }
\captionof{figure}{Declaration evaluation: \rulename{var} and \rulename{return}.}
\label{fig:app-sem-decls}
\end{placedfigure}

\subsubsection{Initialize} \;

These rules give the \emph{first-render} evaluation of state, effect, and
subcomponent declarations, taken when the component is first mounted (the
previous-render environments are empty, $V_e = V_s = \cdot$). A state
declaration installs its initial value and setter, an effect block always runs
once and enqueues its $\mathsf{on}$ closure, and a subcomponent is recursively
initialized with its arguments and has its resulting environment merged in under
the subcomponent's prefix.

\begin{placedfigure}
\Rule{state init}
  {
    \semDerivExpr{V}{e}{v_0}{\cdot} \\
    \semDerivDecl{\Sigma}{\cdot}{\cdot}{V, x= v_0, \setter x= setter_{c.x}}{c}{p}{v}{U}{V'}
  }
  {
    \semDerivDecl{\Sigma}{\cdot}{\cdot}{V}{c}{\synState{x}{e}p}{v}{U}{V'}
  }

\Rule{effect init}
  {
    \semDerivDecl{\Sigma}{\cdot}{\cdot}{V}{c}{p}{v}{U'}{V'}
  }
  {
    \semDerivDecl
      {\Sigma}{\cdot}{\cdot}{V}
      {c}
      {\synOn{x_1, x_2,\ldots,x_n}{e}
        p}
      {v}{\mathsf{on}(\Closure{\_}{e}{V}),\, U'}{V'}
  }

\Rule{subcomp init}{
    \semDerivDecl{\Sigma}{\cdot}{\cdot}{\overline{x'=v}}{c.y}{p_A}{v_A}{U_y}{V_y'} \\
    \semDerivDecl{\Sigma}{\cdot}{\cdot}{V, y= v_A}{c}{p}{v}{U}{V'} \\
    V_y'' = \text{addPrefix}("y.", V_y') \\
    \overline{x=v} \in V \\
    A: (\synComponent{A}{\overline{x'}}{\tau}{p_A}, \Delta, \Gamma) \in \Sigma
  }{
  \semDerivDecl
    {\Sigma}
    {\cdot}
    {\cdot}
    {V}
    {c}
    {\synSubComp{y}{A(\overline{x})} p}
    {v}
    {U_y, U}
    {V' \cup V_y''}
  }
\captionof{figure}{Declaration evaluation: initialization.}
\label{fig:app-sem-init}
\end{placedfigure}

\subsubsection{Rerender} \;

These rules give the \emph{re-render} evaluation of the same declaration forms,
taken on every render after the first. They are distinguished from
initialization by comparing the previous-render environment $V_e$ against the
current values: a state declaration reuses its stored value, and an effect block
re-runs its body only when one of its watched variables $x_1, \ldots, x_n$ has
actually changed (\rulename{effect rerender, yes changes}), otherwise it is
skipped (\rulename{effect rerender, no changes}).

\begin{placedfigure}
\Rule{state rerender}
  {
    \semDerivDecl{\Sigma}{V_e}{V_s}{V, x=v_0, \setter{x}=setter_{c.x}}{c}{p}{v}{U}{V'} \\
    x=v_0 \in V_s, x=v_e \in V_e
  }
  {
    \semDerivDecl
      {\Sigma}{V_e}{V_s}{V}
      {c}
      {\synState{x}{e}
        p}
      {v}{U}{V'}
  }

\Rule{effect rerender, no changes}
  {
    \semDerivDecl{\Sigma}{V_e}{V_s}{V}{c}{p}{v}{U}{V'} \\
    x_1=v_1', x_2=v_2',\ldots,x_n=v_n' \in V \\
    x_1=v_1, x_2=v_2,\ldots,x_n=v_n \in V_e \\
    \forall i \in [n], v_i = v_i' \\
  }
  {
    \semDerivDecl
      {\Sigma}{V_e}{V_s}{V}
      {c}
      {\synOn{x_1, x_2,\ldots,x_n}{e}
        p}
      {v}{U}{V'}
  }

\Rule{effect rerender, yes changes}
  {
    \semDerivDecl{\Sigma}{V_e}{V_s}{V}{c}{p}{v}{U'}{V'} \\
    \semDerivExpr{V}{e}{v_e}{U} \\
    x_1=v_1', x_2=v_2',\ldots,x_n=v_n' \in V \\
    x_1=v_1, x_2=v_2,\ldots,x_n=v_n \in V_e \\
    \exists i \in [n], v_i \ne v_i' \\
  }
  {
    \semDerivDecl
      {\Sigma}{V_e}{V_s}{V}
      {c}
      {\synOn{x_1, x_2,\ldots,x_n}{e} p}
      {v}{U,\, U'}{V'}
  }

\Rule{subcomp rerender}{
    \semDerivDecl{\Sigma}{V_e'}{V_s'}{\overline{x'=v}}{c.y}{p_A}{v_A}{U_y}{V_y'} \\
    \semDerivDecl{\Sigma}{V_e}{V_s}{V, y= v_A}{c}{p}{v}{U}{V'} \\
    V_e' = \text{removePrefix}("y.",\text{startsWith}("y.",V_e)) \\
    V_s' = \text{removePrefix}("y.",\text{startsWith}("y.",V_s)) \\
    V_y'' = \text{addPrefix}("y.", V_y') \\
    \overline{x=v} \in V \\
    A: (\synComponent{A}{\overline{x'}}{\tau}{p_A}, \Delta, \Gamma) \in \Sigma\\
  }{
    \semDerivDecl
      {\Sigma}
      {V_e}{V_s}{V}{c}
      {\synSubComp{y}{A(\overline{x})}p}
      {v}{U_y, U}{V'\cup V_y''}
  }
\captionof{figure}{Declaration evaluation: rerender.}
\label{fig:app-sem-rerender}
\end{placedfigure}

\subsection{Expression} \;

The expression evaluation judgement $\semDerivExpr{V}{e}{v}{U}$ evaluates
expression $e$ in environment $V$ to a value $v$, collecting into the update
queue $U$ every side-effecting item the evaluation produces --- setter
applications, cancellations, removals, and the listener registrations emitted by
$\kw{bind}$ and $\kw{once}$. It is a standard call-by-value big-step semantics;
the queue $U$ is the operational counterpart of the effect $F$ assigned to $e$
by expression typing.

\begin{placedfigure}
\Rule{var}
  {
  }
  {
    \semDerivExpr{V,x=v}{x}{v}{\cdot}
  }
\quad
\Rule{constant}
  {
    c \in \alpha \cup \{true, false\} \cup \{()\}
  }
  {
    \semDerivExpr{V}{c}{c}{\cdot}
  }
\quad
\Rule{setter app}
  {
    \semDerivExpr{V}{e_1}{setter_x}{U_1} \\
    \semDerivExpr{V}{e_2}{\Closure{y}{e_3}{V'}}{U_2}
  }
  {
    \semDerivExpr{V}{e_1\;e_2}{()}{U_1, U_2, setter_x(\Closure{y}{e_3}{V'})}
  }
\quad
\Rule{function app}
  {
    \semDerivExpr{V}{e_1}{\Closure{x}{e_3}{V'}}{U_1} \\
    \semDerivExpr{V}{e_2}{v_2}{U_2} \\
    \semDerivExpr{V', x: v_2}{e_3}{v_3}{U_3} \\
  }
  {
    \semDerivExpr{V}{e_1\;e_2}{v_3}{U_1, U_2, U_3}
  }
\quad
\Rule{closure create}
  {
  }
  {
    \semDerivExpr{V}{\synFn{x}{e}}{\Closure{x}{e}{V}}{\cdot}
  }
\quad
\Rule{seq}
  {
    \semDerivExpr{V}{e_1}{v_1}{U_1} \\
    \semDerivExpr{V}{e_2}{v_2}{U_2} \\
  }
  {
    \semDerivExpr{V}{e_1\;;\;e_2}{v_2}{U_1, U_2}
  }
\quad
\Rule{branch1}
  {
    \semDerivExpr{V}{e_1}{true}{U_1} \\
    \semDerivExpr{V}{e_2}{v_2}{U_2} \\
  }
  {
    \semDerivExpr{V}{\synIf{e_1}{e_2}{e_3}}{v_2}{U_1, U_2}
  }

\quad
\Rule{branch2}
  {
    \semDerivExpr{V}{e_1}{false}{U_1} \\
    \semDerivExpr{V}{e_3}{v_3}{U_3} \\
  }
  {
    \semDerivExpr{V}{\synIf{e_1}{e_2}{e_3}}{v_3}{U_1, U_3}
  }
\quad
\Rule{prod}{
    \semDerivExpr{V}{e_1}{v_1}{U_1} \\
    \semDerivExpr{V}{e_2}{v_2}{U_2}
    }{
    \semDerivExpr{V}{(e_1, e_2)}{(v_1, v_2)}{U_1, U_2}
  }
\quad
\Rule{fst}{
    \semDerivExpr{V}{e}{(v_1, v_2)}{U}
  }{
    \semDerivExpr{V}{\kw{fst}\;e}{v_1}{U}
  }
\quad
\Rule{snd}{
    \semDerivExpr{V}{e}{(v_1, v_2)}{U}
  }{
    \semDerivExpr{V}{\kw{snd}\;e}{v_2}{U}
  }
\quad
\Rule{cancel}{
    \\
  }{
    \semDerivExpr{V}{\kw{cancel}\;\ell\langle\overline{v}\rangle}{()}{\mathsf{cancel}(\ell\langle\overline{v}\rangle)}
  }
\quad
\Rule{remove}{
    \\
  }{
    \semDerivExpr{V}{\kw{remove}\;\ell\langle\overline{v}\rangle}{()}{\mathsf{remove}(\ell\langle\overline{v}\rangle)}
  }

\quad
\Rule{bind}{
    \semDerivExpr{V}{e}{\Closure{x}{e'}{V'}}{U} \\
    \valTyDeriv{\Closure{x}{e'}{V'}}{\tyArrow{\tau}{\tyUnit}[F]}
  }{
    \semDerivExpr{V}{\kw{bind}\;\ell\langle\overline{v_1}\rangle\;e}{()}{U,\,
    \mathsf{listen}(\ell\langle\overline{v_1}\rangle,\,
    \Closure{x}{e'}{V'},\, \always{},\, F)}
  }

\quad
\Rule{once}{
    \semDerivExpr{V}{e}{\Closure{x}{e'}{V'}}{U} \\
    \valTyDeriv{\Closure{x}{e'}{V'}}{\tyArrow{\tau}{\tyUnit}[F]}
  }{
    \semDerivExpr{V}{\kw{once}\;\ell\langle\overline{v_1}\rangle\;e}{()}{U,\,
    \mathsf{listen}(\ell\langle\overline{v_1}\rangle,\,
    \Closure{x}{e'}{V'},\, \eventually{},\, F)}
  }
\captionof{figure}{Expression evaluation.}
\label{fig:app-sem-expr}
\end{placedfigure}

\section{Instrumented Semantics}
\subsection{Declarations} \;

The instrumented declaration judgement
$\semPfDerivDecl{\Gamma}{X}{V_s}{V}{c}{p}{v}{T}{V'}$ mirrors the ordinary
declaration semantics but additionally records a \emph{trace} $T$ --- a
structured log of the render with one entry per declaration
($\mathit{LET}$, $\mathit{EFF}$, $\mathit{RET}$, and subcomponent
$\mathit{ARG}$s) --- and is parameterized by the set $X$ of variables that
changed this render. The trace is the object that the metatheory's
\emph{justification} judgements later check against $\Delta$; it is instrumented
bookkeeping only and is erased in ordinary execution. These two rules cover the
\rulename{var} (let-bound) and \rulename{return} declarations.

\begin{placedfigure}
\Rule{var}
  {
    \semDerivExpr{V}{e}{v_0}{\cdot} \\
    \semPfDerivDecl{\Gamma}{X}{V_s}{V, x=v_0}{c}{p}{v}{U'}{V'} \\
  }
  {
  \semPfDerivDecl{\Gamma}{X}{V_s}{V}{c}{\synCompLet{x}{e}p}{v}{\semPfTraceLet{x}{\mathit{df}(e)}, U}{V'}
  }
\quad
\Rule{return}
  {
    x = v_s \in V_s \\
    x = v \in V \\
    \valTyEnvDeriv{\Gamma}{V}
  }
  {
    \semPfDerivDecl{\Gamma}{X}{V_s}{V}{c}{\synReturn{x}}{v}{\semPfTraceReturn{x}}{V}
  }
\captionof{figure}{Instrumented declaration evaluation: \rulename{var} and \rulename{return}.}
\label{fig:app-inst-decls}
\end{placedfigure}

\subsubsection{Rerender} \;

These are the re-render rules of the instrumented semantics, the trace-producing
counterparts of the re-render rules of the ordinary semantics. An effect block
that fires (its watched set meets the changed set $X$) records an $\mathit{EFF}$
entry carrying the update queue it produced (\rulename{effect rerender, yes
changes}); one that does not fire records nothing new
(\rulename{effect rerender, no changes}), and subcomponents recurse with the
appropriately re-prefixed changed set.

\begin{placedfigure}
\Rule{state rerender}
  {
    \semPfDerivDecl{\Gamma}{X}{V_s}{V, x=v, \setter{x}=setter_{c.x}}{c}{p}{v}{T}{V'} \\
    x : \tau \in \Gamma \\
    x = v \in V_s \\
    x = setter_{c.x} \in V_s \\
    \valTyDeriv{v}{\tau}
  }
  {
    \semPfDerivDecl
      {\Gamma}{X}{V_s}{V}{c}
      {\synState{x}{e}
        p}
      {v}{T}{V'}
  }
\Rule{effect rerender, no changes}
  {
    \semPfDerivDecl
      {\Gamma}{X}{V_s}{V}{c}
      {p}
      {v}{T}{V'} \\
    \overline{x} \cap X = \emptyset\\
  }
  {
    \semPfDerivDecl
      {\Gamma}{X}{V_s}{V}{c}
      {\synOn{\overline{x}}{e}
        p}
      {v}{T}{V'}
  }
\Rule{effect rerender, yes changes}
  {
    \semPfDerivDecl{\Gamma}{X}{V_s}{V}{c}{p}{v}{T}{V'} \\
    \semDerivExpr{V}{e}{()}{U} \\
    \overline{x} \cap X \neq \emptyset\\
    \exists \Gamma'. (\Gamma' \subseteq \Gamma, \valTyEnvDeriv{\Gamma'}{V})
  }
  {
    \semPfDerivDecl
      {\Gamma}{X}{V_s}{V}{c}
      {\synOn{\overline{x}}{e} p}
      {v}{\semPfTraceOn{\overline{x} \cap X}{U}, T}{V'}
  }
\Rule{subcomp rerender}{
    \semPfDerivDecl{\Gamma_{sF}}{X'}{V_s'}{\overline{x'=v}}{c.y}{p_F}{v_F}{T_y,\semPfTraceReturn{z}}{V_y'} \\
    \semPfDerivDecl{\Gamma_s}{X}{V_s}{V, y= v_F}{c}{p}{v}{T}{V'} \\
    \overline{x=v} \in V \\
    F: (\synComponent{F}{\overline{x'}}{\tau_{rF}}{p_F}, \Delta_F, \Gamma_{sF}) \in \Sigma\\\\
    X' = \text{removePrefix}("y.",\text{startsWith}("y.",X)) \\
    V_s' = \text{removePrefix}("y.",\text{startsWith}("y.",V_s)) \\
    V_y'' = \text{addPrefix}("y.", V_y') \\
  }{
    \semPfDerivDecl
      {\Gamma}
      {X}{V_s}{V}{c}
      {\synSubComp{y}{A(\overline{x})}p}
      {v}{T_y, T}{V'\cup V_y''}
  }
\captionof{figure}{Instrumented declaration evaluation: rerender.}
\label{fig:app-inst-rerender}
\end{placedfigure}

\section{Preservation}

\subsection{Values and Types}

\;

These rules define when a runtime value is well-typed. The value-typing
judgement $\valTyDeriv{v}{\tau}$ assigns a closed value $v$ its type $\tau$
(closures also carry the latent effect $F$ of their body), and the
value-environment judgement $\valTyEnvDeriv{\Gamma}{V}$ lifts this pointwise,
holding when every binding in the runtime environment $V$ matches its declared
type in $\Gamma$. These are the semantic typing relations used throughout the
preservation proof.

\begin{placedfigure}
\Rule{val-unit}
  { }
  {\valTyDeriv{()}{\tyUnit}}
\quad
\Rule{val-true}
  { }
  {\valTyDeriv{true}{\tyBool}}
\quad
\Rule{val-false}
  { }
  {\valTyDeriv{false}{\tyBool}}
\quad
\Rule{val-const}
  { }
  {\valTyDeriv{c}{\alpha}}
\quad
\Rule{val-pair}
  {
    \valTyDeriv{v_1}{\tau_1} \\
    \valTyDeriv{v_2}{\tau_2}
  }
  {\valTyDeriv{(v_1, v_2)}{\tau_1 \times \tau_2}}
\quad
\Rule{val-closure}
  {
    \exists \Gamma \\
    \valTyEnvDeriv{\Gamma}{V} \\
    \tyDerivExpr{\Gamma, x: \tau_1}{e}{\tau_2}{F}
  }
  {\valTyDeriv{\Closure{x}{e}{V}}{(\tyArrow{\tau_1}{\tau_2}[F])}}
\quad
\Rule{val-setter}
  {\mathit{Main}: (\tinyComp{Main}, \Delta, \Gamma) \in \Sigma \\
    x: \tau \in \Gamma
  }
  {\valTyDeriv{setter_x}{(\tySetter{x})}}
\captionof{figure}{Value typing.}
\label{fig:app-val-typing}
\end{placedfigure}

\breathe{}

\begin{placedfigure}
\Rule{valenv-empty}
  { }
  {\valTyEnvDeriv{\cdot}{\cdot}}
\quad
\Rule{valenv-cons}
  {
    \valTyEnvDeriv{\Gamma}{V} \\
    \valTyDeriv{v}{\tau}
  }
  {\valTyEnvDeriv{\Gamma, x:\tau}{V,x = v}}
\captionof{figure}{Value environment typing.}
\label{fig:app-valenv-typing}
\end{placedfigure}

\subsection{Side Conditions}




\subsubsection{Typing Update Queues} \;

We first give the precise effect of a single queue item, one rule per kind of
event the queue may carry, then assemble the effect of a whole queue with
\rulename{tq nil} and \rulename{tq cons}. Subtyping (including
$*$-splitting and $+$-branching across the queue) is handled by the rules in
\textit{Derived from Subtyping Relations} below.

\paragraph{Single events.}\;

\begin{placedfigure}
\Rule{tq setter}
  { }
  {\effQueueDeriv{\Next{1r}{\stch{x}}}{setter_x(f)}}
\quad
\Rule{tq event}
  {
    \valTyDeriv{v_2}{\Sigma_E(\eventEffBar{\ell}{v_1})}
  }
  {\effQueueDeriv{\eventEffBar{\ell}{v_1}}{(\eventEffBar{\ell}{v_1},\, v_2)}}
\quad
\Rule{tq cancel}
  { }
  {\effQueueDeriv{\cancelEv{\eventEffBar{\ell}{v}}}{\mathsf{cancel}(\eventEffBar{\ell}{v})}}
\quad
\Rule{tq remove}
  { }
  {\effQueueDeriv{\remove{\eventEffBar{\ell}{v}}}{\mathsf{remove}(\eventEffBar{\ell}{v})}}
\quad
\Rule{tq listen once}
  {
    \valTyDeriv{c}{\tyArrow{\tau}{\tyUnit}[F]}\\
    \tau = \Sigma_E(\eventEffBar{\ell}{v})
  }
  {\effQueueDeriv
    {\eventually{\eventEffBar{\ell}{v}}(F)}
    {\mathsf{listen}(\eventEffBar{\ell}{v},\, c,\, \eventually{},\, F)}}
\quad
  \Rule{tq listen always}
  {
    \valTyDeriv{c}{\tyArrow{\tau}{\tyUnit}[F]}\\
    \tau = \Sigma_E(\eventEffBar{\ell}{v})
  }
  {\effQueueDeriv
    {\always{\eventEffBar{\ell}{v}}(F)}
    {\mathsf{listen}(\eventEffBar{\ell}{v},\, c,\, \always{},\, F)}}
\captionof{figure}{Queue-item typing: single events.}
\label{fig:app-tq-single}
\end{placedfigure}

\paragraph{Queue assembly.} \;

\begin{placedfigure}
\Rule{tq nil}
  { }
  {\effQueueDeriv{\cdot}{\cdot}}
\quad
\Rule{tq cons}
  {
    \effQueueDeriv{F_1}{e} \\
    \effQueueDeriv{F_2}{U}
  }
  {\effQueueDeriv{F_1 * F_2}{e,\, U}}
\captionof{figure}{Queue-item typing: queue assembly.}
\label{fig:app-tq-assembly}
\end{placedfigure}

\;

\subsubsection{Derived from Subtyping Relations} \;

These two rules extend queue-item typing to whole queues using the sub-effecting
structure of effects. \rulename{tq-merge} types a concatenated queue with the
product of its parts' effects, and \rulename{tq-pick} lets a queue that realises
one branch of a $+$ effect be typed at the whole choice. Together with the base
rules above, they make $\effQueueDeriv{F}{U}$ closed under the $*$ and $+$
connectives, mirroring sub-effecting on the effect side.

\begin{placedfigure}
\Rule{tq-merge}
  {
    \effQueueDeriv{F_1}{U_1} \\
    \effQueueDeriv{F_2}{U_2}
  }
  {\effQueueDeriv{F_1*F_2}{U_1,U_2}}
\quad
\Rule{tq-pick}
  {
    \effQueueDeriv{F_i}{U} \\
    i\in\{1,2\}
  }
  {\effQueueDeriv{F_1+F_2}{U}}
\captionof{figure}{Queue typing derived from subtyping.}
\label{fig:app-tq-subtyping}
\end{placedfigure}

\textbf{Lemma (Listen typing inversion).}
If $\effQueueDeriv{F'}{U}$ and
$\mathsf{listen}(\eventEffBar{\ell}{v},\, c,\, m,\, F) \in U$,
then $\valTyDeriv{c}{\tyArrow{\Sigma_E(\eventEffBar{\ell}{v})}{\tyUnit}[F]}$.

\textit{Proof.}
By induction on the derivation of $\effQueueDeriv{F'}{U}$.
\begin{itemize}
  \item \rulename{tq nil}: $U = \cdot$ contains no items; vacuous.
  \item \rulename{tq event}, \rulename{tq setter}, \rulename{tq cancel},
        \rulename{tq remove}: $U$ is a single non-listen item; vacuous.
  \item \rulename{tq listen once} / \rulename{tq listen always}:
        $U = \mathsf{listen}(\eventEffBar{\ell}{v},\, c,\, m,\, F)$.
        The rule's premises include
        $\valTyDeriv{c}{\tyArrow{\tau}{\tyUnit}[F]}$ together with
        $\tau = \Sigma_E(\eventEffBar{\ell}{v})$, which is the goal.
  \item \rulename{tq cons}: $U = e,\, U''$ derived from
        $\effQueueDeriv{F_1}{e}$ and $\effQueueDeriv{F_2}{U''}$ with
        $F' = F_1 * F_2$.  The listen item is either $e$ (apply the IH to the
        smaller derivation $\effQueueDeriv{F_1}{e}$) or in $U''$ (apply the IH
        to $\effQueueDeriv{F_2}{U''}$).
  \item Subtyping rules ($*$-merge, $+$-pick): structurally smaller
        $\effQueueDeriv{F''}{U}$ premises witness the same items; apply the IH.
\end{itemize}
\qed

\subsubsection{Justifying Traces} \;

The justification judgement $\justifies{\Delta, X}{T}$ reads ``the effect
environment $\Delta$, together with the changed-variable set $X$, justifies the
trace $T$''. It is the central soundness relation of the metatheory: it walks
the trace produced by a render and checks that every recorded effect (each
$\mathit{EFF}$, $\mathit{LET}$, $\mathit{RET}$, external firing, and success
continuation) was permitted by $\Delta$ --- that is, caused by a variable that
actually changed --- so that a well-typed program only ever fires effects its
declared dependencies allow.



\begin{placedfigure}
\Rule{justify-args}
  {\justifies{\Delta, X}{T}}
  {\justifies{\Delta, X}{\semPfTraceArgs{\overline x}{\overline v_1}{\overline v_2},T}}
\quad
\Rule{justify-return}
  {
    \causesderiv{\Delta}{x}{\stch{return}}
  }
  {\justifies{\Delta, X}{\semPfTraceReturn{x}}}

\bigskip

\Rule{justify-ext}
  {
    \forall i.\; R_i \in \left\{
      \semPfTraceFire{(\eventEffBar{\ell_i}{v_{1i}},\, v_{2i})},\;
      \semPfTraceFireCxl{(\eventEffBar{\ell_i}{v_{1i}},\, v_{2i})},\;
      \semPfTraceFireSuc{(\eventEffBar{\ell_i}{v_{1i}},\, v_{2i})}{(U_{i1},F_{i1}),\ldots,(U_{i,k_i},F_{i,k_i})}
    \right\} \\
    \forall i\;\text{with}\;R_i = \semPfTraceFireSuc{(\eventEffBar{\ell_i}{v_{1i}},\, v_{2i})}{(U_{i1},F_{i1}),\ldots,(U_{i,k_i},F_{i,k_i})}.\;\forall j.\big(\\
    \quad\effQueueDeriv{F_{ij}}{U_{ij}}\\
    \quad\exists c_{ij}.\;\valTyDeriv{c_{ij}}{\tyArrow{\Sigma_E(\eventEffBar{\ell_i}{v_{1i}})}{\tyUnit}[F_{ij}]}\big)\\
    \justifies{\Delta, X}{T}
  }
  {\justifies{\Delta, X}{
    \semPfTraceExt{
      (\eventEffBar{\ell_1}{v_{11}},\, v_{21}),
      (\eventEffBar{\ell_2}{v_{12}},\, v_{22}), \ldots,
      (\eventEffBar{\ell_n}{v_{1n}},\, v_{2n})},
    R_1, R_2, \ldots, R_n,
    T
  }}

\bigskip

\Rule{justify-let}
  {
    \overline{\causesderiv{\Delta}{x_i}{\stch y}}\\
    \justifies{\Delta, X}{T}
  }
  {\justifies{\Delta, X}{\semPfTraceLet{y}{\overline x}, T}}
\quad
\Rule{justify-on}
  {
    \overline x \subseteq X \\
    \exists F. (
      \overline{\causesderiv{\Delta}{x_i}{F}}
      \text{ and }
      \effQueueDeriv{F}{U}
    )\\
    \justifies{\Delta, X}{T}
  }
  {\justifies{\Delta, X}{\semPfTraceOn{\overline x}{U}, T}}
\captionof{figure}{Justifying traces.}
\label{fig:app-justify-traces}
\end{placedfigure}



\subsubsection{Justifying Variable Change Sets} \;

The change-set judgement $\justifiesVarSet{\Delta, X_1}{T}{X_2}$ computes, from a
trace $T$ evaluated under the incoming changed set $X_1$, the set $X_2$ of
variables whose state was actually changed during that render. It reads the
setter applications out of the trace's $\mathit{EFF}$ and
$\mathit{FIRE\text{-}SUC}$ entries (and the changed arguments out of an
$\mathit{ARG}$ entry); the resulting $X_2$ becomes the changed set against which
the next render's justification is checked.

\begin{placedfigure}
\Rule{jvs-empty}
  { }
  {
    \justifiesVarSet{\Delta, X_1}{\cdot}{\emptyset}
  }
\quad
\Rule{jvs-return}
  {
    \justifiesVarSet{\Delta, X_1}{T}{X_2}
  }
  {
    \justifiesVarSet{\Delta, X_1}{T,\semPfTraceReturn{x}}{X_2}
  }
\quad
\Rule{jvs-on}
  {
    \justifiesVarSet{\Delta, X_1}{T}{X_2}
  }
  {
  \justifiesVarSet
    {\Delta, X_1}
    {T,\semPfTraceOn{\overline x}{U}}
    {X_2 \cup \{y \mid setter_y(f) \in U\}}
  }
\quad
\Rule{jvs-args}
  {
    \
  }
  {
  \justifiesVarSet
    {\Delta, X_1}
    {\semPfTraceArgs{\overline x}{\overline v_1}{\overline v_2}}
    {\{x_i \mid v_{1i} \neq v_{2i}\}}
  }
\quad
\Rule{jvs-let}
  {
    \justifiesVarSet{\Delta, X_1}{T}{X_2}
  }
  {
    \justifiesVarSet{\Delta, X_1}{T,\semPfTraceLet{x}{\overline y}}{X_2}
  }
\quad
\Rule{jvs-ext}
  {
    \justifiesVarSet{\Delta, X_1}{T}{X_2}
  }
  {
    \justifiesVarSet
      {\Delta, X_1}
      {T,\semPfTraceExt{
        (\eventEffBar{\ell_1}{v_{11}},\, v_{21}), \ldots,
        (\eventEffBar{\ell_n}{v_{1n}},\, v_{2n})}}
      {X_2}
  }
\quad
\Rule{jvs-fire}
  {
    \justifiesVarSet{\Delta, X_1}{T}{X_2}
  }
  {
    \justifiesVarSet{\Delta, X_1}{T,\semPfTraceFire{(\eventEffBar{\ell}{v_1},\, v_2)}}{X_2}
  }
\quad
\Rule{jvs-firecxl}
  {
    \justifiesVarSet{\Delta, X_1}{T}{X_2}
  }
  {
    \justifiesVarSet{\Delta, X_1}{T,\semPfTraceFireCxl{(\eventEffBar{\ell}{v_1},\, v_2)}}{X_2}
  }
\quad
\Rule{jvs-firesuc}
  {
    \justifiesVarSet{\Delta, X_1}{T}{X_2}
  }
  {
    \justifiesVarSet
      {\Delta, X_1}
      {T,\semPfTraceFireSuc{(\eventEffBar{\ell}{v_1},\, v_2)}{(U_1, F_1), \ldots, (U_k, F_k)}}
      {X_2 \cup \{y \mid \exists j.\; setter_y(f) \in U_j\}}
  }
\captionof{figure}{Justifying variable change sets.}
\label{fig:app-justify-varsets}
\end{placedfigure}

\subsubsection{Flush Queue} \;

The flush-queue judgment threads the component state
$(X, V_s, \mathcal{L}, \mathcal{C})$ through an update queue $U$, applying
each item to the appropriate carrier: setters update $V_s$ and $X$, cancels
extend $\mathcal{C}$, removes reset listener buckets in $\mathcal{L}$, and
listener registrations extend $\mathcal{L}$. The judgment is defined only on
listener-safe queues (\listenersafe{U} from \rulename{lsafe}): the four rules
below partition the item constructors that the listener-safe judgement admits, so a
derivation exists for $U$ iff $U$ is lsafe. Labelled events and effect-block
evaluations are ruled out by construction --- they are dispatched at the
component level (via \rulename{pop event}), not flushed here.

\begin{placedfigure}
\Rule{fq nil}
  { }
  {
    \setterQueueDeriv{(X, V_{s}, \mathcal{L}, \mathcal{C})}{\cdot}{(X, V_{s}, \mathcal{L}, \mathcal{C})}
  }

\bigskip

\Rule{fq setter}
  {
    \setterderiv{V_{s1}}{setter_x(f)}{V_{s2}}\\
    \setterQueueDeriv
      {(X_1 \cup \{x\}, V_{s2}, \mathcal{L}, \mathcal{C})}
      {U}
      {(X_3, V_{s3}, \mathcal{L}', \mathcal{C}')}
  }
  {
    \setterQueueDeriv
      {(X_1, V_{s1}, \mathcal{L}, \mathcal{C})}
      {setter_x(f),U}
      {(X_3, V_{s3}, \mathcal{L}', \mathcal{C}')}
  }

\bigskip

\Rule{fq cancel}
  {
    \setterQueueDeriv
      {(X_1, V_{s}, \mathcal{L}, \mathcal{C} \uplus \{\!\!\{\,\eventEffBar{\ell}{v}\,\}\!\!\})}
      {U}
      {(X_2, V_{s}', \mathcal{L}', \mathcal{C}')}
  }
  {
    \setterQueueDeriv
      {(X_1, V_{s}, \mathcal{L}, \mathcal{C})}
      {\mathsf{cancel}(\eventEffBar{\ell}{v}),U}
      {(X_2, V_{s}', \mathcal{L}', \mathcal{C}')}
  }

\bigskip

\Rule{fq remove}
  {
    \mathcal{L}_1 = \mathcal{L}[\eventEffBar{\ell}{v} \mapsto \emptyset] \\
    \setterQueueDeriv
      {(X_1, V_{s}, \mathcal{L}_1, \mathcal{C})}
      {U}
      {(X_2, V_{s}', \mathcal{L}', \mathcal{C}')}
  }
  {
    \setterQueueDeriv
      {(X_1, V_{s}, \mathcal{L}, \mathcal{C})}
      {\mathsf{remove}(\eventEffBar{\ell}{v}),U}
      {(X_2, V_{s}', \mathcal{L}', \mathcal{C}')}
  }

\bigskip

\Rule{fq listen}
  {
    m \in \{\eventually{}, \always{}\} \\
    \tau = \Sigma_E(\eventEffBar{\ell}{v})\\
    \valTyDeriv{c}{\tyArrow{\tau}{\tyUnit}[F]} \\
    \mathcal{L}_1 = \mathcal{L}[\eventEffBar{\ell}{v} \mapsto
      \mathcal{L}(\eventEffBar{\ell}{v}) \cup \{(c, m, F)\}] \\
    \setterQueueDeriv
      {(X_1, V_{s}, \mathcal{L}_1, \mathcal{C})}
      {U}
      {(X_2, V_{s}', \mathcal{L}', \mathcal{C}')}
  }
  {
    \setterQueueDeriv
      {(X_1, V_{s}, \mathcal{L}, \mathcal{C})}
      {\mathsf{listen}(\eventEffBar{\ell}{v},\, c,\, m,\, F),U}
      {(X_2, V_{s}', \mathcal{L}', \mathcal{C}')}
  }
\captionof{figure}{Flush-queue judgment.}
\label{fig:app-flush-queue}
\end{placedfigure}

\subsection{Instrumented Well Typed Configuration}\;

A configuration
$\confPfTop{\mathcal{S}}{A}{\overline{x=v}}{status}{T_1}{T_2}$ is a whole-program
runtime state: the component $A$ with its current arguments, the state tuple
$\mathcal{S} = (X_1, X_2, V_s, \mathcal{L}, \mathcal{C})$, a status of
\kw{waiting} or \kw{rendered}, and the render's trace split into an
already-processed prefix $T_1$ and a not-yet-processed suffix $T_2$.
\rulename{wt-config} defines when such a configuration is \emph{well-typed}: the
component checks, $T_1$ and $T_2$ are jointly justified by $\Delta$, $T_1$
produces the recorded change set $X_2$, and the store $V_s$ together with every
registered listener in $\mathcal{L}$ is well-typed. The remaining rules
(\rulename{Init} through \rulename{waiting to rendered}) are the configuration
step relation $\to$, which drains the trace one entry at a time --- flushing
updates and lets, dispatching or cancelling external event firings, and
re-rendering when arguments change. Preservation
(Theorem~\ref{appendix:thm:preservation-configurations}) shows this relation
preserves well-typedness.

\begin{placedfigure}
\Rule{wt-config}
  {
    \mathcal{S} = (X_1, X_2, V_s, \mathcal{L}, \mathcal{C}) \\\\
    \tyDerivComp
      {\Sigma}
      {\synComponent{A}{\overline{x_i : \tau_i}}{\tau_r}{p}}
      {(\Delta, \Gamma_s)}
    \\\\
    \justifies{\Delta, X_1}{T_1,T_2}\\
    \justifiesVarSet{\Delta,X_1}{T_1}{X_2}\\
    \valTyEnvDeriv{\Gamma_s}{V_s} \\
    \overline{\valTyDeriv{v}{\tau}}\\
    \forall \eventEffBar{\ell}{v_1} \in \mathrm{dom}(\mathcal{L}),\;
      \forall (c, m, F) \in \mathcal{L}(\eventEffBar{\ell}{v_1}).\;
      \valTyDeriv{c}{\tyArrow{\Sigma_E(\eventEffBar{\ell}{v_1})}{\tyUnit}[F]}\\
  }
  {
    \confPfTop
      {\mathcal{S}}
      {A}{\overline{x=v}}
      {status}{T_1}{T_2}
  }
\captionof{figure}{Instrumented well-typed configuration.}
\label{fig:app-wt-config}
\end{placedfigure}

\breathe{}

\begin{placedfigure}
\Rule{Init}
  {
    A: (\synComponent{A}{\overline{x:\tau}}{\tau_r}{p}, \Delta, \Gamma_s) \in \Sigma \\
    \semPfDerivDecl{\Gamma_s}{\emptyset}{\cdot}{\overline{x = v}}{``"}{p}{v_r}{T}{V}
  }{
    \cdot \to
      \confPfTop
        {(\emptyset, \emptyset, V, \emptyset, \emptyset)}
        {A}{\overline{x=v}}
        {rendered}{\cdot}{T}
  }

\bigskip

\Rule{No Updates}
  { }{
    \confPfTop
        {\mathcal{S}}
        {A}{\overline{x=v}}
        {rendered}{T_1}{\semPfTraceReturn{y}}
    \\\\\to\\\\
    \confPfTop
        {\mathcal{S}}
        {A}{\overline{x=v}}
        {waiting}{T_1,\semPfTraceReturn{y}}{\cdot}
  }

\bigskip

\Rule{Flush Update}
  {
    \mathcal{S} = (X_1, X_{21}, V_{s1}, \mathcal{L}_1, \mathcal{C}_1) \\
    \mathcal{S}' = (X_1, X_{22}, V_{s2}, \mathcal{L}_2, \mathcal{C}_2) \\\\
    \listenersafe{U} \\
    \setterQueueDeriv
      {(X_{21},V_{s1}, \mathcal{L}_1, \mathcal{C}_1)}
      {U}
      {(X_{22},V_{s2}, \mathcal{L}_2, \mathcal{C}_2)}
  }{
    \confPfTop
      {\mathcal{S}}
      {A}{\overline{x=v}}
      {rendered}{T_1}{\semPfTraceOn{\overline y}{U},T_2}
    \\\\\to\\\\
    \confPfTop
      {\mathcal{S}'}
      {A}{\overline{x=v}}
      {rendered}{T_1,\semPfTraceOn{\overline y}{U}}{T_2}
  }

\bigskip

\Rule{Flush Let}
  { }{
    \confPfTop
      {\mathcal{S}}
      {A}{\overline{x=v}}
      {rendered}{T_1}{\semPfTraceLet{y}{\overline z},T_2}
    \\\\\to\\\\
    \confPfTop
      {\mathcal{S}}
      {A}{\overline{x=v}}
      {rendered}{T_1,\semPfTraceLet{y}{\overline z}}{T_2}
  }

\bigskip

\Rule{Receive Ext}
  {
    \mathcal{S}  = (X_1, X_2, V_s, \mathcal{L}, \mathcal{C})  \\
    \mathcal{S}' = (X_1, X_2, V_s, \mathcal{L}, \mathcal{C}) \\
    \mathcal{E} \;\Rightarrow_{\mathsf{ext}}\;
      (\eventEffBar{\ell_1}{v_{11}},\, v_{21}),
      (\eventEffBar{\ell_2}{v_{12}},\, v_{22}), \ldots,
      (\eventEffBar{\ell_n}{v_{1n}},\, v_{2n}) \\
    \overline{\valTyDeriv{v_{2i}}{\Sigma_E(\eventEffBar{\ell_i}{v_{1i}})}} \\
    T_{fire} =
      \semPfTraceFire{(\eventEffBar{\ell_1}{v_{11}},\, v_{21})},
      \semPfTraceFire{(\eventEffBar{\ell_2}{v_{12}},\, v_{22})},
      \ldots,
      \semPfTraceFire{(\eventEffBar{\ell_n}{v_{1n}},\, v_{2n})}\\
    T_{ext} =
      \semPfTraceExt{
        (\eventEffBar{\ell_1}{v_{11}},\, v_{21}),
        (\eventEffBar{\ell_2}{v_{12}},\, v_{22}), \ldots,
        (\eventEffBar{\ell_n}{v_{1n}},\, v_{2n})}
  }{
    \confPfTop
      {\mathcal{S}}
      {A}{\overline{x=v}}
      {rendered}{
        \semPfTraceArgs{\overline x}{\overline{v'_1}}{\overline{v'_2}},
      }{
        T_2
      }
    \\\\\to\\\\
    \confPfTop
      {\mathcal{S}'}
      {A}{\overline{x=v}}
      {rendered}{
        \semPfTraceArgs{\overline x}{\overline{v'_1}}{\overline{v'_2}},
        T_{ext}
      }{T_{fire}, T_2}
  }
\captionof{figure}{Configuration stepping: initialization and flushing.}
\label{fig:app-config-step-flush}
\end{placedfigure}

\bigskip

\begin{placedfigure}
\Rule{Fire Cancel}
  {
    \mathcal{S}  = (X_1, X_2, V_s, \mathcal{L}, \mathcal{C})  \\
    \mathcal{S}' = (X_1, X_2, V_s, \mathcal{L}, \mathcal{C}') \\
    \mathcal{C}(\eventEffBar{\ell}{v_1}) \geq 1 \\
    \mathcal{C}' = \mathcal{C} \setminus \{\!\!\{\,\eventEffBar{\ell}{v_1}\,\}\!\!\}
  }{
    \confPfTop
      {\mathcal{S}}
      {A}{\overline{x=v}}
      {rendered}{T_1}{\semPfTraceFire{(\eventEffBar{\ell}{v_1},\, v_2)},T_2}
    \\\\\to\\\\
    \confPfTop
      {\mathcal{S}'}
      {A}{\overline{x=v}}
      {rendered}{T_1,\semPfTraceFireCxl{(\eventEffBar{\ell}{v_1},\, v_2)}}{T_2}
  }

\bigskip

\Rule{Fire Successful}
  {
    \mathcal{S}  = (X_1, X_2,  V_s,  \mathcal{L},  \mathcal{C})  \\
    \mathcal{S}' = (X_1, X_2', V_s', \mathcal{L}'', \mathcal{C}') \\
    \mathcal{C}(\eventEffBar{\ell}{v_1}) = 0 \\\\
    \mathcal{L}(\eventEffBar{\ell}{v_1})
      = \{(\Closure{x_1}{e_1}{V_1}, m_1, F_1), \ldots, (\Closure{x_k}{e_k}{V_k}, m_k, F_k)\} \\\\
    \overline{m_i \in \{\eventually{}, \always{}\}} \\
    \valTyDeriv{v_2}{\Sigma_E(\eventEffBar{\ell}{v_1})} \\
    \overline{\semDerivExpr{V_i,\; x_i = v_2}{e_i}{()}{U_i}} \\
    \overline{\listenersafe{U_i}} \\
    \overline{\effQueueDeriv{F_i}{U_i}} \\\\
    \mathcal{L}' = \mathcal{L}\bigl[\eventEffBar{\ell}{v_1} \mapsto
      \{(\Closure{x_i}{e_i}{V_i}, m_i, F_i) \mid m_i = \always{}\}\bigr] \\\\
    \setterQueueDeriv
      {(X_2, V_s, \mathcal{L}', \mathcal{C})}
      {U_1, \ldots, U_k}
      {(X_2', V_s', \mathcal{L}'', \mathcal{C}')}
  }{
    \confPfTop
      {\mathcal{S}}
      {A}{\overline{x=v}}
      {rendered}{T_1}{\semPfTraceFire{(\eventEffBar{\ell}{v_1},\, v_2)},T_2}
    \\\\\to\\\\
    \confPfTop
      {\mathcal{S}'}
      {A}{\overline{x=v}}
      {rendered}{T_1,\semPfTraceFireSuc{(\eventEffBar{\ell}{v_1},\, v_2)}{(U_1, F_1), \ldots, (U_k, F_k)}}{T_2}
  }
\captionof{figure}{Configuration stepping: event firing.}
\label{fig:app-config-step-fire}
\end{placedfigure}

\bigskip

\begin{placedfigure}
\Rule{waiting to rendered}
  {
    A: (\synComponent{A}{\overline{x:\tau}}{\tau_r}{p}, \Delta, \Gamma_s) \in \Sigma \\
    \mathcal{S} = (X_1, X_2, V_{s1}, \mathcal{L}, \mathcal{C}) \\
    \mathcal{S}' = (X_2 \cup \{x_i \mid v_{1i} \neq v_{2i}\}, \{x_i \mid v_{1i} \neq v_{2i}\}, V_{s2}, \mathcal{L}, \mathcal{C}) \\
    \semPfDerivDecl{\Gamma_s}{X_2 \cup \{x_i \mid v_{1i} \neq v_{2i}\}}{V_{s1}}{\overline{x = v_2}}{``"}{p}{v_r}{T_2}{V_{s2}} \\
    \exists i \;s.t. \; v_{1i} \neq v_{2i} \\
    \valTyDeriv{v_r}{\tau_r} \\
    \overline{\valTyDeriv{v_2}{\tau_2}}\\
  }{
    \confPfTop
      {\mathcal{S}}
      {A}{\overline{x=v_1}}
      {waiting}{T_1}{\cdot}
    \\\\\to\\\\
    \confPfTop
      {\mathcal{S}'}
      {A}{\overline{x=v_2}}
      {rendered}{\semPfTraceArgs{\overline x}{\overline v_1}{\overline v_2}}{T_2}
  }
\captionof{figure}{Configuration stepping: re-rendering.}
\label{fig:app-config-step-rerender}
\end{placedfigure}

\subsection{Configurations}




\begin{theorem}[Preservation for Configurations]
\label{appendix:thm:preservation-configurations}

Given a well-typed configuration

$$
\confPfTop
  {\mathcal{S}}
  {A}{\overline{x=v}}
  {status$_1$}{T_1}{T_2}
$$
where $\mathcal{S} = (X_1, X_2, V_{s1}, \mathcal{L}, \mathcal{C})$
\begin{itemize}
  \item \tyDerivComp{\Sigma}{\tinyComp{A}}{(\Delta, \Gamma_s)}
  \item \justifies{\Delta, X_1}{T_1,T_2}
  \item \justifiesVarSet{\Delta,X_1}{T_1}{X_2}
  \item \valTyEnvDeriv{\Gamma_s}{V_{s1}}
  \item $\overline{\valTyDeriv{v}{\tau}}$
  \item For every $\eventEffBar{\ell}{v_1} \in \mathrm{dom}(\mathcal{L})$ and every
        $(c, m, F) \in \mathcal{L}(\eventEffBar{\ell}{v_1})$,
        $\valTyDeriv{c}{\tyArrow{\Sigma_E(\eventEffBar{\ell}{v_1})}{\tyUnit}[F]}$.
  \item For every firing
        $(\eventEffBar{\ell}{v_1},\, v_2) \in \mathcal{E}$ supplied by the
        external scheduler,
        $\valTyDeriv{v_2}{\Sigma_E(\eventEffBar{\ell}{v_1})}$.
  \item For every $\semPfTraceFire{(\eventEffBar{\ell}{v_1},\, v_2)}$ record
        appearing in $T_2$,
        $\valTyDeriv{v_2}{\Sigma_E(\eventEffBar{\ell}{v_1})}$.
\end{itemize}

If the configuration makes a step

$$
  \confPfTop
      {\mathcal{S}}
      {A}{\overline{x=v}}
      {status$_1$}{T_1}{T_2}$$
  $$\to$$
$$
  \confPfTop
      {\mathcal{S}'}
      {A}{\overline{x=v'}}
      {status$_2$}{T_3}{T_4}
$$
where $\mathcal{S}' = (X_3, X_4, V_{s2}, \mathcal{L}', \mathcal{C}')$

Then the resulting configuration is also well-typed:
\begin{itemize}
  \item \justifies{\Delta, X_3}{T_3,T_4}
  \item \justifiesVarSet{\Delta,X_3}{T_3}{X_4}
  \item \valTyEnvDeriv{\Gamma_s}{V_{s2}}
  \item $\overline{\valTyDeriv{v'}{\tau}}$
  \item For every $\eventEffBar{\ell}{v_1} \in \mathrm{dom}(\mathcal{L}')$ and every
        $(c, m, F) \in \mathcal{L}'(\eventEffBar{\ell}{v_1})$,
        $\valTyDeriv{c}{\tyArrow{\Sigma_E(\eventEffBar{\ell}{v_1})}{\tyUnit}[F]}$.
  \item For every $\semPfTraceFire{(\eventEffBar{\ell}{v_1},\, v_2)}$ record
        appearing in $T_4$,
        $\valTyDeriv{v_2}{\Sigma_E(\eventEffBar{\ell}{v_1})}$.
\end{itemize}
\end{theorem}

\begin{proof}
By case analysis on the configuration transition rule.

\begin{itemize}
  \item \textbf{Case No Updates.}

  \Rule*
    { }{
      \confPfTop
          {\mathcal{S}}
          {A}{\overline{x=v}}
          {rendered}{T_1}{\semPfTraceReturn{y}}
      \\\\\to\\\\
      \confPfTop
          {\mathcal{S}}
          {A}{\overline{x=v}}
          {waiting}{T_1,\semPfTraceReturn{y}}{\cdot}
    }

  From the precondition, we have:
  \begin{itemize}
    \item \justifies{\Delta, X_1}{T_1, \semPfTraceReturn{y}}
    \item \justifiesVarSet{\Delta, X_1}{T_1}{X_2}
    \item \valTyEnvDeriv{\Gamma_s}{V_{s1}}
    \item $\overline{\valTyDeriv{v}{\tau}}$
  \end{itemize}

  The resulting configuration has $X_3 = X_1$, $X_4 = X_2$, $V_{s2} = V_s$,
  $T_3 = T_1, \semPfTraceReturn{y}$, and $T_4 = \cdot$.

  We verify:
  \begin{itemize}
    \item \justifies{\Delta, X_1}{T_1, \semPfTraceReturn{y}, \cdot} ---
          This is equivalent to \justifies{\Delta, X_1}{T_1, \semPfTraceReturn{y}},
          which holds by the precondition.
    \item \justifiesVarSet{\Delta, X_1}{T_1, \semPfTraceReturn{y}}{X_2} ---
          By inversion on the justification rules for return traces, we have
          \justifiesVarSet{\Delta, X_1}{T_1}{X_2} from the precondition, and
          the return trace does not modify the variable set.
    \item \valTyEnvDeriv{\Gamma_s}{V_{s1}} --- unchanged from precondition.
    \item $\overline{\valTyDeriv{v}{\tau}}$ --- unchanged from precondition.
    \item $\mathcal{L}$ well-formed --- unchanged from precondition.
  \end{itemize}

  \item \textbf{Case Flush Let.}

  \Rule*
    { }{
      \confPfTop
        {\mathcal{S}}
        {A}{\overline{x=v}}
        {rendered}{T_1}{\semPfTraceLet{y}{\overline z},T_2}
      \\\\\to\\\\
      \confPfTop
        {\mathcal{S}}
        {A}{\overline{x=v}}
        {rendered}{T_1,\semPfTraceLet{y}{\overline z}}{T_2}
    }

  From the precondition, we have:
  \begin{itemize}
    \item \justifies{\Delta, X_1}{T_1, \semPfTraceLet{y}{\overline z}, T_2}
    \item \justifiesVarSet{\Delta, X_1}{T_1}{X_2}
    \item \valTyEnvDeriv{\Gamma_s}{V_{s1}}
    \item $\overline{\valTyDeriv{v}{\tau}}$
  \end{itemize}

  The resulting configuration has $X_3 = X_1$, $X_4 = X_2$, $V_{s2} = V_{s1}$,
  $T_3 = T_1, \semPfTraceLet{y}{\overline z}$, and $T_4 = T_2$.

  We verify:
  \begin{itemize}
    \item \justifies{\Delta, X_1}{T_1, \semPfTraceLet{y}{\overline z}, T_2} ---
          This holds directly from the precondition.
    \item \justifiesVarSet{\Delta, X_1}{T_1, \semPfTraceLet{y}{\overline z}}{X_2} ---
          By the justification rule for let traces,
          \justifiesVarSet{\Delta, X_1}{T_1, \semPfTraceLet{y}{\overline z}}{X_2}
          follows from \justifiesVarSet{\Delta, X_1}{T_1}{X_2} (precondition).
    \item \valTyEnvDeriv{\Gamma_s}{V_{s1}} --- unchanged.
    \item $\overline{\valTyDeriv{v}{\tau}}$ --- unchanged.
    \item $\mathcal{L}$ well-formed --- unchanged.
  \end{itemize}

  \item \textbf{Case Flush Update.}

  \Rule*
    {
      \mathcal{S} = (X_1, X_{21}, V_{s1}, \mathcal{L}_1, \mathcal{C}_1) \\
      \mathcal{S}' = (X_1, X_{22}, V_{s2}, \mathcal{L}_2, \mathcal{C}_2) \\
      \listenersafe{U} \\
      \setterQueueDeriv
        {(X_{21},V_{s1}, \mathcal{L}_1, \mathcal{C}_1)}
        {U}
        {(X_{22},V_{s2}, \mathcal{L}_2, \mathcal{C}_2)}
    }{
      \confPfTop
        {\mathcal{S}}
        {A}{\overline{x=v}}
        {rendered}{T_1}{\semPfTraceOn{\overline y}{U},T_2}
      \\\\\to\\\\
      \confPfTop
        {\mathcal{S}'}
        {A}{\overline{x=v}}
        {rendered}{T_1,\semPfTraceOn{\overline y}{U}}{T_2}
    }

  From the precondition, we have:
  \begin{itemize}
    \item \justifies{\Delta, X_1}{T_1, \semPfTraceOn{\overline y}{U}, T_2}
    \item \justifiesVarSet{\Delta, X_1}{T_1}{X_{21}}
    \item \valTyEnvDeriv{\Gamma_s}{V_{s1}}
    \item $\overline{\valTyDeriv{v}{\tau}}$
  \end{itemize}

  We need the following lemma:

  \textbf{Lemma (Flush Queue Preservation):}
  If \valTyEnvDeriv{\Gamma_s}{V_{s1}}, $\mathcal{L}_1$ is well-formed in the
  sense of the new well-typedness bullet --- that is, for every
  $\eventEffBar{\ell}{v_1} \in \mathrm{dom}(\mathcal{L}_1)$ and every
  $(c, m, F) \in \mathcal{L}_1(\eventEffBar{\ell}{v_1})$,
  $\valTyDeriv{c}{\tyArrow{\Sigma_E(\eventEffBar{\ell}{v_1})}{\tyUnit}[F]}$ ---
  and
  \setterQueueDeriv{(X_1,V_{s1},\mathcal{L}_1,\mathcal{C}_1)}{U}{(X_2,V_{s2},\mathcal{L}_2,\mathcal{C}_2)},
  then \valTyEnvDeriv{\Gamma_s}{V_{s2}} and $\mathcal{L}_2$ is well-formed in
  the same sense.

  \textit{Proof of lemma:} By induction on the derivation of
  \setterQueueDeriv{(X_1,V_{s1},\mathcal{L}_1,\mathcal{C}_1)}{U}{(X_2,V_{s2},\mathcal{L}_2,\mathcal{C}_2)}.
  \begin{itemize}
    \item Base case (\rulename{fq nil}): $U = \cdot$. Then $V_{s2} = V_{s1}$
          and $\mathcal{L}_2 = \mathcal{L}_1$, so both results hold trivially.
    \item Inductive case (\rulename{fq setter}): $U = setter_x(f), U'$. By the
          rule for setter queue derivation, we have
          \setterderiv{V_{s1}}{setter_x(f)}{V_{s1}'} for some $V_{s1}'$, and
          \setterQueueDeriv{(X_1 \cup \{x\}, V_{s1}',\mathcal{L}_1,\mathcal{C}_1)}{U'}{(X_2, V_{s2},\mathcal{L}_2,\mathcal{C}_2)}.

          From \valTyDeriv{setter_x}{(\tySetter{x})}, we know that $f : \tau \to \tau$
          where $x : \tau \in \Gamma_s$. The setter application applies $f$ to the
          current value of $x$ (of type $\tau$), producing a new value of type $\tau$.
          Thus \valTyEnvDeriv{\Gamma_s}{V_{s1}'}.  $\mathcal{L}$ is unchanged
          at this step, so its well-formedness carries through.

          By the induction hypothesis, \valTyEnvDeriv{\Gamma_s}{V_{s2}} and
          $\mathcal{L}_2$ is well-formed.
    \item Inductive case (\rulename{fq cancel}): $U = \mathsf{cancel}(\eventEffBar{\ell}{v}), U'$.
          $V_s$ and $\mathcal{L}$ are unchanged at this step (only $\mathcal{C}$
          gains a multiset entry), so the recursive premise still has a
          well-typed $V_s$ and a well-formed $\mathcal{L}$, and the induction
          hypothesis yields the result.
    \item Inductive case (\rulename{fq remove}): $U = \mathsf{remove}(\eventEffBar{\ell}{v}), U'$.
          The rule sets
          $\mathcal{L}'_1 = \mathcal{L}_1[\eventEffBar{\ell}{v} \mapsto \emptyset]$.
          For the bucket at $\eventEffBar{\ell}{v}$, well-formedness is vacuous
          because the bucket is empty; every other bucket is unchanged from
          $\mathcal{L}_1$, hence retains its well-formedness witness.  $V_s$
          is unchanged.  Induction hypothesis on the recursive premise
          \setterQueueDeriv{(X_1, V_{s},\mathcal{L}'_1,\mathcal{C})}{U'}{(X_2, V_{s2},\mathcal{L}_2,\mathcal{C}_2)}
          gives the result.
    \item Inductive case (\rulename{fq listen}): $U = \mathsf{listen}(\eventEffBar{\ell}{v},\, c,\, m,\, F), U'$.
          The rule's own premise is
          $\valTyDeriv{c}{\tyArrow{\Sigma_E(\eventEffBar{\ell}{v})}{\tyUnit}[F]}$,
          which is exactly the well-formedness obligation for the new entry.
          That premise is not free standing: at the lemma's call sites a
          witness for it is supplied by the \emph{Listen typing inversion}
          lemma applied to the ambient $\effQueueDeriv{F_i}{U_i}$ derivation
          (in Fire Successful this is the rule premise
          $\overline{\effQueueDeriv{F_i}{U_i}}$ on each listener's evaluated
          queue).  Every other listener record in $\mathcal{L}'_1 =
          \mathcal{L}_1[\eventEffBar{\ell}{v} \mapsto \mathcal{L}_1(\eventEffBar{\ell}{v})
          \cup \{(c, m, F)\}]$ was already in $\mathcal{L}_1$ and so is
          well-formed by the precondition.  $V_s$ is unchanged.  Induction
          hypothesis on the recursive premise yields the result.
  \end{itemize}

  Applying the lemma, we get \valTyEnvDeriv{\Gamma_s}{V_{s2}} and that
  $\mathcal{L}_2$ is well-formed.

  By inversion on \justifies{\Delta, X_1}{T_1, \semPfTraceOn{\overline y}{U}, T_2}:
  \begin{itemize}
    \item $\overline y \subseteq X_1$
    \item $\exists F. (\overline{\causesderiv{\Delta}{y_i}{F}} \text{ and } \effQueueDeriv{F}{U})$
  \end{itemize}

  By the rule for justifying variable change sets with effect traces, if
  $$U = setter_{x_1}(f_1), \ldots, setter_{x_n}(f_n)$$
  then flushing the update extends the variable set to include
  $\{x_1, \ldots, x_n\}$.  This matches
  $X_{22} = X_{21} \cup \{x_1, \ldots, x_n\}$.

  We verify:
  \begin{itemize}
    \item \justifies{\Delta, X_1}{T_1, \semPfTraceOn{\overline y}{U}, T_2} ---
          holds by precondition.
    \item \justifiesVarSet{\Delta, X_1}{T_1, \semPfTraceOn{\overline y}{U}}{X_{22}} ---
          From \justifiesVarSet{\Delta, X_1}{T_1}{X_{21}} and the rule for effect traces,
          we get $\justifiesVarSet{\Delta, X_1}{T_1, \semPfTraceOn{\overline y}{U}}{X_{21} \cup \{x_1, \ldots, x_n\}} = X_{22}$.
    \item \valTyEnvDeriv{\Gamma_s}{V_{s2}} --- from lemma above.
    \item $\mathcal{L}_2$ is well-formed --- from lemma above.
    \item $\overline{\valTyDeriv{v}{\tau}}$ --- unchanged.
  \end{itemize}

  \item \textbf{Case waiting to rendered.}

  \Rule*
    {
      A: (\synComponent{A}{\overline{x:\tau}}{\tau_r}{p}, \Delta, \Gamma_s) \in \Sigma \\
      \mathcal{S} = (X_1, X_2, V_{s1}, \mathcal{L}, \mathcal{C}) \\
      \mathcal{S}' = (X_2 \cup \{x_i \mid v_{1i} \neq v_{2i}\}, \{x_i \mid v_{1i} \neq v_{2i}\}, V_{s2}, \mathcal{L}, \mathcal{C}) \\
      \semPfDerivDecl{\Gamma_s}{X_2 \cup \{x_i \mid v_{1i} \neq v_{2i}\}}{V_{s1}}{\overline{x = v_2}}{``"}{p}{v_r}{T_2}{V_{s2}} \\
      \exists i \;s.t. \; v_{1i} \neq v_{2i} \\
      \valTyDeriv{v_r}{\tau_r} \\
      \overline{\valTyDeriv{v_2}{\tau_2}}
    }{
      \confPfTop
        {\mathcal{S}}
        {A}{\overline{x=v_1}}
        {waiting}{T_1}{\cdot}
      \\\\\to\\\\
      \confPfTop
        {\mathcal{S}'}
        {A}{\overline{x=v_2}}
        {rendered}{\semPfTraceArgs{\overline x}{\overline v_1}{\overline v_2}}{T_2}
    }

  This is the re-render case. From the precondition:
  \begin{itemize}
    \item \justifies{\Delta, X_1}{T_1, \cdot}
    \item \justifiesVarSet{\Delta, X_1}{T_1}{X_2}
    \item \valTyEnvDeriv{\Gamma_s}{V_{s1}}
    \item $\overline{\valTyDeriv{v_{1}}{\tau}}$
  \end{itemize}

  From \tyDerivComp{\Sigma}{\tinyComp{A}}{(\Delta, \Gamma_s)}, by inversion:
  \begin{itemize}
    \item \tyDerivDecls{\Sigma}{\Gamma}{\Delta}{p}{\tau_r} where $\Gamma$ contains the argument bindings
  \end{itemize}

  The new arguments $\overline{x = v_2}$ are well-typed: $\overline{\valTyDeriv{v_{2}}{\tau}}$
  (given in the rule premises).

  Let $X' = \{x_i \mid v_{1i} \neq v_{2i}\}$ be the set of changed arguments.
  Since $\exists i \;s.t. \; v_{1i} \neq v_{2i}$, we have $X' \neq \emptyset$, so
  $X_2 \cup X' \neq \emptyset$ as well.

  Apply Preservation for Declarations (Theorem~\ref{appendix:thm:preservation-declarations})
  to the semantic premise, taking the changed-variable set
  to be $X_2 \cup X'$ (matching the $X$ used in the semantic derivation):
  \begin{itemize}
    \item \semPfDerivDecl{\Gamma_s}{X_2 \cup X'}{V_{s1}}{\overline{x = v_2}}{``"}{p}{v_r}{T_2}{V_{s2}}
    \item \tyDerivDecls{\Sigma}{\Gamma}{\Delta}{p}{\tau_r}
    \item \valTyEnvDeriv{\Gamma_s}{V_{s1}} (from precondition)
    \item \valTyEnvDeriv{\Gamma}{\overline{x = v_2}} (from well-typed new arguments)
    \item $X_2 \cup X' \neq \emptyset$
  \end{itemize}

  This gives us:
  \begin{itemize}
    \item \valTyDeriv{v_r}{\tau_r}
    \item \justifies{\Delta, X_2 \cup X'}{T_2}
    \item \valTyEnvDeriv{\Gamma_s}{V_{s2}}
  \end{itemize}

  The resulting configuration has:
  \begin{itemize}
    \item $X_3 = X_2 \cup X'$ --- the union of the previous render's
          accumulated changed-variable set and this render's changed arguments.
          This is exactly the $X$ under which the new body produced $T_2$.
    \item $X_4 = X' = \{x_i \mid v_{1i} \neq v_{2i}\}$
    \item $T_3 = \semPfTraceArgs{\overline x}{\overline v_1}{\overline v_2}$
    \item $T_4 = T_2$
  \end{itemize}

  We verify:
  \begin{itemize}
    \item \justifies{\Delta, X_2 \cup X'}{\semPfTraceArgs{\overline x}{\overline v_1}{\overline v_2}, T_2} ---
          By the justification rule for argument traces (which adds an \textsc{arg} record
          to any justified trace under the same $X$), this reduces to
          \justifies{\Delta, X_2 \cup X'}{T_2}, which we obtained from Declarations
          Preservation above.
    \item \justifiesVarSet{\Delta, X_2 \cup X'}{\semPfTraceArgs{\overline x}{\overline v_1}{\overline v_2}}{X'} ---
          By the justification rule for argument traces, the variable-change set produced
          by an \textsc{arg} record is exactly $\{x_i \mid v_{1i} \neq v_{2i}\} = X'$,
          independent of the input $X$.
    \item \valTyEnvDeriv{\Gamma_s}{V_{s2}} --- from Declarations Preservation above.
    \item $\overline{\valTyDeriv{v_{2}}{\tau}}$ --- from the rule premises.
    \item $\mathcal{L}$ well-formed --- $\mathcal{L}$ is unchanged across this
          transition, so its well-formedness carries from the precondition.
  \end{itemize}

  \item \textbf{Case Receive Ext.}

  \Rule*
    {
      \mathcal{S}  = (X_1, X_2, V_s, \mathcal{L}, \mathcal{C})  \\
      \mathcal{S}' = (X_1, X_2, V_s, \mathcal{L}, \mathcal{C}) \\
      \mathcal{E} \;\Rightarrow_{\mathsf{ext}}\;
        (\eventEffBar{\ell_1}{v_{11}},\, v_{21}), \ldots,
        (\eventEffBar{\ell_n}{v_{1n}},\, v_{2n}) \\
      \overline{\valTyDeriv{v_{2i}}{\Sigma_E(\eventEffBar{\ell_i}{v_{1i}})}} \\
      T_{fire} =
        \semPfTraceFire{(\eventEffBar{\ell_1}{v_{11}},\, v_{21})}, \ldots,
        \semPfTraceFire{(\eventEffBar{\ell_n}{v_{1n}},\, v_{2n})}\\
      T_{ext} =
        \semPfTraceExt{
          (\eventEffBar{\ell_1}{v_{11}},\, v_{21}), \ldots,
          (\eventEffBar{\ell_n}{v_{1n}},\, v_{2n})}
    }{
      \confPfTop
        {\mathcal{S}}
        {A}{\overline{x=v}}
        {rendered}{\semPfTraceArgs{\overline x}{\overline{v'_1}}{\overline{v'_2}},}{T_2}
      \\\\\to\\\\
      \confPfTop
        {\mathcal{S}'}
        {A}{\overline{x=v}}
        {rendered}{\semPfTraceArgs{\overline x}{\overline{v'_1}}{\overline{v'_2}},
          T_{ext}}{T_{fire}, T_2}
    }

  Let $T_1 = \semPfTraceArgs{\overline x}{\overline{v'_1}}{\overline{v'_2}}$ for
  brevity (the bare-subscript $v_1, v_2$ slots are reserved for the per-event
  discriminant and payload pattern).  From the precondition:
  \begin{itemize}
    \item \justifies{\Delta, X_1}{T_1, T_2}
    \item \justifiesVarSet{\Delta, X_1}{T_1}{X_2}
    \item \valTyEnvDeriv{\Gamma_s}{V_s}, $\overline{\valTyDeriv{v}{\tau}}$,
          $\mathcal{L}$ well-formed
    \item Each $(\eventEffBar{\ell_i}{v_{1i}},\, v_{2i}) \in \mathcal{E}$
          satisfies $\valTyDeriv{v_{2i}}{\Sigma_E(\eventEffBar{\ell_i}{v_{1i}})}$
          --- the external-scheduler well-typedness bullet.
  \end{itemize}

  The resulting configuration has $X_3 = X_1$, $X_4 = X_2$, and $V_s$,
  $\mathcal{L}$, $\mathcal{C}$ all unchanged;
  $T_3 = T_1, T_{ext}$, and $T_4 = T_{fire}, T_2$.

  We verify:
  \begin{itemize}
    \item \justifies{\Delta, X_1}{T_3, T_4} ---
          By the justification rule for the \textsc{ext}-followed-by-$R_i$ block,
          we can extend a justified
          tail with an \textsc{ext} record followed by $n$ records $R_1, \ldots,
          R_n$ where each $R_i$ is one of $\semPfTraceFire{}$,
          $\semPfTraceFireCxl{}$, $\semPfTraceFireSuc{}$.  Choose each
          $R_i = \semPfTraceFire{(\eventEffBar{\ell_i}{v_{1i}},\, v_{2i})}$;
          the additional witness obligation on $\semPfTraceFireSuc{}$ records is
          vacuous because none are chosen.  The tail-justification premise is
          \justifies{\Delta, X_1}{T_2}, which follows from
          \justifies{\Delta, X_1}{T_1, T_2} by inversion using the justification
          rule for argument traces (which discharges the leading $T_1$).
    \item \justifiesVarSet{\Delta, X_1}{T_3}{X_2} ---
          The rules for variable-set justification on \textsc{ext}
          and on \textsc{fire} are both the identity in the
          running variable set: extending the prefix by an \textsc{ext} or by
          a \textsc{fire} preserves the set.  From the precondition
          \justifiesVarSet{\Delta, X_1}{T_1}{X_2}, extending by one \textsc{ext}
          record yields \justifiesVarSet{\Delta, X_1}{T_3}{X_2}.
    \item \valTyEnvDeriv{\Gamma_s}{V_s} --- unchanged.
    \item $\overline{\valTyDeriv{v}{\tau}}$ --- unchanged.
    \item $\mathcal{L}$ well-formed --- $\mathcal{L}$ is unchanged.
    \item External scheduler payload well-typedness --- the firings consumed
          by this step are precisely those whose payload-typing premise was
          discharged by the rule premise
          $\overline{\valTyDeriv{v_{2i}}{\Sigma_E(\eventEffBar{\ell_i}{v_{1i}})}}$.
          Any further firings remaining in $\mathcal{E}$ continue to satisfy
          the precondition's well-typedness clause.
  \end{itemize}

  \item \textbf{Case Fire Cancel.}

  \Rule*
    {
      \mathcal{S}  = (X_1, X_2, V_s, \mathcal{L}, \mathcal{C})  \\
      \mathcal{S}' = (X_1, X_2, V_s, \mathcal{L}, \mathcal{C}') \\
      \mathcal{C}(\eventEffBar{\ell}{v_1}) \geq 1 \\
      \mathcal{C}' = \mathcal{C} \setminus \{\!\!\{\,\eventEffBar{\ell}{v_1}\,\}\!\!\}
    }{
      \confPfTop
        {\mathcal{S}}
        {A}{\overline{x=v}}
        {rendered}{T_1}{\semPfTraceFire{(\eventEffBar{\ell}{v_1},\, v_2)},T_2}
      \\\\\to\\\\
      \confPfTop
        {\mathcal{S}'}
        {A}{\overline{x=v}}
        {rendered}{T_1,\semPfTraceFireCxl{(\eventEffBar{\ell}{v_1},\, v_2)}}{T_2}
    }

  From the precondition:
  \begin{itemize}
    \item \justifies{\Delta, X_1}{T_1, \semPfTraceFire{(\eventEffBar{\ell}{v_1},\, v_2)}, T_2}
    \item \justifiesVarSet{\Delta, X_1}{T_1}{X_2}
    \item \valTyEnvDeriv{\Gamma_s}{V_s}, $\overline{\valTyDeriv{v}{\tau}}$,
          $\mathcal{L}$ well-formed
  \end{itemize}

  The resulting configuration has $X_3 = X_1$, $X_4 = X_2$, $V_{s2} = V_s$,
  $\mathcal{L}' = \mathcal{L}$,
  $\mathcal{C}' = \mathcal{C} \setminus \{\!\!\{\,\eventEffBar{\ell}{v_1}\,\}\!\!\}$,
  $T_3 = T_1, \semPfTraceFireCxl{(\eventEffBar{\ell}{v_1},\, v_2)}$, and $T_4 = T_2$.

  We verify:
  \begin{itemize}
    \item \justifies{\Delta, X_1}{T_3, T_4} ---
          The combined trace is
          $T_1, \semPfTraceFireCxl{(\eventEffBar{\ell}{v_1},\, v_2)}, T_2$.  We use the
          \textsc{ext}/$R$-block rule with the degenerate
          \emph{single}-record extension: choose
          $R_1 = \semPfTraceFireCxl{(\eventEffBar{\ell}{v_1},\, v_2)}$, no
          $\semPfTraceFireSuc{}$ records present, so the witness obligation is
          vacuous.  We need to peel the \textsc{ext}-prefix that the rule
          requires.  Inversion on the precondition's justification of
          $T_1, \semPfTraceFire{(\eventEffBar{\ell}{v_1},\, v_2)}, T_2$ exposes the
          \textsc{ext} prefix in $T_1$ (call it $T_1 = T_1', \semPfTraceExt{\ldots}$
          followed by some $R$-block then a tail $T''$ with
          $\semPfTraceFire{(\eventEffBar{\ell}{v_1},\, v_2)} \in $ the $R$-block,
          and then $T_2$); replace the $\semPfTraceFire{}$ at position of
          $(\eventEffBar{\ell}{v_1},\, v_2)$ by
          $\semPfTraceFireCxl{(\eventEffBar{\ell}{v_1},\, v_2)}$ inside that
          $R$-block: this is still a valid $R_i$ choice, and no witness changes.
          The tail $T''$ followed by $T_2$ is justified by the inner premise of
          the original rule application.  Therefore the new combined trace is
          also justified.

          The inversion uniquely identifies the \textsc{ext} prefix because at
          most one \textsc{ext} record can appear in $T_1$ at any time within
          a render cycle: Receive Ext requires $T_1$ to be exactly the
          single argument record $\semPfTraceArgs{\ldots}{\ldots}{\ldots}$, so once one Ext block has been
          deposited, a second Receive Ext cannot fire until the next
          waiting-to-rendered transition resets $T_1$ to $\semPfTraceArgs{\ldots}{\ldots}{\ldots}$
          alone.  Hence the Ext prefix introducing this $\semPfTraceFire{}$ is
          uniquely the only one in $T_1$.
    \item \justifiesVarSet{\Delta, X_1}{T_3}{X_2} ---
          The variable-set justification rule on \textsc{fire-cxl}
          is the identity in the variable set.  From the
          precondition
          \justifiesVarSet{\Delta, X_1}{T_1}{X_2}, extending the prefix by
          $\semPfTraceFireCxl{(\eventEffBar{\ell}{v_1},\, v_2)}$ yields
          \justifiesVarSet{\Delta, X_1}{T_3}{X_2}.
    \item \valTyEnvDeriv{\Gamma_s}{V_s} --- unchanged.
    \item $\overline{\valTyDeriv{v}{\tau}}$ --- unchanged.
    \item $\mathcal{L}$ well-formed --- $\mathcal{L}$ is unchanged; $\mathcal{C}$
          is a multiset of event labels with no typing content, so its shrinkage
          carries no obligation.
  \end{itemize}

  \item \textbf{Case Fire Successful.}

  \Rule*
    {
      \mathcal{S}  = (X_1, X_2,  V_s,  \mathcal{L},  \mathcal{C})  \\
      \mathcal{S}' = (X_1, X_2', V_s', \mathcal{L}'', \mathcal{C}') \\
      \mathcal{C}(\eventEffBar{\ell}{v_1}) = 0 \\
      \mathcal{L}(\eventEffBar{\ell}{v_1}) = \{(\Closure{x_i}{e_i}{V_i}, m_i, F_i)\}_{i=1}^{k} \\
      \valTyDeriv{v_2}{\Sigma_E(\eventEffBar{\ell}{v_1})} \\
      \overline{\semDerivExpr{V_i,\; x_i = v_2}{e_i}{()}{U_i}} \\
      \overline{\listenersafe{U_i}} \\
      \overline{\effQueueDeriv{F_i}{U_i}} \\
      \mathcal{L}' = \mathcal{L}\bigl[\eventEffBar{\ell}{v_1} \mapsto
        \{(\Closure{x_i}{e_i}{V_i}, m_i, F_i) \mid m_i = \always{}\}\bigr] \\
      \setterQueueDeriv
        {(X_2, V_s, \mathcal{L}', \mathcal{C})}
        {U_1, \ldots, U_k}
        {(X_2', V_s', \mathcal{L}'', \mathcal{C}')}
    }{
      \confPfTop
        {\mathcal{S}}
        {A}{\overline{x=v}}
        {rendered}{T_1}{\semPfTraceFire{(\eventEffBar{\ell}{v_1},\, v_2)},T_2}
      \\\\\to\\\\
      \confPfTop
        {\mathcal{S}'}
        {A}{\overline{x=v}}
        {rendered}{T_1,\semPfTraceFireSuc{(\eventEffBar{\ell}{v_1},\, v_2)}{(U_1, F_1), \ldots, (U_k, F_k)}}{T_2}
    }

  From the precondition:
  \begin{itemize}
    \item \justifies{\Delta, X_1}{T_1, \semPfTraceFire{(\eventEffBar{\ell}{v_1},\, v_2)}, T_2}
    \item \justifiesVarSet{\Delta, X_1}{T_1}{X_2}
    \item \valTyEnvDeriv{\Gamma_s}{V_s}, $\overline{\valTyDeriv{v}{\tau}}$,
          $\mathcal{L}$ well-formed
    \item External-scheduler payload well-typedness: the firing
          $(\eventEffBar{\ell}{v_1},\, v_2)$ this step consumes was introduced
          into the trace by a prior Receive Ext step that discharged
          $\valTyDeriv{v_2}{\Sigma_E(\eventEffBar{\ell}{v_1})}$, which is
          therefore available here and matches the rule's payload-typing
          premise.
  \end{itemize}

  We establish each of the obligations for the resulting configuration in
  turn.

  \medskip\noindent\textit{Step 1: each $U_i$ has an effect-queue witness.}
  Let $\tau = \Sigma_E(\eventEffBar{\ell}{v_1})$ be the payload type fixed by
  the signature at the firing label.  This step is delivered directly by the
  rule premise $\overline{\effQueueDeriv{F_i}{U_i}}$.  As noted in the rule,
  this premise is justified by
  Theorem~\ref{appendix:thm:preservation-expressions} applied to $e_i$ under
  the listener's argument-extended typing context $\Gamma_{V_i}, x_i{:}\tau$.
  Applying that theorem requires two facts:
  \begin{itemize}
    \item $\valTyDeriv{\Closure{x_i}{e_i}{V_i}}{\tyArrow{\tau}{\tyUnit}[F_i]}$
          --- the listener's closure has the expected arrow type.  This is
          exactly the $\mathcal{L}$-well-formedness bullet (from the
          precondition) instantiated at $\eventEffBar{\ell}{v_1}$ for each
          $(c_i, m_i, F_i) \in \mathcal{L}(\eventEffBar{\ell}{v_1})$.  The
          listener's arrow type is supplied by $\Sigma_E$ pinning down the
          carried type.
    \item $\valTyDeriv{v_2}{\tau}$ where $v_2$ is the payload actually
          delivered to the closure --- i.e.\ the value assigned to $x_i$ in
          the rule's premise
          $\semDerivExpr{V_i, x_i = v_2}{e_i}{()}{U_i}$.  This is exactly the
          rule's payload-typing premise
          $\valTyDeriv{v_2}{\Sigma_E(\eventEffBar{\ell}{v_1})}$, itself
          inherited from the precondition's external-scheduler well-typedness
          bullet (the firing was placed in the trace by an earlier Receive
          Ext step which discharged precisely this judgement).  Because the
          payload occupies its own syntactic slot in the trace record, the
          discriminant tuple $\overline{v_1}$ never enters the listener
          context, and the typing obligation $\valTyDeriv{v_2}{\tau}$ matches
          the rule's hypothesis directly.
  \end{itemize}

  \medskip\noindent\textit{Step 2: $\mathcal{L}'$ is well-formed.}  $\mathcal{L}'$
  is obtained from $\mathcal{L}$ by replacing the bucket at
  $\eventEffBar{\ell}{v_1}$ with the sub-set
  $\{(\Closure{x_i}{e_i}{V_i}, m_i, F_i) \mid m_i = \always{}\}$.  Every record
  retained in $\mathcal{L}'(\eventEffBar{\ell}{v_1})$ was already in
  $\mathcal{L}(\eventEffBar{\ell}{v_1})$ and so has a typing witness by the
  $\mathcal{L}$-well-formedness assumption.  All other buckets are unchanged.
  Hence $\mathcal{L}'$ is well-formed.

  \medskip\noindent\textit{Step 3: $V_s'$ well-typed and $\mathcal{L}''$ well-formed.}
  Apply the Flush Queue Preservation lemma (extended above) to the rule premise
  \setterQueueDeriv{(X_2, V_s, \mathcal{L}', \mathcal{C})}{U_1, \ldots, U_k}{(X_2', V_s', \mathcal{L}'', \mathcal{C}')}.
  The lemma's preconditions are \valTyEnvDeriv{\Gamma_s}{V_s} (from the
  precondition) and that $\mathcal{L}'$ is well-formed (from Step~2).  The
  lemma's \rulename{fq listen} case requires a closure-typing witness for each
  newly registered listener; this is supplied by the \emph{Listen typing
  inversion} lemma applied to the rule's premises
  $\overline{\effQueueDeriv{F_i}{U_i}}$, which establishes
  $\valTyDeriv{c}{\tyArrow{\Sigma_E(\eventEffBar{\ell}{v})}{\tyUnit}[F]}$ for
  every $\mathsf{listen}$ item in any $U_i$.  The lemma then yields
  \valTyEnvDeriv{\Gamma_s}{V_s'} and $\mathcal{L}''$ well-formed.

  \medskip\noindent\textit{Step 4: trace justification.}  The new combined trace
  is $T_1, \semPfTraceFireSuc{(\eventEffBar{\ell}{v_1},\, v_2)}{(U_1, F_1), \ldots, (U_k, F_k)}, T_2$,
  where each $F_j$ is the effect stored alongside the listener closure $c_j$ in
  $\mathcal{L}(\eventEffBar{\ell}{v_1})$ at the moment of firing, lifted from
  the rule premise $\mathcal{L}(\eventEffBar{\ell}{v_1}) = \{(c_j, m_j, F_j)\}_{j=1}^{k}$.
  We use the \textsc{ext}/$R$-block rule.  Invert the
  precondition's justification of
  $T_1, \semPfTraceFire{(\eventEffBar{\ell}{v_1},\, v_2)}, T_2$
  by identifying the \textsc{ext} prefix that introduced
  $\semPfTraceFire{(\eventEffBar{\ell}{v_1},\, v_2)}$ as one of its $R_i$ slots.
  As in Case Fire Cancel, this Ext prefix is unique: Receive Ext only fires
  when $T_1$ is the bare argument record, so $T_1$ contains at most one
  Ext block within a render cycle.
  In the new trace we replace that $R_i$ choice by
  $R_i' = \semPfTraceFireSuc{(\eventEffBar{\ell}{v_1},\, v_2)}{(U_1, F_1), \ldots, (U_k, F_k)}$,
  which is also a permitted $R_i$ choice.  The witness obligations imposed
  by the rule on $\semPfTraceFireSuc{}$ records are, for each $j$:
  \begin{itemize}
    \item $\effQueueDeriv{F_j}{U_j}$ --- discharged directly by the rule premise
          $\overline{\effQueueDeriv{F_j}{U_j}}$ of \rulename{Fire Successful}
          (the symbol $F_j$ on both sides refers to the same listener-stored
          effect), itself supported by the witnesses produced in Step~1.
    \item $\exists c.\;\valTyDeriv{c}{\tyArrow{\Sigma_E(\eventEffBar{\ell}{v_1})}{\tyUnit}[F_j]}$
          --- discharged by witnessing with the listener's own closure $c_j$.
          $\mathcal{L}$-well-formedness (from the precondition) instantiated at
          $\eventEffBar{\ell}{v_1}$ for the record $(c_j, m_j, F_j)$ yields
          $\valTyDeriv{c_j}{\tyArrow{\Sigma_E(\eventEffBar{\ell}{v_1})}{\tyUnit}[F_j]}$,
          tying $F_j$ to the event's declared payload type in $\Sigma_E$.
  \end{itemize}
  All other parts of the
  justification derivation, including the inner tail-justification, are unchanged.
  Hence \justifies{\Delta, X_1}{T_1, \semPfTraceFireSuc{(\eventEffBar{\ell}{v_1},\, v_2)}{(U_1, F_1), \ldots, (U_k, F_k)}, T_2}.

  \medskip\noindent\textit{Step 5: variable-set justification matches.}  The
  variable-set justification rule on \textsc{fire-suc} extends
  $X_2$ to $X_2 \cup \{y \mid \exists j.\; setter_y(f) \in U_j\}$.  We show
  this equals the $X_2'$ that the flush-queue derivation in Step~3 produced.

  By induction on the flush-queue derivation
  \setterQueueDeriv{(X_2, V_s, \mathcal{L}', \mathcal{C})}{U_1, \ldots, U_k}{(X_2', V_s', \mathcal{L}'', \mathcal{C}')},
  $X_2'$ is exactly $X_2$ extended by the set of $x$ such that $setter_x(f)$
  appears somewhere in $U_1, \ldots, U_k$: only \rulename{fq setter} extends
  the variable set, and it does so by adding the $x$ of its $setter_x(f)$
  item.  Cancel/remove/listen items leave the
  set unchanged.  Therefore
  $X_2' = X_2 \cup \{y \mid \exists j.\; setter_y(f) \in U_j\}$, which matches
  the result of the variable-set justification rule for \textsc{fire-suc} applied
  to the precondition \justifiesVarSet{\Delta, X_1}{T_1}{X_2}.  Hence
  \justifiesVarSet{\Delta, X_1}{T_3}{X_4} with $T_3 = T_1, \semPfTraceFireSuc{(\eventEffBar{\ell}{v_1},\, v_2)}{(U_1, F_1), \ldots, (U_k, F_k)}$
  and $X_4 = X_2'$.

  \medskip\noindent\textit{Step 6: argument values.}  Unchanged.

\end{itemize}
\end{proof}

\subsection{Declarations}

\begin{theorem}[Preservation for Declarations]
\label{appendix:thm:preservation-declarations}
Given
\begin{itemize}
  \item \semPfDerivDecl{\Gamma_s}{X}{V_s}{V}{c}{p}{v}{T}{V'}
  \item \tyDerivDecls{\Sigma}{\Gamma}{\Delta}{p}{\tau}
  \item \valTyEnvDeriv{\Gamma_s}{V_s}
  \item \valTyEnvDeriv{\Gamma}{V}
  \item $X \neq \emptyset$
\end{itemize}

Then
\valTyDeriv{v}{\tau}
and \justifies{\Delta, X}{T}
and \valTyEnvDeriv{\Gamma_s}{V'}
\end{theorem}
\begin{proof} By induction on the size of the semantic derivation

\begin{itemize}
  \item effect rerender, yes changes

  \Rule*
  {
    \semPfDerivDecl{\Gamma_s}{X}{V_s}{V}{c}{p}{v}{T}{V'} \\
    \semDerivExpr{V}{e}{()}{U} \\
    \overline x \cap X \neq \emptyset
  }
  {
    \semPfDerivDecl
      {\Gamma_s}{X}{V_s}{V}{c}
      {\synOn{\overline{x}}{e} p}
      {v}{\semPfTraceOn{\overline{x} \cap X}{U}, T}{V'}
  }

  \Rule*{
    \tyRuleOnDecl{\Sigma}{\Gamma}{\Delta}{e}{F}
    \\
    \tyDerivDecls{\Sigma}{\Gamma}{\Delta}{p}{\tau}
  }
  {
    \tyDerivDecls{\Sigma}{\Gamma}{\Delta}{\synOn{\overline x}{e} p}{\tau}
  }

  First, apply expressions preservation

  \begin{itemize}
    \item \semDerivExpr{V}{e}{()}{U}
    \item \tyDerivExpr{\Gamma}{e}{\tyUnit}{F}
    \item \valTyEnvDeriv{\Gamma}{V}
  \end{itemize}
  This gives us:
  (1) \valTyDeriv{()}{\tyUnit}
  and (2) \effQueueDeriv{F}{U}

  Now apply inductive hypothesis
  \begin{itemize}
    \item \semPfDerivDecl{\Gamma_s}{X}{V_s}{V}{c}{p}{v}{T}{V'}
    \item \tyDerivDecls{\Sigma}{\Gamma}{\Delta}{p}{\tau}
    \item \valTyEnvDeriv{\Gamma_s}{V_s}
    \item \valTyEnvDeriv{\Gamma}{V}
  \end{itemize}\
  which gives
  \begin{itemize}
    \item \valTyDeriv{v}{\tau}
    \item \justifies{\Delta, X}{T}
    \item \valTyEnvDeriv{\Gamma_s}{V'}
  \end{itemize}

  this means we have:

  \Rule*
  {
    \overline x \subseteq X \\
    \exists F. (
      \overline{\causesderiv{\Delta}{x_i}{F}}
      \text{ and }
      \effQueueDeriv{F}{U}
    )\\
    \justifies{\Delta, X}{T}
  }
  {\justifies{\Delta, X}{(\semPfTraceOn{\overline x \cap X}{U}, T)}}

  so we have proved the result for this case:
  \begin{itemize}
    \item \valTyDeriv{v}{\tau}
    \item \justifies{\Delta, X}{(\semPfTraceOn{\overline x \cap X}{U}, T)}
    \item \valTyEnvDeriv{\Gamma_s}{V'}
  \end{itemize}

  \item effect rerender, no changes

  \Rule*
    {
      \semPfDerivDecl
        {\Gamma}{X}{V_s}{V}{c}
        {p}
        {v}{T}{V'} \\
      \overline x \cap X = \emptyset
    }
    {
      \semPfDerivDecl
        {\Gamma}{X}{V_s}{V}{c}
        {\synOn{\overline{x}}{e}
          p}
        {v}{T}{V'}
    }

  \Rule*
  {
    \tyRuleOnDecl{\Sigma}{\Gamma}{\Delta}{e}{F}
    \\
    \tyDerivDecls{\Sigma}{\Gamma}{\Delta}{p}{\tau}
  }
  {
    \tyDerivDecls{\Sigma}{\Gamma}{\Delta}{\synOn{\overline x}{e} p}{\tau}
  }

  Apply inductive hypothesis
  \begin{itemize}
    \item \semPfDerivDecl{\Gamma_s}{X}{V_s}{V}{c}{p}{v}{T}{V'}
    \item \tyDerivDecls{\Sigma}{\Gamma}{\Delta}{p}{\tau}
    \item \valTyEnvDeriv{\Gamma_s}{V_s}
    \item \valTyEnvDeriv{\Gamma}{V}
  \end{itemize}\
  which gives
  \begin{itemize}
    \item \valTyDeriv{v}{\tau}
    \item \justifies{\Delta, X}{T}
    \item \valTyEnvDeriv{\Gamma_s}{V'}
  \end{itemize}

  so we have proven the result for this case.

  \item state rerender

  \Rule*
    {
      \semPfDerivDecl{\Gamma_s}{X}{V_s}{V, x=v, \setter{x}=setter_{c.x}}{c}{p}{v}{T}{V'} \\
      x : \tau \in \Gamma_s \\
      x = v \in V_s \\
      x = setter_{c.x} \in V_s \\
      \valTyDeriv{v}{\tau}
    }
    {
      \semPfDerivDecl
        {\Gamma}{X}{V_s}{V}{c}
        {\synState{x}{e}
          p}
        {v}{T}{V'}
    }

  \Rule*
    {
      \tyRuleStateDecl{\Sigma}{\Gamma}{\Delta}{x}{e}{\tau}{F}{p}
      \\
      \tyDerivDecls{\Sigma}{\Gamma}{\Delta}{p}{\tau}
    }
    {
      \tyDerivDecls{\Sigma}{\Gamma}{\Delta}{\synState{x}{e} p}{\tau}
    }

  Apply inversion on \valTyEnvDeriv{\Gamma_s}{V_s} to get \valTyDeriv{v}{\tau}

  Apply inductive hypothesis
  \begin{itemize}
    \item \semPfDerivDecl{\Gamma_s}{X}{V_s}{V, x=v, \setter{x}=setter_{c.x}}{c}{p}{v}{T}{V'}
    \item \tyDerivDecls{\Sigma}{\Gamma}{\Delta}{p}{\tau}
    \item \valTyEnvDeriv{\Gamma_s}{V_s}
    \item \valTyEnvDeriv{\Gamma, x:\tau, \setter{x}: (\tySetter[\tau]{x})}{V, x=v, \setter{x}=setter_{c.x}}
  \end{itemize}\
  which proves the result for this case:
  \begin{itemize}
    \item \valTyDeriv{v}{\tau}
    \item \justifies{\Delta, X}{T}
    \item \valTyEnvDeriv{\Gamma_s}{V'}
  \end{itemize}

  \item return

    \Rule*
      {
        x = v_s \in V_s \\
        x = v \in V \\
        \valTyEnvDeriv{\Gamma_s}{V}
      }
      {
        \semPfDerivDecl{\Gamma}{X}{V_s}{V}{c}{\synReturn{x}}{v}{\semPfTraceReturn{x}}{V}
      }

    \tyRuleReturnDecl{\Sigma}{\Gamma}{\Delta}{x}{\tau}

  Base case.

  \begin{itemize}
    \item \valTyDeriv{v}{\tau}
    \item \Rule*
            {
              \causesderiv{\Delta}{x}{\stch{\kw {return}}}
            }
            {
              \justifies{\Delta, X}{\semPfTraceReturn{x}}
            }
    \item \valTyEnvDeriv{\Gamma_s}{V}
  \end{itemize}

\end{itemize}
\end{proof}

\subsection{Expressions}

\begin{theorem}[Preservation for Expressions]
\label{appendix:thm:preservation-expressions}
Given
\begin{itemize}
  \item \semDerivExpr{V}{e}{v}{U}
  \item \tyDerivExpr{\Gamma}{e}{\tau}{F}
  \item \valTyEnvDeriv{\Gamma}{V}
\end{itemize}

Then
\valTyDeriv{v}{\tau}
and \effQueueDeriv{F}{U}
\end{theorem}

\begin{proof}
By induction on the size of \semDerivExpr{V}{e}{v}{U}.

\begin{itemize}
  \item \semRuleVarDefault{} ~~~ \tyRuleVarDefault{} ~~~
  \Rule*
    {\valTyDeriv{v}{\tau} \and \valTyEnvDeriv{\Gamma}{V}}
    {\valTyEnvDeriv{\Gamma, x:\tau}{V,x =v}}

  \;

  \valTyDeriv{v}{\tau} ~~~ \effQueueDeriv{\cdot}{\cdot}
  \item \semRuleClosureCreate{V}{x}{e} ~~~ \tyRuleFnDefault{} ~~~ \valTyEnvDeriv{\Gamma}{V}

  \Rule*
    {
      \valTyEnvDeriv{\Gamma}{V}
      \and \tyDerivExpr{\Gamma, x: \tau_1}{e}{\tau_2}{F}
    }
    {\valTyDeriv{\Closure{x}{e}{V}}{(\tyArrow{\tau_1}{\tau_2}[F])}}
  ~~~ \Rule*{ }{\effQueueDeriv{\cdot}{\cdot}}

  \;

  \item \semRuleFunctionApplication{
    V=V,
    eone=e_1,
    x=x,
    ethree=e_3,
    Vprime=V',
    Uone=U_1,
    etwo=e_2,
    vtwo=v_2,
    Utwo=U_2,
    vthree=v_3,
    Uthree=U_3}

  \tyRuleFnAppDefault{}

  \valTyEnvDeriv{\Gamma}{V}

  \;

  $\mathit{IH}_1$ ~~~ \valTyDeriv{\Closure{x}{e_3}{V'}}{(\tyArrow{\tau_1}{\tau_2}[F])} ~~~ $\effQueueDeriv{F_1}{U_1}$

  $\mathit{IH}_2$ ~~~ \valTyDeriv{v_2}{\tau_1} ~~~ $\effQueueDeriv{F_2}{U_2}$

  \;

  inversion on $\mathit{IH}_1$ ~~~ \Rule*
    {
      \exists \Gamma'. \valTyEnvDeriv{\Gamma'}{V'}
      \and \tyDerivExpr{\Gamma', x: \tau_1}{e_3}{\tau_2}{F}
    }
    {\valTyDeriv{\Closure{x}{e_3}{V'}}{(\tyArrow{\tau_1}{\tau_2}[F])}}

  \;

  apply IH to
  \semDerivExpr{V',x:v_2}{e_3}{v_3}{U_3} ~~~
  \tyDerivExpr{\Gamma', x: \tau_1}{e_3}{\tau_2}{F} ~~~
  \valTyEnvDeriv{\Gamma'}{V'}

  $\mathit{IH}_3$ ~~~ \valTyDeriv{v_3}{\tau_2} ~~~ $\effQueueDeriv{F_3}{U_3}$

  \;

  \valTyDeriv{v_3}{\tau_2} ~~~
  \Rule*
    {
      \effQueueDeriv{F_1}{U_1} \and
      \Rule*
        {
          \effQueueDeriv{F_2}{U_2}\and
          \effQueueDeriv{F_3}{U_3}
        }
        {\effQueueDeriv{F_2*F_3}{U_2,U_3}}
    }
    {\effQueueDeriv{F_1 * F_2 * F_3}{U_1,U_2,U_3}}

  \breathe{}

  \item \semRuleSetterCallDefault{}

  \Rule*{
    \tyDerivExpr{\Gamma}{e_1}{(\tySetter[\tau]{x})}{F_1}\\
    \tyDerivExpr{\Gamma}{e_2}{(\tyArrow{\tau}{\tau}[\cdot])}{F_2}
  }{
    \tyDerivExpr{\Gamma}{e_1\;e_2}{\tau_2}{F_1 * F_2 * \Next{1r}{\stch{x}}}
  }

  \valTyEnvDeriv{\Gamma}{V}

  \;

  $\mathit{IH}_1$ ~~~
    \valTyDeriv{setter_x}{\tySetter[\tau]{x}} ~~~
    $\effQueueDeriv{F_1}{U_1}$

  \;

  $\mathit{IH}_2$ ~~~
    \Rule*
      {\exists \Gamma'. \valTyEnvDeriv{\Gamma'}{V'} \and \tyDerivExpr{\Gamma', x:\tau}{e_3}{\tau}{\cdot}}
      {\valTyDeriv{\Closure{x}{e_3}{V'}}{\tyArrow{\tau}{\tau}[\cdot]}} ~~~
    $\effQueueDeriv{F_2}{U_2}$

  \;

  \valTyDeriv{()}{\tyUnit}~~~
  \Rule*
    {
      \effQueueDeriv{F_1}{U_1} \and
      \Rule*
        {
          \effQueueDeriv{F_2}{U_2}\and
          \effQueueDeriv{\Next{1r}{\stch{x}}}{setter_x(\Closure{x}{e_3}{V'})}
        }
        {\effQueueDeriv{F_2*\Next{1r}{\stch{x}}}{U_2,setter_x(\Closure{x}{e_3}{V'})}}
    }
    {\effQueueDeriv{F_1 * F_2 * \Next{1r}{\stch{x}}}{U_1,U_2,setter_x(\Closure{x}{e_3}{V'})}}

  \breathe{}

  \item \semRuleSeqDefault{}

  \tyRuleSeqDefault{}

  \valTyEnvDeriv{\Gamma}{V}

  \;

  $\mathit{IH}_1$ ~~~
    \valTyDeriv{v_1}{\tau_1} ~~~
    $\effQueueDeriv{F_1}{U_1}$

  $\mathit{IH}_2$ ~~~
    \valTyDeriv{v_2}{\tau_2} ~~~
    $\effQueueDeriv{F_2}{U_2}$

  \;

  \valTyDeriv{v_2}{\tau_2} ~~~
  \Rule*
    {
      \effQueueDeriv{F_1}{U_1}\and
      \effQueueDeriv{F_2}{U_2}
    }
    {\effQueueDeriv{F_1*F_2}{U_1,U_2}}

  \breathe{}

  \item \Rule{branch1}
    {
      \semDerivExpr{V}{e_1}{true}{U_1} \\
      \semDerivExpr{V}{e_2}{v_2}{U_2} \\
    }
    {
      \semDerivExpr{V}{\synIf{e_1}{e_2}{e_3}}{v_2}{U_1, U_2}
    }
  WLOG choose branch 1

  \tyRuleBranchDefault

  \;

  $\mathit{IH}_1$ ~~~
    \valTyDeriv{true}{\tyBool} ~~~
    $\effQueueDeriv{F_1}{U_1}$

  $\mathit{IH}_2$ ~~~
    \valTyDeriv{v_2}{\tau} ~~~
    $\effQueueDeriv{F_2}{U_2}$

  \;

  \valTyDeriv{v_2}{\tau} ~~~
  \Rule*
    {
      \effQueueDeriv{F_1}{U_1}\and
      \Rule*
        {
          \effQueueDeriv{F_2}{U_2}
        }
        {\effQueueDeriv{F_2 + F_3}{U_2}}
    }
    {\effQueueDeriv{F_1*(F_2 + F_3)}{U_1,U_2}}

  \;

  \item \Rule{prod}{
    \semDerivExpr{V}{e_1}{v_1}{U_1} \and
    \semDerivExpr{V}{e_2}{v_2}{U_2} \and
    }{
    \semDerivExpr{V}{(e_1, e_2)}{(v_1, v_2)}{U_1, U_2} \and
  }
  ~~~ \tyRuleProdDefault{}

  \valTyEnvDeriv{\Gamma}{V}

  \;

  $\mathit{IH}_1$ ~~~
    \valTyDeriv{v_1}{\tau_1} ~~~
    \effQueueDeriv{F_1}{U_1}

  $\mathit{IH}_2$ ~~~
    \valTyDeriv{v_2}{\tau_2} ~~~
    \effQueueDeriv{F_2}{U_2}

  \;

  \Rule*
    {
      \valTyDeriv{v_1}{\tau_1}\and
      \valTyDeriv{v_2}{\tau_2}
    }{
      \valTyDeriv{(v_1,v_2)}{\tau_1 \times \tau_2}
    } ~~~
  \Rule*
    {
      \effQueueDeriv{F_1}{U_1}\and
      \effQueueDeriv{F_2}{U_2}
    }
    {\effQueueDeriv{F_1*F_2}{U_1,U_2}}

  \item \Rule{fst}{
      \semDerivExpr{V}{e}{(v_1, v_2)}{U}
    }{
      \semDerivExpr{V}{\kw{fst}\;e}{v_1}{U}
    }
  ~~~ \tyRuleFstDefault{}
  ~~~ \valTyEnvDeriv{\Gamma}{V}

  WLOG \kw{fst} case

  $\mathit{IH}$ ~~~
    \Rule*
      {{\valTyDeriv{v_1}{\tau_1} \and \valTyDeriv{v_2}{\tau_2}}}
      {\valTyDeriv{(v_1, v_2)}{\tau_1 \times \tau_2}} ~~~
    \effQueueDeriv{F}{U}

  \valTyDeriv{v_1}{\tau_1} ~~~ \effQueueDeriv{F}{U}

  \breathe{}

  \item \Rule{bind}{
      \semDerivExpr{V}{e}{\Closure{x}{e'}{V'}}{U} \\
      \valTyDeriv{\Closure{x}{e'}{V'}}{\tyArrow{\tau}{\tyUnit}[F]}
    }{
      \semDerivExpr{V}{\kw{bind}\;\eventEffBar{\ell}{v_1}\;e}{()}{U,\,
      \mathsf{listen}(\eventEffBar{\ell}{v_1},\,
      \Closure{x}{e'}{V'},\, \always{},\, F)}
    }

  \Rule{ty bind}{
      \tyDerivExpr{\Gamma}{e}{(\tyArrow{\tau}{\tyUnit}[F])}{F_e}\\
      \tau = \Sigma_E(\eventEffBar{\ell}{v_1})
    }{
      \tyDerivExpr
        {\Gamma}
        {\kw{bind}\;\eventEffBar{\ell}{v_1}\;e}
        {\tyUnit}
        {F_e * \always{\eventEffBar{\ell}{v_1}}(F)}
    }

  \valTyEnvDeriv{\Gamma}{V}

  \;

  apply IH to
  \semDerivExpr{V}{e}{\Closure{x}{e'}{V'}}{U} ~~~
  \tyDerivExpr{\Gamma}{e}{(\tyArrow{\tau}{\tyUnit}[F])}{F_e} ~~~
  \valTyEnvDeriv{\Gamma}{V}

  $\mathit{IH}$ ~~~ \valTyDeriv{\Closure{x}{e'}{V'}}{\tyArrow{\tau}{\tyUnit}[F]} ~~~ $\effQueueDeriv{F_e}{U}$

  \;

  \valTyDeriv{()}{\tyUnit} ~~~
  \Rule*
    {
      \effQueueDeriv{F_e}{U} \and
      \Rule*
        {
          \valTyDeriv{\Closure{x}{e'}{V'}}{\tyArrow{\tau}{\tyUnit}[F]} \and
          \tau = \Sigma_E(\eventEffBar{\ell}{v_1})
        }
        {\effQueueDeriv
          {\always{\eventEffBar{\ell}{v_1}}(F)}
          {\mathsf{listen}(\eventEffBar{\ell}{v_1},\, \Closure{x}{e'}{V'},\, \always{},\, F)}}
    }
    {\effQueueDeriv{F_e * \always{\eventEffBar{\ell}{v_1}}(F)}{U,\, \mathsf{listen}(\eventEffBar{\ell}{v_1},\, \Closure{x}{e'}{V'},\, \always{},\, F)}}

  Note that the closure typing $\valTyDeriv{\Closure{x}{e'}{V'}}{\tyArrow{\tau}{\tyUnit}[F]}$
  used to build the \rulename{tq listen always} premise is exactly the value-typing
  conclusion delivered by the IH, so the type-level witness on the new listen
  item matches the latent effect $F$ recorded by the instrumented \rulename{bind}
  rule.

  \breathe{}

  \item \Rule{once}{
      \semDerivExpr{V}{e}{\Closure{x}{e'}{V'}}{U} \\
      \valTyDeriv{\Closure{x}{e'}{V'}}{\tyArrow{\tau}{\tyUnit}[F]}
    }{
      \semDerivExpr{V}{\kw{once}\;\eventEffBar{\ell}{v_1}\;e}{()}{U,\,
      \mathsf{listen}(\eventEffBar{\ell}{v_1},\,
      \Closure{x}{e'}{V'},\, \eventually{},\, F)}
    }

  \Rule{ty once}{
      \tyDerivExpr{\Gamma}{e}{(\tyArrow{\tau}{\tyUnit}[F])}{F_e}\\
      \tau = \Sigma_E(\eventEffBar{\ell}{v_1})
    }{
      \tyDerivExpr
        {\Gamma}
        {\kw{once}\;\eventEffBar{\ell}{v_1}\;e}
        {\tyUnit}
        {F_e * \eventually{\eventEffBar{\ell}{v_1}}(F)}
    }

  \valTyEnvDeriv{\Gamma}{V}

  \;

  Identical to the \kw{bind} case modulo the mode tag $\eventually{}$ in place of
  $\always{}$: apply IH to the closure subderivation to obtain
  \valTyDeriv{\Closure{x}{e'}{V'}}{\tyArrow{\tau}{\tyUnit}[F]} and
  $\effQueueDeriv{F_e}{U}$, then close with \rulename{tq listen once} in place of
  \rulename{tq listen always}.

  \;

  \valTyDeriv{()}{\tyUnit} ~~~
  \Rule*
    {
      \effQueueDeriv{F_e}{U} \and
      \Rule*
        {
          \valTyDeriv{\Closure{x}{e'}{V'}}{\tyArrow{\tau}{\tyUnit}[F]} \and
          \tau = \Sigma_E(\eventEffBar{\ell}{v_1})
        }
        {\effQueueDeriv
          {\eventually{\eventEffBar{\ell}{v_1}}(F)}
          {\mathsf{listen}(\eventEffBar{\ell}{v_1},\, \Closure{x}{e'}{V'},\, \eventually{},\, F)}}
    }
    {\effQueueDeriv{F_e * \eventually{\eventEffBar{\ell}{v_1}}(F)}{U,\, \mathsf{listen}(\eventEffBar{\ell}{v_1},\, \Closure{x}{e'}{V'},\, \eventually{},\, F)}}

  \breathe{}

  \item \Rule{cancel}{
        \\
      }{
        \semDerivExpr{V}{\kw{cancel}\;\eventEffBar{\ell}{v}}{()}{\mathsf{cancel}(\eventEffBar{\ell}{v})}
      }

  \Rule{ty cancel}{
        \\
      }{
        \tyDerivExpr
          {\Gamma}
          {\kw{cancel}\;\eventEffBar{\ell}{v}}
          {\tyUnit}
          {\cancelEv{\eventEffBar{\ell}{v}}}
      }

  \valTyEnvDeriv{\Gamma}{V}

  \;

  No subderivations; no IH to apply. The semantic rule produces $()$ with the
  singleton queue $\mathsf{cancel}(\eventEffBar{\ell}{v})$, and the typing rule
  assigns it the effect $\cancelEv{\eventEffBar{\ell}{v}}$. The matching queue
  typing rule \rulename{tq cancel} relates the two directly.

  \;

  \valTyDeriv{()}{\tyUnit} ~~~
  \Rule*
    { }
    {\effQueueDeriv{\cancelEv{\eventEffBar{\ell}{v}}}{\mathsf{cancel}(\eventEffBar{\ell}{v})}}

  \breathe{}

  \item \Rule{remove}{
        \\
      }{
        \semDerivExpr{V}{\kw{remove}\;\eventEffBar{\ell}{v}}{()}{\mathsf{remove}(\eventEffBar{\ell}{v})}
      }

  \Rule{ty remove}{
        \\
      }{
        \tyDerivExpr
          {\Gamma}
          {\kw{remove}\;\eventEffBar{\ell}{v}}
          {\tyUnit}
          {\remove{\eventEffBar{\ell}{v}}}
      }

  \valTyEnvDeriv{\Gamma}{V}

  \;

  Symmetric to the \kw{cancel} case: the semantic rule produces $()$ with the
  singleton queue $\mathsf{remove}(\eventEffBar{\ell}{v})$, the typing rule
  assigns effect $\remove{\eventEffBar{\ell}{v}}$, and \rulename{tq remove}
  relates them directly.

  \;

  \valTyDeriv{()}{\tyUnit} ~~~
  \Rule*
    { }
    {\effQueueDeriv{\remove{\eventEffBar{\ell}{v}}}{\mathsf{remove}(\eventEffBar{\ell}{v})}}

  \breathe{}

  \item trivial cases:

  \item \semRuleConstantDefault{} ~~~ \tyRuleBoolDefault{} ~~~ \valTyEnvDeriv{\Gamma}{V}

  \item \semRuleConstantDefault{} ~~~ \tyRuleBaseDefault{} ~~~ \valTyEnvDeriv{\Gamma}{V}

  \item \semRuleConstant{V}{()} ~~~
  \Rule*{
    \and
  }{
    \tyDerivExpr{\Gamma}{()}{\tyUnit}{\cdot}
  } ~~~ \valTyEnvDeriv{\Gamma}{V}

\end{itemize}
\end{proof}

\end{document}